\documentclass[superscriptaddress, amsmath,amssymb,
aps, pra, letterpaper, tightenlines, reprint, notitlepages]{revtex4-2}

\usepackage{graphicx}
\usepackage{floatrow}

\usepackage{dcolumn}
\usepackage{bm}

\usepackage{amsfonts}
\usepackage{dsfont}
\usepackage{bm}
\usepackage{bbold}
\usepackage{xcolor}
\usepackage{lipsum,babel}
\usepackage[normalem]{ulem}

\usepackage[utf8]{inputenc}
\usepackage{amsmath}
\usepackage{physics}
\usepackage{hyperref}
\usepackage{amssymb}
\usepackage{braket}
\usepackage{mathtools}
\usepackage{color}
\usepackage{amsthm}
\usepackage{algorithm}
\usepackage{algorithmic}

\usepackage{lipsum}

\usepackage{multirow}

\DeclareMathOperator{\E}{\mathbb{E}}

\DeclareMathOperator{\Prob}{Pr}

\DeclareMathOperator{\Per}{Per}

\DeclareMathOperator{\poly}{poly}

\DeclareMathOperator{\Haar}{Haar}

\DeclareMathOperator{\sort}{\mathsf{sort}}

\renewcommand{\abs}[1]{\left\lvert #1 \right\rvert}

\theoremstyle{plain}
\newtheorem{theorem}{Theorem}
\newtheorem{corollary}{Corollary}
\newtheorem{proposition}{Proposition}

\newtheorem{definition}{Definition}
\newtheorem{lemma}{Lemma}

\newtheorem{remark}{Remark}

\theoremstyle{definition}

\begin{document}

\title{Polynomial-Time Classical Simulation of BosonSampling under \\ Any Constant-Fraction Photon Loss}

\author{Byeongseon Go}
\email{gbs1997@snu.ac.kr}
\affiliation{NextQuantum Center, Department of Physics and Astronomy, Seoul National University, Seoul 08826, Republic of Korea}
\author{Hyunseok Jeong}
\email{h.jeong37@gmail.com}
\affiliation{NextQuantum Center, Department of Physics and Astronomy, Seoul National University, Seoul 08826, Republic of Korea}

\begin{abstract}

Understanding how photon loss changes the classical complexity of BosonSampling is central to assessing the robustness of near-term photonic quantum advantage.
In this work, we present a polynomial-time classical simulation algorithm for lossy BosonSampling that applies whenever a constant fraction of the input photons survives.
We consider the standard fixed-loss model, where $N$ input photons propagate through an $M$-mode interferometer and exactly $N-n$ photons are lost uniformly at random, leaving $n$ surviving photons at the output.
For any constant surviving fraction $n/N\in(0,1)$, the algorithm samples from a distribution within total variation distance $\varepsilon$ of the fixed-loss BosonSampling distribution with probability at least $1-\delta$ over Haar-random interferometers, in time $\operatorname{poly}(M,N,\varepsilon^{-1},\delta^{-1})$.
We further extend the fixed-loss analysis to the uniform-transmission-loss model and obtain analogous polynomial-time simulability for any fixed transmission rate $\eta\in(0,1)$.
These results substantially extend the known provably simulable regime of $n=O(\sqrt{\varepsilon N})$ surviving photons or $\eta=O(\sqrt{\varepsilon/N})$ transmission rate, and demonstrate that even the experimentally relevant linear-survival regime can admit efficient classical simulation under the conditions considered here.

\end{abstract}

\maketitle


\section{Introduction}\label{sec:introduction}

BosonSampling is a restricted linear-optical model in which $N$ photons are injected into an $M$-mode random interferometer and detected using photon-number-resolving detectors~\cite{aaronson2011computational}.  
Strong complexity-theoretic evidence suggests that BosonSampling is classically intractable~\cite{aaronson2011computational, bouland2022noise, bouland2025exponential, bouland2026complexity, hamilton2017gaussian, deshpande2022quantum, grier2022complexity, go2026computational}, which, together with the experimental feasibility of linear-optical systems, has made it a leading candidate for demonstrating quantum advantage.
Since its original proposal, rapid experimental progress has led to several large-scale photonic implementations reporting quantum-advantage demonstrations~\cite{zhong2020quantum, zhong2021phase, madsen2022quantum, deng2023gaussian, liu2026gaussian}.

However, realistic implementations inevitably suffer from noise, which can substantially reduce the classical complexity of the resulting sampling problem~\cite{renema2018efficient,moylett2019classically,shchesnovich2019noise,brod2020classical,oh2025classical,renema2020simulability,van2024efficient,villalonga2021efficient,bulmer2022boundary,oh2024classical,umanskii2024classical,shi2022effect,bulmer2024simulating, aharonov2023polynomial, schuster2025polynomial, noh2020efficient, gao2018efficient}.
It is therefore not sufficient merely to realize large-scale experiments; one must also establish whether their actual noisy output distributions remain outside the reach of efficient classical algorithms.
A comprehensive assessment of quantum advantage consequently requires both stronger hardness results that incorporate realistic imperfections~\cite{aaronson2016bosonsampling, go2025quantum, go2025sufficient, go2026hardness} and sharper classical algorithms that exploit such imperfections to probe the limits of the claimed advantage.

Photon loss is an unavoidable imperfection in BosonSampling implementations and cannot be intrinsically corrected within its restricted linear-optical framework.
Previous results established polynomial-time classical simulation when the surviving fraction vanishes as the input photon number increases~\cite{oszmaniec2018classical, garcia2019simulating, qi2020regimes, oh2021classical, liu2023simulating, park2025matrix, lim2026classical}.
A central question is whether such guarantees extend to a fixed nonzero surviving fraction, where the number of surviving photons grows \textit{linearly} with the input photon number.
Since an extensive number of photons still survives, polynomial-time simulability in this regime is not immediate. 
At the same time, this regime arises naturally under a fixed overall transmission probability and is directly relevant to current photonic experiments~\cite{zhong2020quantum, zhong2021phase, madsen2022quantum, deng2023gaussian, liu2026gaussian}.
Characterizing its computational complexity is therefore important both for understanding how photon loss reshapes the computational complexity of BosonSampling and for assessing the robustness of photonic quantum advantage.




\begin{figure*}[t]
\includegraphics[width=0.85\linewidth]{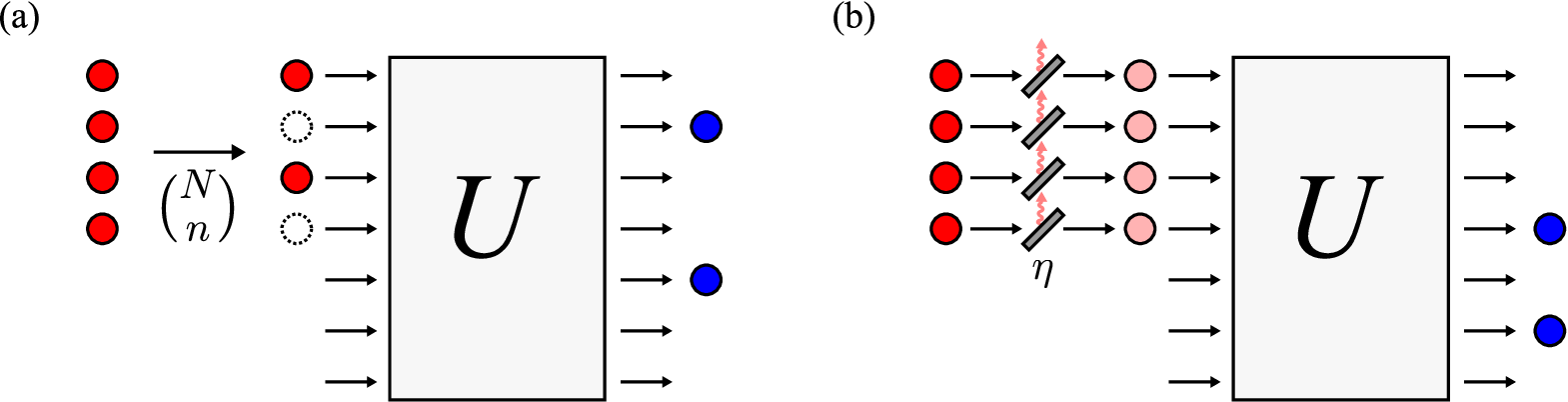}
\caption{Illustration of two representative photon-loss models for BosonSampling.
(a) Fixed-loss model, shown for $N=4$ input photons and $n=2$ surviving photons.
In this model, exactly $n$ out of the $N$ input photons survive, with the surviving subset chosen uniformly at random.
The resulting output distribution is therefore the uniform average over the $\binom{N}{n}$ BosonSampling output distributions corresponding to all possible $n$-photon input subsets of the $N$ input modes.
Our main simulation algorithm in Theorem~\ref{thm:main-informal} is developed for this setting.
(b) Uniform-transmission-loss model, shown for $N=4$ input photons.
Here each of the $N$ input photons is transmitted independently with probability $\eta$, so the surviving photon number follows a binomial distribution.
For the typical surviving-photon number $n\simeq \eta N$, the two loss models are reducible to each other (see, e.g., Ref.~\cite{oszmaniec2018classical} or Sec.~\ref{sec:extensions}). 
Accordingly, our fixed-loss simulation framework extends directly to the uniform-transmission-loss model for the regime $\eta \simeq n/N$.
}
\label{fig: schematics}
\end{figure*}


In this work, we present a polynomial-time classical simulation algorithm for lossy BosonSampling that applies throughout the linear-survival regime.
Our construction focuses on the standard fixed-loss model, in which $N$ input photons propagate through an $M$-mode Haar-random interferometer, exactly $N-n$ photons are lost uniformly at random, and the remaining $n$ photons are measured at the output~\cite{aaronson2016bosonsampling, oszmaniec2018classical, wang2018toward, brod2020classical, oh2023classical} (see, e.g., Fig.~\ref{fig: schematics}(a)).
We initially target the collision-collapsed distribution in Ref.~\cite{oh2023classical}, which retains all collision-free outcomes while coarse-graining all collision outcomes.
This choice is natural under the dilute-mode condition considered here, where collision-free outcomes carry the dominant probability weight.
For any constant surviving fraction $n/N\in(0,1)$, our algorithm samples from a distribution within total variation distance $\varepsilon$ of the target distribution in time $\operatorname{poly}(M,N,\varepsilon^{-1},\delta^{-1})$, with probability at least $1-\delta$ over Haar-random interferometers.

Moreover, while the above result is formulated for the collision-collapsed output distribution, imposing an additional bosonic-birthday-bound condition on the mode number extends the guarantee to the full Fock-output distribution, with the collision events absorbed into the algorithmic error.
The result further extends to the uniform-transmission-loss model depicted in Fig.~\ref{fig: schematics}(b), another standard model of photon loss, for which the resulting simulation again runs in time $\operatorname{poly}(M,N,\varepsilon^{-1},\delta^{-1})$ for any fixed transmission probability $\eta \in (0,1)$.

These results therefore demonstrate that efficient classical simulation can persist even in the linear-survival regime. 
They substantially extend the known regime of provable polynomial-time classical simulability beyond $n=O(\sqrt{\varepsilon N})$ surviving photons or  $\eta=O(\sqrt{\varepsilon/N})$ transmission rate~\cite{oszmaniec2018classical, garcia2019simulating} under the dilute-mode condition.
This extension comes with tunable approximation accuracy, with the runtime remaining polynomial even when the target approximation error decreases inverse-polynomially with system size.
Consequently, in the settings considered here, increasing the input photon number while maintaining a fixed transmission rate does not asymptotically drive the problem outside this classically simulable regime.
This identifies a concrete limitation of fixed-transmission scaling alone as a strategy for achieving scalable photonic quantum advantage.

The remainder of the paper is organized as follows.
Section~\ref{sec:related} summarizes our main results and contributions to the literature.
Section~\ref{sec:preliminaries} defines fixed-loss BosonSampling and introduces the first-quantized representation underlying our sampling algorithm.
Section~\ref{sec:main-result} states the main theorem and outlines its proof.
Section~\ref{sec:hoeffding} derives the decomposition of the output probabilities and establishes the truncation-error bound under the Gaussian measure.
Section~\ref{sec:haar-error} transfers this bound to Haar-random submatrices and derives the corresponding error bound over the output space.
Section~\ref{sec:marginals} develops an efficient procedure for evaluating the marginals required by the sampler.
Section~\ref{sec:sampling} combines these ingredients with the marginal-sampling procedure to prove the main theorem.
Section~\ref{sec:full-fock} extends the analysis to the full Fock-output distribution and establishes the corresponding approximation guarantee.
Section~\ref{sec:extensions} generalizes the fixed-loss analysis to the uniform-transmission-loss model.
Finally, Sec.~\ref{sec:discussion} concludes with a discussion of the implications of our results and open directions.

\section{Overview of contributions}\label{sec:related}

In this section, we summarize our main simulability consequences for lossy BosonSampling and compare them with closely related results. 





\subsection{Our simulation algorithm and its consequences}\label{subsec: our result}

We mainly consider fixed-loss BosonSampling with $N$ input photons sent through an $M$-mode Haar-random interferometer.
As illustrated in Fig.~\ref{fig: schematics}(a), exactly $N-n$ photons are lost uniformly at random, so that $n$ photons survive and are subsequently measured at the output.
We first focus on the collision-collapsed distribution~\cite{oh2023classical}, which retains each collision-free Fock outcome separately while aggregating the entire collision sector into a single outcome.
This coarse-graining is natural in the dilute-mode regime of primary interest, where the number of modes $M$ is parametrically larger than the number of surviving photons $n$ and collision events are therefore rare~\cite{aaronson2011computational, arkhipov2012bosonic}.

Our main theorem can then be summarized as follows.

\begin{theorem}[Informal]\label{thm:main-informal}
Let $1\leq n<N$, let $\varepsilon,\delta\in(0,1)$, and suppose $M=\Omega(n^2N)$.
Choose $d\in\{0,\ldots,n-1\}$ such that
\begin{align}
 \frac{(n/N)^{d+1}}{1-n/N}
 &\leq
 O(\varepsilon^2\delta).
 \label{eq:intro-informal-degree-condition}
\end{align}
Then there exists a randomized classical algorithm that, for at least a $1-\delta$ fraction of Haar-random interferometers, samples from a distribution within total variation distance $\varepsilon$ of the collision-collapsed fixed-loss BosonSampling distribution, with runtime $MnN^3\varepsilon^{-2}\delta^{-1}2^{O(d)}$.
\end{theorem}

The theorem is parametrized by the truncation degree $d$, which controls the approximation error and the runtime, providing a tunable tradeoff between them.

The main consequence of Theorem~\ref{thm:main-informal} is polynomial-time classical simulability for every fixed surviving fraction $n/N\in(0,1)$, thus covering the full linear-survival regime.

\begin{corollary}[Polynomial-time simulability in the linear-survival regime]\label{cor:fixed-loss-regimes}
For every fixed surviving fraction $n/N\in(0,1)$, the algorithm of Theorem~\ref{thm:main-informal} runs in time $\operatorname{poly}(M,N,\varepsilon^{-1},\delta^{-1})$.
\end{corollary}

As noted, the target distribution in Theorem~\ref{thm:main-informal} is the collision-collapsed formulation, which merges all collision outcomes rather than resolving individual collision patterns.
Nevertheless, the bosonic birthday bound~\cite{aaronson2011computational, arkhipov2012bosonic, qi2020regimes, go2026computational} guarantees that, under the dilute-mode condition $M=\Omega(n^2N)$ in Theorem~\ref{thm:main-informal}, the total collision probability is sufficiently suppressed so that collision-free outcomes carry the dominant probability weight (see, e.g., Lemma~\ref{lem:bosonic-birthday-bound}).
Furthermore, imposing an additional mode-number condition extends the collision-collapsed approximation guarantee to the full Fock-output distribution, with the total collision probability incorporated as an additional contribution to the algorithmic error.

\begin{corollary}[Full Fock-output approximation]\label{cor:full-fock-summary}
Under the additional mode-number condition $M=\Omega(n^2/(\varepsilon\delta))$, the collision-collapsed sampler of Theorem~\ref{thm:main-informal} yields a full fixed-loss Fock-output approximation guarantee, while preserving asymptotically the same runtime, accuracy, and success-probability guarantees.
\end{corollary}

Lastly, although Theorem~\ref{thm:main-informal} is formulated for the fixed-loss setting, the same result extends directly to the conventional uniform-transmission-loss model in Fig.~\ref{fig: schematics}(b).
The extension follows by first sampling the number of surviving photons from the corresponding binomial distribution and then applying the fixed-loss algorithm conditioned on the sampled photon-number sector.

\begin{corollary}[Uniform-transmission loss]\label{cor:uniform-transmission-summary}
For the uniform-transmission-loss model, where each of the $N$ photons is independently transmitted with probability $\eta$, Theorem~\ref{thm:main-informal} implies polynomial-time classical simulability for every fixed transmission probability $\eta\in(0,1)$ under the stated dilute-mode condition. 
\end{corollary}

The precise statements and assumptions are given later in Theorem~\ref{thm:main}, Sec.~\ref{sec:full-fock}, and Sec.~\ref{sec:extensions}.

\subsection{Simulation approaches in the linear-survival regime}

The classical simulation of lossy BosonSampling with a constant surviving fraction has been considered in previous work~\cite{renema2018classical, moylett2019classically}.
The remaining gap was a \textit{provably efficient} sampling algorithm that simultaneously achieves a polynomial runtime up to the linear-survival regime and an explicit total-variation-distance guarantee for the final output distribution.
Reference~\cite{renema2018classical} showed that higher-order many-photon interference can be truncated with a small approximation error, and combined this approximation with a Metropolis sampler.
This construction falls short of providing either a polynomial-time procedure for estimating output probability to inverse-polynomial accuracy or a direct sampler with an explicit bound on the total runtime required to achieve a prescribed total variation distance.
Reference~\cite{moylett2019classically} instead developed a direct sampler with an explicit approximation-error bound and runtime analysis, but its total runtime remains exponential in the linear-survival regime.

We fill this gap by developing ($i$) a new low-degree truncation method based on a Hoeffding decomposition of the probabilities~\cite{hoeffding1992class, rempala2002central, le2025hoeffding}, and ($ii$) a polynomial-time procedure for evaluating the prefix marginals of the resulting truncated approximation, including for inverse-polynomial error parameters, based on randomized color coding~\cite{alon1995color, alon2008biomolecular}.
The latter provides the key ingredient needed to invoke the marginal-sampling framework~\cite{bremner2017achieving, aharonov2023polynomial, oh2023classical, takahashi2021classically, cirstoiu2024fourier, schuster2025polynomial}.
Together, these ingredients yield a direct classical sampler with explicit guarantees on both the final total variation distance and total runtime, which remains polynomial throughout the linear-survival regime even for inverse-polynomial total variation distance.

Separately, polynomial-time simulation is known for constant-depth interferometers below an architecture-dependent transmission threshold~\cite{oh2025classical}; this shallow-depth setting is distinct from the dense-interferometer setting considered here.

\subsection{Provably efficient samplers in stronger-loss regimes}

We next turn to provably efficient algorithms that apply under stronger photon loss.
For the fixed-loss model, Ref.~\cite{oszmaniec2018classical} derived a polynomial-time algorithm for $n=O(\sqrt{\varepsilon N})$, where $\varepsilon$ is the target total variation distance.
For the uniform-transmission-loss model, Ref.~\cite{garcia2019simulating} obtained analogous polynomial-time guarantees, achieving total variation distance $\varepsilon$ whenever the transmission rate satisfies $\eta=O(\sqrt{\varepsilon/N})$.
These results provide guarantees for arbitrary interferometers and
require no dilute-mode assumption.



Our algorithm reaches the linear-survival regime for any fixed surviving fraction $n/N\in(0,1)$, with runtime $\operatorname{poly}(M,N,\varepsilon^{-1},\delta^{-1})$.
The uniform-transmission-loss extension likewise has polynomial runtime for every fixed transmission rate $\eta\in(0,1)$.
Importantly, the admissible survivor-number scaling, and likewise the transmission-rate scaling, is now independent of the target approximation error, while the cost of increased accuracy is absorbed into the computational runtime.
These guarantees are probabilistic over Haar-random interferometers and are established under the dilute-mode condition.


In summary, the previous sampling algorithms~\cite{oszmaniec2018classical, garcia2019simulating} apply uniformly to arbitrary interferometers and impose no dilute-mode condition, while their polynomial-time simulability is guaranteed for $O(\sqrt{\varepsilon N})$ surviving photons.
By contrast, our algorithm reaches substantially larger surviving-photon regimes, with survivor-number scaling independent of the target approximation error, while its guarantee is probabilistic over Haar-random interferometers and requires the dilute-mode condition.


\subsection{The low-degree sampling framework}

A viable route to classical simulation of noisy sampling problems is the low-degree sampling framework, which combines a low-degree approximation of the output distribution with marginal evaluation procedure to construct a proper classical sampler.
For noisy random circuit sampling, Ref.~\cite{aharonov2023polynomial} derived a low-degree approximation of the output distribution by truncating high-order Fourier terms, then converted it into a sampler by developing a direct marginal evaluation procedure and combining it with the marginal sampling method of Ref.~\cite{bremner2017achieving}.
Reference~\cite{oh2023classical} subsequently adapted this framework to BosonSampling with Gaussian circuit noise, where the relevant low-degree structure arises from an expansion in a scalar noise parameter, deriving both a low-degree approximation of the noisy output distribution and a direct marginal evaluation procedure.

Extending this framework to obtain a polynomial-time simulation algorithm for photon loss faces two technical barriers.
First, the output distribution is now described by an averaged permanent statistic~\cite{aaronson2016bosonsampling} (see Eq.~\eqref{eq:fock-fixed-loss}), rather than by a polynomial dependence on a scalar noise parameter.
Second, even once an appropriate low-degree truncation is identified, the required truncation degree $d$ to achieve a prescribed approximation error grows logarithmically in the linear-survival regime $n/N \in (0,1)$.
Consequently, even a direct marginal-evaluation procedure of the type developed in Ref.~\cite{oh2023classical}, with complexity $N^{O(d)}$, would yield quasipolynomial rather than polynomial runtime.
This constitutes the principal computational bottleneck to obtaining a polynomial-time sampler in the linear-survival regime.

Our main technical contribution is to overcome both obstacles.
We first derive a Hoeffding-type decomposition of the fixed-loss BosonSampling output distribution~\cite{hoeffding1992class, rempala2002central, le2025hoeffding}, revealing a low-degree hierarchy in which higher-order components are systematically suppressed.
We then develop a \textit{randomized} marginal evaluation procedure based on color coding~\cite{alon1995color, alon2008biomolecular}, reducing the dependence on the truncation degree from $N^{O(d)}$ to $2^{O(d)}$.
The additional randomness introduces controlled contributions to the total-variation-distance error and failure probability, both of which are incorporated into the overall algorithmic error budget.
Combined with the signed-marginal method of Ref.~\cite{bremner2017achieving}, these ingredients yield a polynomial-time classical sampling algorithm for fixed-loss BosonSampling with polynomial dependence on the system size and inverse error parameters.

\section{Fixed-loss BosonSampling in first quantization}\label{sec:preliminaries}

This section introduces the basic notation and clarifies the sample spaces used throughout the analysis.

We distinguish three representations of the output space. 
The first is the physical Fock-output representation, where an outcome is written as a \textit{sorted} string of occupied modes.
The second is the collision-collapsed Fock-output space used in the main theorem, where all collision outcomes are merged into a single dummy symbol, following Ref.~\cite{oh2023classical}. 
The third is the \textit{ordered} first-quantized output string representation used by the sampler, which enables us to express the prefix marginals since it does not have the restriction of sorting. 
The first-quantized representation is auxiliary: pushing it forward by the sorting map recovers the physical Fock-output distribution.

We write $[m]:=\{1,\ldots,m\}$ for a positive integer $m$. 
Let $M$ be the number of optical modes, and let $N$ be the number of input photons. 
The input photons are injected into the first $N$ input modes of an $M$-mode interferometer, characterized by a unitary matrix $U\in\mathrm U(M)$; we use the terms ``interferometer'' and ``unitary matrix'' interchangeably when referring to $U$.
Throughout, $U$ is Haar-random, and we denote by $\Haar(M)$ the Haar measure on $\mathrm U(M)$.
Let $n\leq N$ be the number of surviving photons, so that exactly $N-n$ photons are lost.
The set of surviving photons is chosen uniformly among all $\binom{N}{n}$ subsets of the $N$ input photons~\cite{aaronson2016bosonsampling, oszmaniec2018classical}.



\subsection{Fock-output distribution}

We denote the set of $n$-photon Fock outcomes by
\begin{align}
 \mathcal T_{M,n}
 &:={}
 \left\{
 \bm t = (t_1,\ldots, t_n) \in [M]^n:
 t_1 \leq t_2 \leq \ldots\leq t_n
\right\},
 \label{eq:fock-outcome-set}
\end{align}
Hence, $\bm t$ denotes a sorted list of mode labels occupied by the photons, with repeated entries indicating collisions.
For $\bm t\in\mathcal T_{M,n}$ and $S\subseteq[N]$ with $\abs{S}=n$, let $U_{\bm t,S}$ denote the $n\times n$ matrix obtained from $U$ by selecting the columns in $S$ and the rows in $\bm t$. 
If a mode label appears multiple times in $\bm t$, the corresponding row of $U$ is repeated with the same multiplicity.
For $k\in[M]$, define the occupation multiplicity $s_k(\bm t):=\abs{\{ j \in[n]:t_j = k\}}$, and set $\bm t!:=\prod_{k=1}^{M}s_k(\bm t)!$.

The physical fixed-loss Fock-output distribution is then written by~\cite{aaronson2016bosonsampling}
\begin{align}
 P_U^{N,n}(\bm t)
 &=
 \frac{1}{\binom{N}{n}}
 \sum_{\substack{S \subseteq [N] \\ |S| = n} }
 \frac{\abs{\Per(U_{\bm t,S})}^2}{\bm t!}  ,
 \label{eq:fock-fixed-loss}
\end{align}
which averages the squared permanents over all possible surviving input subsets.

\subsection{Collision-collapsed Fock space}

For the collision-collapsed formulation used in the main theorem, we denote the set of possible collision-free Fock outcomes by
\begin{align}
 \mathcal T_{M,n}^{\mathrm{cf}}
 &:={}
 \left\{
 \bm t \in \mathcal T_{M,n}:
 t_1 < t_2 < \ldots < t_n
\right\}.
 \label{eq:collision-free-fock-set}
\end{align}
Following the approach of Ref.~\cite{oh2023classical}, we adjoin a dummy collision symbol $c\notin\mathcal T_{M,n}^{\mathrm{cf}}$ and define the collision-collapsing map
\begin{align}\label{eq:collision-collapse-definition}
 \mathsf C:\mathcal T_{M,n}
 &\longrightarrow
 \mathcal T_{M,n}^{\mathrm{cf}}\cup\{ c\}
\end{align}
as follows:
\begin{align}
 \mathsf C(\bm t)
 &:={}
 \begin{cases}
 \bm t,&\bm t\in\mathcal T_{M,n}^{\mathrm{cf}},\\
 c,&\bm t\notin\mathcal T_{M,n}^{\mathrm{cf}}.
 \end{cases}
 \label{eq:fock-collision-collapse-map}
\end{align}
We then define the collision-collapsed fixed-loss Fock distribution, the target distribution in our main result, by
\begin{align}
 P_{U,\mathrm{coll}}^{N,n}
 &:={}
 (\mathsf C)_{\#}P_U^{N,n}.
\label{eq:collapsed-fock-target}
\end{align}
Hence, collision-free Fock outcomes are retained, while all collision outcomes are identified with the single symbol $c$.

\subsection{First quantization}

Our sampling algorithm is \textit{sequential}, in the sense that it samples one output-mode label at a time using prefix marginals.  
Therefore, it is convenient to lift the Fock-output representation $\bm t\in\mathcal T_{M,n}$ to an ordered list of mode labels, one for each photon.
Concretely, we write
\begin{align}
\bm r
=
(r_1,\ldots,r_n)
\in[M]^n,
\label{eq:ordered-string}
\end{align}
and refer to $\bm r$ as a first-quantized output string.
This terminology is analogous to first-quantized $n$-particle notation, in which a many-particle configuration is represented by assigning one coordinate (mode label) to each particle.
The ordering is purely auxiliary and does not assign physical labels to the photons; sorting removes this bookkeeping order and recovers the bosonic Fock outcome.

We also define the set of collision-free first-quantized output strings as
\begin{align}
[M]_{\neq}^{n} 
&:={}
\left\{
\bm r=(r_1,\ldots,r_n)\in[M]^n:
r_a\neq r_b\text{ if }a\neq b
\right\}.
\label{eq:collision-free-strings}
\end{align}

To relate $\bm r\in[M]^n$ to the Fock outcome $\bm t\in\mathcal T_{M,n}$ in Eq.~\eqref{eq:fock-outcome-set}, we define the sorting map
\begin{align}\label{eq:sort-definition}
    \sort:[M]^n \rightarrow \mathcal T_{M,n}.
\end{align}
Hence, $\sort(\bm r)$ is obtained by rearranging its entries in nondecreasing order.
One can easily check that $\sort$ maps the collision-free string set $[M]_{\neq}^{n} $ in Eq.~\eqref{eq:collision-free-strings} onto the collision-free Fock outcome set $\mathcal T_{M,n}^{\mathrm{cf}}$ in Eq.~\eqref{eq:collision-free-fock-set}.

For a first-quantized output string $\bm r \in [M]^n$, we define the scaled output submatrix
\begin{align}
A_{\bm r}
&:={}
\sqrt M\,U_{\bm r,[N]}
\in\mathbb C^{n\times N},
\label{eq:scaled-output-matrix}
\end{align}
such that $(A_{\bm r})_{i,j}=\sqrt M\,U_{r_i,j}$.

The central object underlying the fixed-loss model is the average over surviving input subsets of the squared permanent contributions appearing in Eq.~\eqref{eq:fock-fixed-loss}.
Specifically for any $B\in\mathbb C^{n\times N}$, define the averaged permanent statistic
\begin{align}
\Phi(B)
&:={}
\frac{1}{\binom{N}{n}}
\sum_{\substack{S \subseteq [N] \\ |S| = n}}
\abs{\Per(B_S)}^2,
\label{eq:phi-definition}
\end{align}
where $B_S$ is a submatrix of $B$ obtained by taking columns in $S$. 
Define the ordered-string mass function $p_U$ by evaluating $\Phi$ on $A_{\bm r}$ and applying the normalization factor:
\begin{align}
p_U(\bm r)
&:={}
\frac{\Phi(A_{\bm r})}{n!M^n}.
\label{eq:ordered-probability}
\end{align}
Here, the factor $M^n$ comes from the scaling $A_{\bm r}=\sqrt M\,U_{\bm r,[N]}$. Also, the factor $n!$ spreads each physical Fock-output mass over its ordered-string representatives in $[M]^n$.
Fix any $\bm r$ with $\sort(\bm r)=\bm t$; the squared permanent is the same for every such ordering.
Specifically, given that the number of distinct orderings of the multiset $\bm t$ is $n! / \bm t!$, the total mass assigned by $p_U$ to all first-quantized output strings with the same Fock outcome $\bm t \in \mathcal T_{M,n}$ is
\begin{align}
\sum_{\bm r:\, \sort(\bm r)= \bm t}p_U(\bm r)
&=
\frac{n!}{\bm t!}
\frac{1}{n!M^n}
\frac{1}{\binom{N}{n}}
\sum_{\substack{S \subseteq [N] \\ |S| = n}}
\abs{\Per(A_{\bm r,S})}^2
\nonumber\\
&=
\frac{n!}{\bm t!}
\frac{1}{n!}
\frac{1}{\binom{N}{n}}
\sum_{\substack{S \subseteq [N] \\ |S| = n}}
\abs{\Per(U_{\bm t,S})}^2
\nonumber\\
&=
P_U^{N,n}(\bm t), 
\label{eq:ordered-pushforward}
\end{align}
where we used $\Per(U_{\bm r,S}) = \Per(U_{\bm t,S})$ for $\sort(\bm r)= \bm t$.  
Therefore, pushing forward $p_U$ under the sorting map $\sort$ recovers exactly the physical fixed-loss distribution $P_U^{N,n}$.
Since $p_U$ is nonnegative and its pushforward has total mass one, $p_U$ is a probability distribution on $[M]^n$.


Throughout the paper, for probability distributions $p$ and $q$ on a common finite space, the total variation distance is defined as 
\begin{align}
 d_{\mathrm{TV}}(p,q)
 &:={}
 \frac12\sum_x\abs{p(x)-q(x)}
 =
 \frac12\norm{p-q}_1.
 \label{eq:tv-definition}
\end{align}
Since total variation distance cannot increase under deterministic postprocessing, pushing forward by the sorting map in Eq.~\eqref{eq:sort-definition} gives, for any probability distribution $q$ on $[M]^n$,
\begin{align}
 d_{\mathrm{TV}}({\sort}_\#q,P_U^{N,n})
 &\leq
 d_{\mathrm{TV}}(q,p_U).
 \label{eq:pushforward-contraction}
\end{align}
Therefore, by sorting the output string, any sampling algorithm with bounded total variation distance from the first-quantized distribution $p_U$ automatically yields a sampler within the same total variation distance from the physical Fock-output distribution $P_U^{N,n}$.


\section{Main result and proof overview}\label{sec:main-result}

This section presents our main result for the classical simulation of lossy BosonSampling.
We first state the main theorem and then derive its immediate consequences.
We then provide an overview of the proof, emphasizing the crucial ingredients that lead to the final sampling algorithm.


\subsection{Main result and its implications}

We begin by presenting the main algorithmic result of this work.
Note that the target distribution is the collision-collapsed fixed-loss distribution $P_{U,\mathrm{coll}}^{N,n}$ defined in Eq.~\eqref{eq:collapsed-fock-target}.

\begin{theorem}[Classical algorithm for fixed-loss BosonSampling]\label{thm:main}
There exist constants $C_0, C > 0$ such that the following holds. 
Let $1 \leq n < N$, let $\varepsilon,\delta \in (0,1)$, and suppose $M\geq C_0n^2N$.
Choose any truncation degree $d\in\{0,\ldots,n-1\}$ satisfying
\begin{align}
 \frac{(n/N)^{d+1}}{1-n/N}
 &\leq
 \frac{\varepsilon^2\delta}{C}
 \label{eq:main-degree-condition}
\end{align}
if such a value exists; otherwise, set $d=n$.
Then there exists a randomized classical algorithm that runs in $MnN^3\varepsilon^{-2}\delta^{-1}2^{O(d)}$ arithmetic operations and samples from a distribution $\mathcal A_d(U)$ satisfying
\begin{align}
 \Prob_{U\sim\Haar(M)}\!\left[
 d_{\mathrm{TV}}\!\left(\mathcal A_d(U),P_{U,\mathrm{coll}}^{N,n}\right)
 \leq\varepsilon
\right]
 &\geq
 1-\delta.
 \label{eq:main-average-case}
\end{align}
\end{theorem}


While the target distribution in Theorem~\ref{thm:main} is collision-collapsed, a full Fock-output approximation guarantee is stated later in Sec.~\ref{sec:full-fock} under an additional mode-number lower bound.

To satisfy Eq.~\eqref{eq:main-degree-condition}, it suffices to choose the truncation degree as
\begin{align}
d
=
O\left(1 + \frac{\log\!\left(\varepsilon^{-1}\delta^{-1}(1-n/N)^{-1}\right)}{\log(N/n)}\right).
\label{eq:degree-choice}
\end{align}
The key consequence of Eq.~\eqref{eq:degree-choice} appears in the linear-survival regime, characterized by a fixed surviving fraction $n/N\in(0,1)$.
Since both $\log(N/n)$ and $(1-n/N)^{-1}$ are then constants, the required degree scales only logarithmically with the inverse error parameters,
$d=O\!\left(1+\log(\varepsilon^{-1}\delta^{-1})\right)$.
Substituting this choice into the runtime bound of Theorem~\ref{thm:main} yields a runtime polynomial in $M,N,\varepsilon^{-1}$, and $\delta^{-1}$, with a polynomial degree that may depend on the fixed surviving fraction $n/N \in (0,1)$.
Thus, polynomial-time classical simulation persists even for inverse-polynomial approximation error and failure probability throughout the linear-survival regime.

\begin{figure*}[t]
\includegraphics[width=0.9\linewidth]{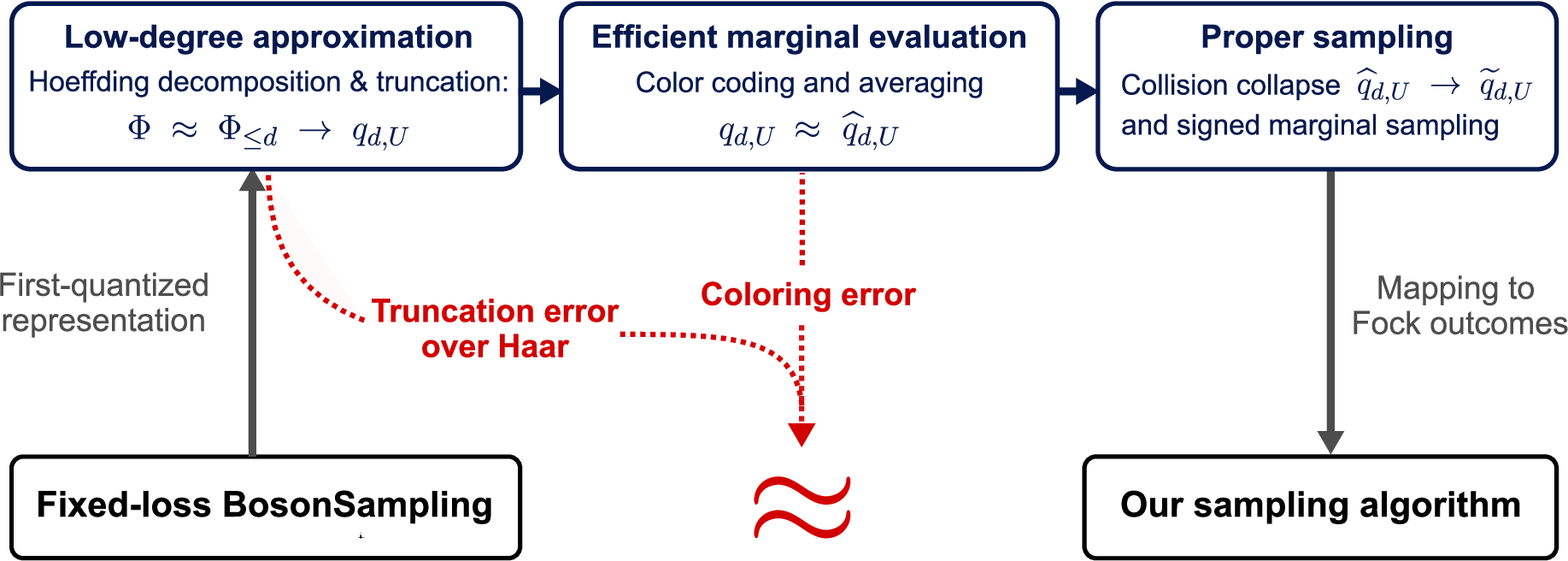}
\caption{
Schematic overview of the proof of Theorem~\ref{thm:main}.
Along the solid path, Hoeffding decomposition and truncation of $\Phi$ yield the signed approximation $q_{d,U}$.
Averaging independent color-coded functions then gives $\widehat q_{d,U}$ with $\E_{\mathrm{col}}[\widehat q_{d,U}]=q_{d,U}$.
For a fixed collection of colorings, the prefix marginals of $\widehat q_{d,U}$ can be evaluated exactly and efficiently.
Collision collapse yields the normalized signed function $\widetilde q_{d,U}$, to which we apply the signed marginal sampling lemma of Ref.~\cite{bremner2017achieving}.
Mapping the resulting ordered strings to Fock outcomes gives the output distribution $\mathcal A_d(U)$, which includes the color-coding randomness.
The red dotted paths show how the Haar truncation bound and the additional coloring error combine to give the final error guarantee in Eq.~\eqref{eq:main-average-case}.
The large $\approx$ represents this guarantee between $\mathcal A_d(U)$ and the collision-collapsed target $P_{U,\mathrm{coll}}^{N,n}$.
}
\label{fig: outline}
\end{figure*}

\subsection{Proof overview of Theorem~\ref{thm:main}}

We now provide an overview of the proof of Theorem~\ref{thm:main}, following the construction illustrated in Fig.~\ref{fig: outline}.
The proof first establishes a low-degree approximation of the output distribution and then develops the marginal-evaluation procedure needed to convert this approximation into an efficient sampler.

The proof starts with the averaged permanent statistic $\Phi$ in Eq.~\eqref{eq:phi-definition}.
In Sec.~\ref{sec:hoeffding}, we derive a Hoeffding-type decomposition of $\Phi$ into components indexed by subsets of input columns~\cite{hoeffding1992class, rempala2002central, le2025hoeffding}.
Each component depends only on a subset of $j$ input columns and appears with coefficient $\binom nj/\binom Nj$, the probability that those columns belong to the uniformly chosen surviving input subset.
Truncating the decomposition at order $d$ defines $\Phi_{\leq d}$.
Under the Gaussian measure, distinct components are orthogonal, giving the mean-square truncation error, normalized by $(n!)^2$,
\begin{align}
R_d(n,N)
&=\sum_{j=d+1}^{n}\frac{\binom nj}{\binom Nj}
\leq\frac{(n/N)^{d+1}}{1-n/N}.
\label{eq:R_d-proof-outline}
\end{align}

In Sec.~\ref{sec:haar-error}, we evaluate $\Phi_{\leq d}$ on the scaled Haar submatrices and divide by $n!M^n$ to obtain the signed approximation $q_{d,U}$ to the first-quantized distribution $p_U$ (Definition~\ref{def:signed-ordered-approximation}).
For collision-free strings, the selected rows are distinct, and under $M\geq C_0n^2N$ as stated in Theorem~\ref{thm:main}, we can apply Haar--Gaussian density domination in Lemma~\ref{lem:rectangular-domination} to bound the mean-square pointwise error by $C_HR_d(n,N)/M^{2n}$ for a constant $C_H$.
Applying Cauchy--Schwarz and summing over at most $M^n$ such strings then gives a Haar-averaged squared collision-free $L^1$ error of at most $C_HR_d(n,N)$.

For a fixed surviving fraction $n/N\in(0,1)$, $R_d(n,N)$ decreases geometrically with the truncation degree.
This approximation bound alone, however, does not directly provide an efficient sampling algorithm.
The sequential sampling procedure requires evaluating the prefix marginals of $q_{d,U}$, obtained by summing over all collision-free completions of a given prefix~\cite{bremner2017achieving}.
Even after truncation, the approximation still contains contributions from all input-column subsets of size at most $d$, leaving $N^{O(d)}$ such subsets to enumerate in a direct evaluation (see, e.g., the marginal evaluation in Ref.~\cite{oh2023classical}).
In the linear-survival regime with constant $n/N\in(0,1)$, the truncation degree $d$ needed to achieve the prescribed accuracy and failure probability scales logarithmically with the inverse error parameters, as indicated by Eq.~\eqref{eq:R_d-proof-outline}.
Consequently, improving the approximation guarantee also increases the exponent of $N$ in the direct evaluation cost, which constitutes the remaining computational bottleneck.

To overcome this bottleneck, in Sec.~\ref{sec:marginals}, we develop a randomized marginal-evaluation procedure based on color coding~\cite{alon1995color, alon2008biomolecular}.
Specifically, at each retained order $j$, we assign one of $j$ colors independently to every input column and physical output mode, and retain a contribution only when its selected columns and modes each have distinct colors.
These color conditions let us combine the retained contributions via a determinant coefficient formula, without enumerating the selected subsets separately.
Reweighting each retained contribution by the inverse of its inclusion probability preserves $q_{d,U}$ in expectation.
We then average the signed functions obtained from sufficiently many independent colorings to construct $\widehat q_{d,U}$ with controlled additional approximation error.
For a fixed collection of colorings used in this average, we obtain the prefix marginals of $\widehat q_{d,U}$ exactly by averaging the corresponding prefix marginals of the individual colored functions.
Accounting for all colorings in the average, this procedure requires $N^3\varepsilon^{-2}\delta^{-1}2^{O(d)}$ arithmetic operations per prefix query after $MN^2\varepsilon^{-2}\delta^{-1}2^{O(d)}$ preprocessing, thus making the exponent of $N$ independent of $d$.

We next use this signed approximation to construct a proper sampling algorithm.
Note that $\widehat q_{d,U}$ may take negative values and need not have total mass one, so we cannot use it directly as a probability distribution.
To apply the signed marginal sampling lemma~\cite{bremner2017achieving} (see Lemma~\ref{lem:signed-marginal-sampler}), we perform collision collapse: we retain the collision-free values of $\widehat q_{d,U}$ and assign the remaining signed mass to a single dummy outcome, yielding the normalized signed function $\widetilde q_{d,U}$.
Throughout the sequential sampling procedure, we hold the full collection of colorings used to construct $\widehat q_{d,U}$ fixed, so all queried prefix marginals are exact marginals of the same normalized signed function $\widetilde q_{d,U}$.
The signed marginal sampling lemma then produces a proper probability distribution over ordered first-quantized strings.
Pushing this distribution forward to the collision-collapsed Fock space yields the final output distribution $\mathcal A_d(U)$ of the algorithm, which incorporates the color-coding randomness used to construct $\widehat q_{d,U}$.

The collision-free approximation error controls the error at the dummy collision outcome as in Lemma~\ref{lem:collapse-error}, and ultimately the error of the proper sampler through Lemma~\ref{lem:signed-marginal-sampler}.
Combining this bound with the additional error introduced by the randomized marginal evaluation gives
\begin{align}
\E_U\!\left[
d_{\mathrm{TV}}\!\left(\mathcal A_d(U),P_{U,\mathrm{coll}}^{N,n}\right)^2
\right]
&\leq 4C_HR_d(n,N)+\frac{\varepsilon^2\delta}{2}.
\label{eq:proof-overview-sampling-second-moment}
\end{align}
The second term is the additional error introduced by color coding, reduced to this scale by averaging sufficiently many independent colorings.
The degree condition in Theorem~\ref{thm:main} bounds the first term by $\varepsilon^2\delta/2$ for a suitable constant $C$.
Hence, the Haar expectation of the squared total variation distance is at most $\varepsilon^2\delta$.
Markov's inequality then yields $d_{\mathrm{TV}}\!\left(\mathcal A_d(U),P_{U,\mathrm{coll}}^{N,n}\right)\leq\varepsilon$ with probability at least $1-\delta$ over $U \sim \Haar(M)$, giving Eq.~\eqref{eq:main-average-case}.

Finally, the signed marginal sampler makes $O(Mn)$ prefix-marginal queries.
Combining this query complexity with the marginal-evaluation cost, including preprocessing, yields a total arithmetic runtime of $MnN^3\varepsilon^{-2}\delta^{-1}2^{O(d)}$ per output sample, as stated in Theorem~\ref{thm:main}.


\section{Hoeffding decomposition of the averaged permanent statistic}
\label{sec:hoeffding}

This section carries out the first step of the proof roadmap, namely, constructing the low-degree truncation of the averaged permanent statistic $\Phi$ in Eq.~\eqref{eq:phi-definition}.  
We show that $\Phi$ admits a Hoeffding-type decomposition (see, e.g., Refs.~\cite{hoeffding1992class, rempala2002central, le2025hoeffding}) into canonical components indexed by subsets of input columns. 
The order-$j$ component appears with coefficient $\binom{n}{j}/\binom{N}{j}$, and this coefficient provides the suppression underlying the truncation when $N$ is large compared with $n$.


This section is divided into two parts.
We first derive a Hoeffding decomposition of $\Phi(B)$, valid for all complex matrices $B\in\mathbb C^{n\times N}$, and then introduce its degree-truncated version by keeping only components of order at most $d$.
We then specialize the input to a matrix $X \sim\mathcal N_{\mathbb C}(0,1)^{n\times N}$ with i.i.d.\ complex Gaussian entries and compute the error induced by discarding the higher-order Hoeffding components beyond $d$.


\subsection{Exact algebraic decomposition}

We begin with lower-order statistics indexed by input columns. 
Specifically, for $B \in \mathbb C^{n \times N}$ and $L\subseteq[N]$ with $\abs{L}\leq n$, define
\begin{align}
 G_L(B)
 &:={}
 (n-|L|)!
 \sum_{\substack{I\subseteq[n]\\\abs{I}=\abs{L}}}
 \abs{\Per(B_{I,L})}^2,
 \label{eq:G-definition}
\end{align}
with the conventions $\Per(\varnothing)=1$ and $G_\varnothing(B)=n!$; see Fig.~\ref{fig: hoeffding}(a) for a schematic illustration.

For $J\subseteq[N]$ with $\abs{J}\leq n$, define
\begin{align}
 H_J(B)
 &:={}
 \sum_{L\subseteq J}
 (-1)^{\abs{J}-\abs{L}}G_L(B),
 \label{eq:H-definition}
\end{align}
which is indeed the M\"obius transform of the set-indexed family $L\mapsto G_L(B)$ on the subset lattice of $J$~\cite{rota1964foundations, stanley2011enumerative}. 
A schematic illustration of $H_J$ is given in Fig.~\ref{fig: hoeffding}(b).


The next lemma is the M\"obius inversion formula on the lattice of subsets of $S\subseteq[N]$.

\begin{lemma}[M\"obius inversion]
\label{lem:boolean-inversion}
For every $S\subseteq[N]$ with $\abs{S}\leq n$, it follows that 
\begin{align}
 G_S(B)
 &=
 \sum_{J\subseteq S}H_J(B).
 \label{eq:boolean-inversion}
\end{align}
\end{lemma}

\begin{proof}
Summing Eq.~\eqref{eq:H-definition} over $J \subseteq S$, we obtain
\begin{align}
\sum_{J\subseteq S}H_J(B)
&=
\sum_{J\subseteq S}
\sum_{L\subseteq J}
(-1)^{\abs{J}-\abs{L}}G_L(B).
\label{eq:boolean-inversion-expanded}
\end{align}
The sum ranges over all pairs $(L,J)$ satisfying
$L\subseteq J\subseteq S$.  Reversing the order of summation gives
\begin{align}
\sum_{J\subseteq S}H_J(B)
&=
\sum_{L\subseteq S}
G_L(B)
\sum_{\substack{J:\,L\subseteq J\subseteq S}}
(-1)^{\abs{J}-\abs{L}}.
\label{eq:boolean-inversion-reordered}
\end{align}
For fixed $L\subseteq S$, every intermediate set $J$ can be written uniquely as
$J=L\cup R$, where $R\subseteq S\setminus L$.  Hence
\begin{align}
\sum_{\substack{J:\,L\subseteq J\subseteq S}}
(-1)^{\abs{J}-\abs{L}}
&=
\sum_{R\subseteq S\setminus L}
(-1)^{\abs{R}}
\nonumber \\
&=
\sum_{r = 0}^{\abs{S} - \abs{L}}
\binom{\abs{S} - \abs{L}}{r}(-1)^r
\nonumber \\
&=
(1-1)^{\abs{S}-\abs{L}}
\nonumber \\
&=
\begin{cases}
1,& L=S,\\
0,& L\subsetneq S.
\end{cases}
\label{eq:boolean-inner-cancellation}
\end{align}
Only the term $L=S$ remains in Eq.~\eqref{eq:boolean-inversion-reordered}, and therefore
\begin{align}
\sum_{J\subseteq S}H_J(B)
&=
G_S(B).
\end{align}
This proves the lemma.
\end{proof}

\begin{figure*}[t]
\includegraphics[width=0.93\linewidth]{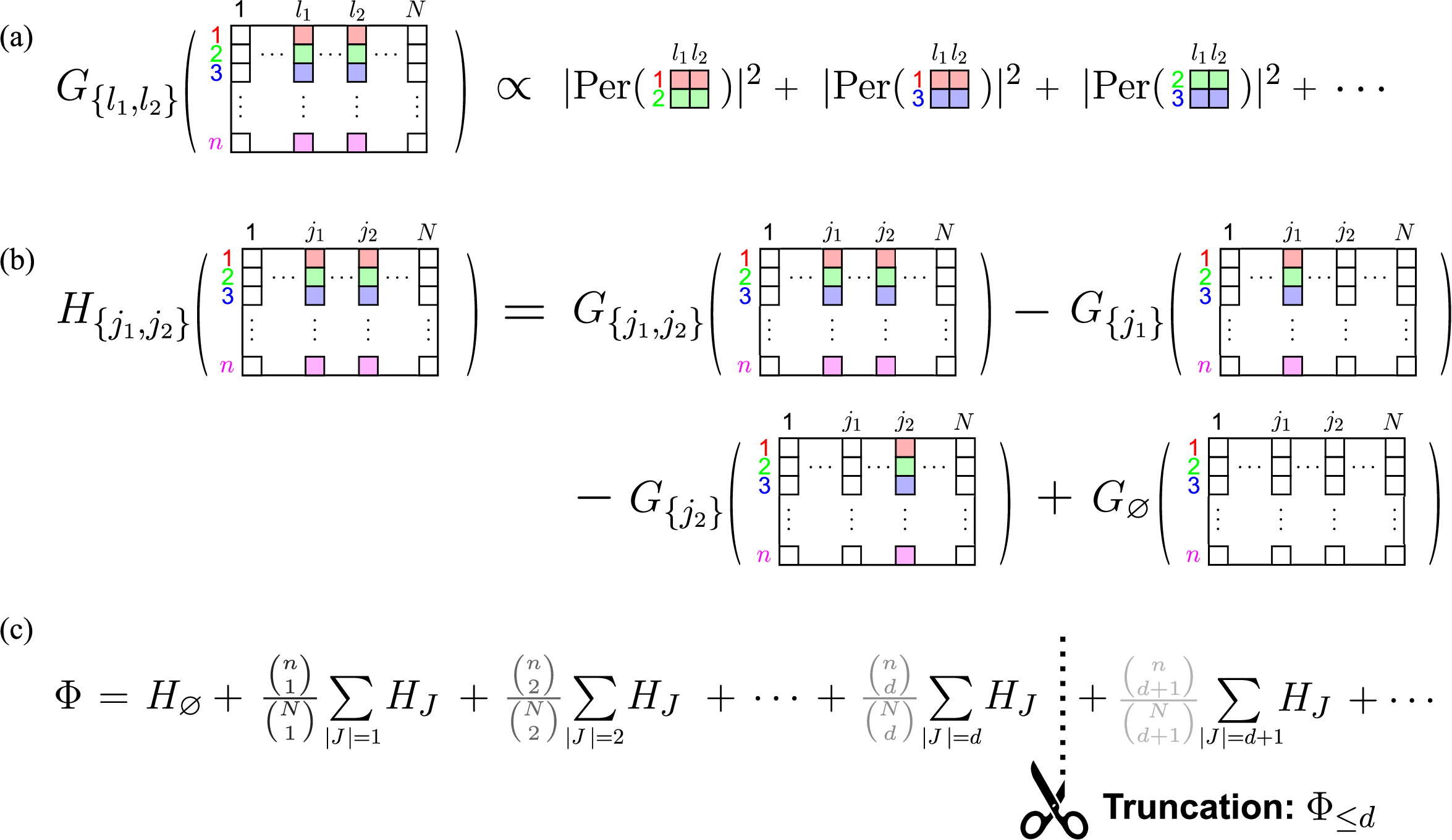}
\caption{Schematic illustration of $G_L$, $H_J$, the Hoeffding decomposition of $\Phi$, and its degree-$d$ truncation.
(a) For $L=\{l_1,l_2\}$, $G_L(B)$ is obtained, up to the prefactor $(n-\abs{L})!$, by summing the squared permanents of the submatrices $B_{I,L}$ over all row subsets $I\subseteq[n]$ with $\abs{I}=\abs{L}$.
(b) For $J=\{j_1,j_2\}$, the component $H_J(B)$ is obtained from $\{G_L:L\subseteq J\}$ by M\"obius transform,
$H_{\{j_1,j_2\}}=G_{\{j_1,j_2\}}-G_{\{j_1\}}-G_{\{j_2\}}+G_{\varnothing}$.
(c) Grouping the Hoeffding decomposition of $\Phi$ by the order $j=\abs{J}$ yields blocks weighted by $\binom{n}{j}/\binom{N}{j}$.
For $n<N$, these coefficients decrease with $j$, as indicated schematically by the fading prefactors, which motivates the truncation $\Phi_{\le d}$ obtained by keeping the terms with $\abs{J}\le d$ and discarding the higher-order tail.
}
\label{fig: hoeffding}
\end{figure*}

Applying this inversion identity to the averaged permanent statistic $\Phi$ in Eq.~\eqref{eq:phi-definition} yields the following Hoeffding decomposition.

\begin{proposition}[Hoeffding decomposition of the averaged permanent statistic]
\label{prop:exact-decomposition}
For every $B\in\mathbb C^{n\times N}$, with $H_J$ defined by Eq.~\eqref{eq:H-definition}, the averaged permanent statistic $\Phi$ in Eq.~\eqref{eq:phi-definition} admits a decomposition
\begin{align}
 \Phi(B)
 &=
 n!+
 \sum_{j=1}^n
 \frac{\binom{n}{j}}{\binom{N}{j}}
 \sum_{\substack{J\subseteq[N]\\\abs{J}=j}}
 H_J(B).
 \label{eq:exact-decomposition}
\end{align}
\end{proposition}

\begin{proof}
If $\abs{S}=n$, then the definition of $G_S(B)$ contains only the row subset $I=[n]$.  Hence
\begin{align}
G_S(B)
=
(n-n)!
\sum_{\substack{I\subseteq[n]\\\abs{I}=n}}
\abs{\Per(B_{I,S})}^2
=
\abs{\Per(B_S)}^2.
\label{eq:G-top-order}
\end{align}
Using Lemma~\ref{lem:boolean-inversion} in the definition of $\Phi$ gives
\begin{align}
\Phi(B)
&=
\frac{1}{\binom{N}{n}}
\sum_{\substack{S\subseteq[N]\\\abs{S}=n}}
G_S(B)
=
\frac{1}{\binom{N}{n}}
\sum_{\substack{S\subseteq[N]\\\abs{S}=n}}
\sum_{J\subseteq S}H_J(B).
\label{eq:phi-double-subset-sum}
\end{align}
We reorganize the double sum by first fixing $J$.  If $\abs{J}=j$, then the number of $n$-element sets $S$ satisfying $J\subseteq S\subseteq[N]$ is $\binom{N-j}{n-j}$, so it holds that 
\begin{align}
\Phi(B)
&=
\sum_{j=0}^{n}
\frac{\binom{N-j}{n-j}}{\binom{N}{n}}
\sum_{\substack{J\subseteq[N]\\\abs{J}=j}}
H_J(B).
\label{eq:decomp-counting}
\end{align}
The coefficient can then be rewritten as
\begin{align}
\frac{\binom{N-j}{n-j}}{\binom{N}{n}}
&=
\frac{
\dfrac{(N-j)!}{(n-j)!(N-n)!}
}{
\dfrac{N!}{n!(N-n)!}
}
\nonumber\\
&=
\frac{n!(N-j)!}{(n-j)!N!}
\nonumber\\
&=
\frac{\binom{n}{j}}{\binom{N}{j}}.
\label{eq:decomposition-binomial-identity}
\end{align}
The contribution with $j=0$ is $H_\varnothing(B)=G_\varnothing(B)=n!$.  Separating this term from Eq.~\eqref{eq:decomp-counting} proves Eq.~\eqref{eq:exact-decomposition}.
\end{proof}

Importantly, for $n<N$, the coefficient $\binom{n}{j}/\binom{N}{j}$ decreases with the component order $j$ (see Fig.~\ref{fig: hoeffding}(c) for illustration).
This motivates our truncation-based simulation algorithm, in which we retain only the low-order components of Eq.~\eqref{eq:exact-decomposition}.
For later convenience, we define the truncation explicitly.

\begin{definition}[Degree-$d$ truncation of $\Phi$]\label{def:truncated-phi}
For $0\leq d\leq n$, we define the degree-$d$ truncation $\Phi_{\leq d}$ of the averaged permanent statistic $\Phi$ by retaining only components of order at most $d$ in Proposition~\ref{prop:exact-decomposition}, such that 
\begin{align}
 \Phi_{\leq d}(B)
 &:={}
 n!+
 \sum_{j=1}^d
 \frac{\binom{n}{j}}{\binom{N}{j}}
 \sum_{\substack{J\subseteq[N]\\\abs{J}=j}}
 H_J(B).
 \label{eq:truncated-phi}
\end{align}
\end{definition}


\subsection{Gaussian truncation error}

We now quantify the truncation error between the averaged permanent statistic $\Phi$ and its truncated version $\Phi_{\leq d}$ when the input matrix is Gaussian.
Specifically, let $X=(Z_1,\ldots,Z_N) \sim\mathcal N_{\mathbb C}(0,1)^{n\times N}$ have i.i.d.\ complex Gaussian entries; equivalently, the columns $Z_1,\ldots,Z_N$ are mutually independent with $Z_l\sim\mathcal N_{\mathbb C}(0,1)^{n\times 1}$ for every $l\in[N]$.
For later use, for a column index $l\in[N]$ and row indices $a,b\in[n]$, we define the centered bilinear form
\begin{align}
 B_{ab}^{(l)}
 &:={}
 X_{a,l}\overline{X_{b,l}}-\delta_{ab}.
 \label{eq:centered-bilinear}
\end{align}
By construction, $B_{ab}^{(l)}$ is centered with respect to the column $Z_l$, namely,
\begin{align}\label{eq:centered-bilinear-centered}
    \E_{Z_l}[B_{ab}^{(l)}] = 0.
\end{align}

We first turn $H_J$ in Eq.~\eqref{eq:H-definition} into an explicit product of centered bilinears $B_{ab}^{(l)}$ in Eq.~\eqref{eq:centered-bilinear}.


\begin{lemma}[Centered-bilinear expansion]
\label{lem:centered-expansion}
Let $X=(Z_1,\ldots,Z_N) \sim\mathcal N_{\mathbb C}(0,1)^{n\times N}$ have i.i.d.\ complex Gaussian entries, and let $B_{ab}^{(l)}$ be defined as in Eq.~\eqref{eq:centered-bilinear}.
Fix $J\subseteq[N]$.
Then
\begin{align}
H_J(X)
&=
(n-\abs{J})!
\sum_{\substack{R\subseteq[n]\\\abs{R}=\abs{J}}}
\sum_{\sigma,\tau\in\operatorname{Bij}(J,R)}
\prod_{l\in J}
B_{\sigma(l),\tau(l)}^{(l)},
\label{eq:centered-expansion}
\end{align}
where $\operatorname{Bij}(J,R)$ denotes the set of all bijections from $J$ to $R$.
In particular, for every $l\in J$,
\begin{align}
\E_{Z_l}[H_J(X)]
&=
0.
\label{eq:canonical-degeneracy}
\end{align}
\end{lemma}

\begin{proof}
The case $J=\varnothing$ gives $H_\varnothing=G_\varnothing=n!$, which agrees with Eq.~\eqref{eq:centered-expansion}.  
Hence, assume $\abs{J} \geq 1$.


For each $l \in J$, let $P_l$ be the operator that takes the expectation over the $l$th column $Z_l$, while keeping all other columns fixed.
Then, for every $L\subseteq J$,  it holds that 
\begin{align}
\left(
\prod_{l\in J\setminus L}P_l
\right)
G_J(X)
&=
G_L(X).
\label{eq:partial-expectation-GJ}
\end{align}
We now show how Eq.~\eqref{eq:partial-expectation-GJ} holds. 
Expand both permanents in $G_J$, which yields 
\begin{equation}
G_J(X)
=
(n-\abs{J})!
\sum_{\substack{R\subseteq[n]\\\abs{R}=\abs{J}}}
\sum_{\sigma,\tau\in\operatorname{Bij}(J,R)}
\prod_{l\in J}X_{\sigma(l),l}\overline{X_{\tau(l),l}}. 
\label{eq:appendix-GJ-expansion}
\end{equation}
Here, $\operatorname{Bij}(J,R)$ denotes the set of all bijections from $J$ to $R$, which encode the row-column matchings in the permanent expansion; these play the role of permutations when the row and column index sets are distinct.
Then, because the entries of the column $Z_l$ are standard complex Gaussian, it holds for every $l\in J\setminus L$ that 
\begin{align}\label{eq:orthogonality-of-Gaussian-entries}
\E\!\left[X_{\sigma(l),l}\overline{X_{\tau(l),l}}\right]
=\delta_{\sigma(l),\tau(l)}.
\end{align}
Thus, when applying $P_l$ to $G_J(X)$, each term in Eq.~\eqref{eq:appendix-GJ-expansion} survives only when $\sigma$ and $\tau$ agree on $J \setminus L$.
Fix an $|L|$-element row set $R_L$ and bijections $\sigma_L,\tau_L:L\to R_L$.  Extending them to surviving $\sigma, \tau$ amounts to choosing one common injection from $J \setminus L$ into $[n]\setminus R_L$, where the number of choices is ${(n-\abs{L})!}/{(n-\abs{J})!}$.
Multiplying by the prefactor $(n-\abs{J})!$ in Eq.~\eqref{eq:appendix-GJ-expansion} leaves $(n-\abs{L})!$, and the remaining sum is exactly $G_L$ (replacing $R \rightarrow R_L$ and $\sigma, \tau \rightarrow \sigma_L,\tau_L$ in Eq.~\eqref{eq:appendix-GJ-expansion}), yielding Eq.~\eqref{eq:partial-expectation-GJ}.

Now let $Q_l:=\operatorname{Id}-P_l$.  
By Eq.~\eqref{eq:orthogonality-of-Gaussian-entries}, applying $Q_l$ to the factor
$X_{\sigma(l),l}\overline{X_{\tau(l),l}}$ in the expansion of $G_J$ in Eq.~\eqref{eq:appendix-GJ-expansion} replaces it by the centered bilinear form $B_{\sigma(l),\tau(l)}^{(l)}$.
Combining this observation with the expansion $\prod_{l\in J}Q_l=\prod_{l\in J}(\operatorname{Id}-P_l)$ and the identity in Eq.~\eqref{eq:partial-expectation-GJ}, we obtain
\begin{align}
H_J(X)
&=
\sum_{L\subseteq J}
(-1)^{\abs{J}-\abs{L}}G_L(X)
\nonumber\\
&=
\sum_{L\subseteq J}
(-1)^{\abs{J}-\abs{L}}
\left(
\prod_{l\in J\setminus L}P_l
\right)
G_J(X)
\nonumber\\
&=
\left(
\prod_{l\in J}Q_l
\right)
G_J(X) 
\nonumber \\
&=
(n-\abs{J})!
\sum_{\substack{R\subseteq[n]\\\abs{R}=\abs{J}}}
\sum_{\sigma,\tau\in\operatorname{Bij}(J,R)}
\prod_{l\in J}
B_{\sigma(l),\tau(l)}^{(l)}, 
\label{eq:centering-operator}
\end{align}
which is exactly Eq.~\eqref{eq:centered-expansion}.


Finally, every summand in Eq.~\eqref{eq:centered-expansion} contains the centered form $B_{\sigma(l),\tau(l)}^{(l)}$ for each $l\in J$.
By Eq.~\eqref{eq:centered-bilinear-centered}, this factor has zero expectation with respect to column $l$.
Therefore, taking the expectation over any one column $l\in J$ annihilates every summand, yielding Eq.~\eqref{eq:canonical-degeneracy}.

\end{proof}

We next examine the inner products of these components. 

\begin{lemma}[Orthogonality and second moment]
\label{lem:component-norm}
Let $X=(Z_1,\ldots,Z_N) \sim\mathcal N_{\mathbb C}(0,1)^{n\times N}$ have i.i.d.\ complex Gaussian entries. 
For any distinct subsets $J,K\subseteq[N]$ with $\abs{J},\abs{K}\leq n$, the components $H_J$ and $H_K$ are orthogonal:
\begin{align}
 \E_X[H_J(X)\overline{H_K(X)}]
 &=
 0.
 \label{eq:subset-orthogonality}
\end{align}
Moreover, the second moment of each component is
\begin{align}
 \E_X[\abs{H_J(X)}^2]
 &=
 \frac{(n!)^2}{\binom{n}{|J|}}.
 \label{eq:component-norm}
\end{align}
\end{lemma}



\begin{proof}
For each fixed column $l$, it holds that 
\begin{align}
 \E\!\left[
 X_{a,l}\overline{X_{b,l}}\,\overline{X_{a',l}}X_{b',l}
\right]
 &=
 \delta_{ab}\delta_{a'b'}
 +
 \delta_{aa'}\delta_{bb'},
 \label{eq:complex-fourth-moment}
\end{align}
by applying Wick's theorem.
Consequently, the centered-bilinear form in Eq.~\eqref{eq:centered-bilinear} satisfies
\begin{align}
&\E\!\left[
B_{ab}^{(l)}\overline{B_{a'b'}^{(l)}}
\right] \nonumber \\
&=
\E\!\left[
\left(
X_{a,l}\overline{X_{b,l}}-\delta_{ab}
\right)
\left(
\overline{X_{a',l}}X_{b',l}-\delta_{a'b'}
\right)
\right]
\nonumber\\
&=
\left(
\delta_{ab}\delta_{a'b'}
+
\delta_{aa'}\delta_{bb'}
\right)
-
\delta_{a'b'}\delta_{ab}
-
\delta_{ab}\delta_{a'b'}
+
\delta_{ab}\delta_{a'b'}
\nonumber\\
&=
\delta_{aa'}\delta_{bb'}.
\label{eq:bilinear-orthonormal}
\end{align}
Thus, for each fixed column $l$, the centered bilinears ${B_{ab}^{(l)}}$ form an orthonormal family.

Using Eq.~\eqref{eq:centered-expansion}, each summand in $H_J(X)$ is indexed by a row set $R\subseteq[n]$ with $\abs{R}=\abs{J}$ and two bijections $\sigma,\tau:J\to R$.  
Expanding $\abs{H_J(X)}^2$ and taking expectation over $X$, the independence of the columns and the orthonormality relation in Eq.~\eqref{eq:bilinear-orthonormal} imply that two product terms are orthogonal unless the corresponding bijections agree pointwise on every $l \in J$.
Specifically, there are $\binom{n}{\abs{J}}(\abs{J}!)^2$ such products, each with factor $(n-\abs{J})!^2$. 
Therefore,
\begin{align}
\E_X[\abs{H_J(X)}^2]
=
(n-\abs{J})!^2\binom{n}{\abs{J}}(\abs{J}!)^2
=
\frac{(n!)^2}{\binom{n}{\abs{J}}}.
\label{eq:norm-count}
\end{align}
This gives Eq.~\eqref{eq:component-norm}.

It remains to prove orthogonality for different subsets.
If $J\neq K$, then at least one of $J\setminus K$ and $K\setminus J$ is nonempty.
By exchanging the roles of $J$ and $K$ if necessary, choose $l\in J\setminus K$. Then $H_K$ is independent of the $l$th column $Z_l$, while Eq.~\eqref{eq:canonical-degeneracy} gives
\begin{align}
 \E_{Z_l}[H_J]
 &=
 0.
\end{align}
Conditioning on all columns except $Z_l$ gives Eq.~\eqref{eq:subset-orthogonality}.
\end{proof}

We can now compute the truncation error exactly.  The next proposition is the quantitative low-degree approximation statement used in the rest of the paper.

\begin{proposition}[Exact Gaussian truncation error]
\label{prop:gaussian-truncation}
Let $X \sim\mathcal N_{\mathbb C}(0,1)^{n\times N}$ have i.i.d.\ complex Gaussian entries.  
For every $0\leq d\leq n$, it holds that
\begin{align}
 \frac{
 \E_X\abs{\Phi(X)-\Phi_{\leq d}(X)}^2
 }{(n!)^2}
 &=
 \sum_{j=d+1}^n
 \frac{\binom{n}{j}}{\binom{N}{j}}
 =:
 R_d(n,N).
 \label{eq:exact-truncation-error}
\end{align}
Moreover, if $N>n$, then
\begin{align}
 R_d(n,N)
 &\leq
 \frac{(n/N)^{d+1}}{1-n/N}.
 \label{eq:geometric-truncation-error}
\end{align}
\end{proposition}

\begin{proof}
By Proposition~\ref{prop:exact-decomposition} and the definition of $\Phi_{\leq d}$, we have 
\begin{align}
 \Phi(X)-\Phi_{\leq d}(X)
 &=
 \sum_{j=d+1}^n
 \frac{\binom{n}{j}}{\binom{N}{j}}
 \sum_{\substack{J\subseteq[N]\\\abs{J}=j}}H_J(X).
 \label{eq:truncation-expansion}
\end{align}
When we take the squared modulus and average over $X$, the orthogonality statement in Lemma~\ref{lem:component-norm} removes all cross terms between distinct subsets $J$.
Therefore,
\begin{align}
 \E\abs{\Phi(X)-\Phi_{\leq d}(X)}^2
 &=
 \sum_{j=d+1}^n
 \left(
 \frac{\binom{n}{j}}{\binom{N}{j}}
\right)^2
 \sum_{\substack{J\subseteq[N]\\\abs{J}=j}}
 \E\abs{H_J(X)}^2
 \nonumber\\
 &=
 \sum_{j=d+1}^n
 \left(
 \frac{\binom{n}{j}}{\binom{N}{j}}
\right)^2
 \binom{N}{j}
 \frac{(n!)^2}{\binom{n}{j}}
 \nonumber\\
 &=
 (n!)^2
 \sum_{j=d+1}^n
 \frac{\binom{n}{j}}{\binom{N}{j}}.
 \label{eq:truncation-orthogonal-sum}
\end{align}
This proves Eq.~\eqref{eq:exact-truncation-error}.

For the geometric bound, assume $N>n$.  
For every $1 \leq j\leq n$,
\begin{align}
 \frac{\binom{n}{j}}{\binom{N}{j}}
 =
 \prod_{s=0}^{j-1}
 \frac{n-s}{N-s}
 \leq
 \left(\frac{n}{N}\right)^j.
 \label{eq:binomial-ratio-geometric}
\end{align}
Summing the resulting geometric series gives
\begin{align}
R_d(n,N)
\leq
\sum_{j=d+1}^{n}
\left(\frac{n}{N}\right)^j
\leq
\frac{(n/N)^{d+1}}{1-n/N}.
\end{align}
This proves Eq.~\eqref{eq:geometric-truncation-error}.
\end{proof}


\section{From Gaussian truncation to Haar-average error}
\label{sec:haar-error}

In this section, we transfer the exact Gaussian truncation estimate of Proposition~\ref{prop:gaussian-truncation}  to the Haar-random interferometer used in the sampling problem.
The ideal ordered first-quantized string distribution is $p_U$ defined in Eq.~\eqref{eq:ordered-probability}, while the new object $q_{d,U}$, defined below, is its signed approximation obtained by evaluating $\Phi_{\leq d}$ on the scaled Haar submatrices $A_{\bm r} = \sqrt M \,U_{\bm r,[N]}$ associated with a first-quantized output string $\bm r\in[M]^{n}$. 
The main result of this section is an $L^1$-error bound for $p_U-q_{d,U}$ restricted to collision-free output strings, expressed in terms of the Gaussian truncation error $R_d(n,N)$ obtained in Proposition~\ref{prop:gaussian-truncation}.




We first define the signed approximation to the first-quantized distribution $p_U$ in Eq.~\eqref{eq:ordered-probability} as follows.  

\begin{definition}[Signed approximation of $p_U$]\label{def:signed-ordered-approximation}
For each output string $\bm r\in[M]^n$, we define the signed function $q_{d,U}$ pointwise by
\begin{align}
 q_{d,U}(\bm r)
 &:={}
 \frac{\Phi_{\leq d}(A_{\bm r})}{n!M^n},
 \label{eq:qd-definition}
\end{align}
where $A_{\bm r}=\sqrt M\,U_{\bm r,[N]}$, and $\Phi_{\leq d}$ denotes the degree-$d$ truncation in Definition~\ref{def:truncated-phi}.
\end{definition} 
Note that $q_{d,U}$ is a signed function and may take negative values; it is therefore not, in general, a probability distribution. 



We also define the collision-free $L^1$ error of the signed approximation $q_{d,U}$ relative to $p_U$:
\begin{align}
 \Delta_{\mathrm{cf}}(U)
 &:={}
 \sum_{\bm r\in [M]_{\neq}^{n} }
 \abs{
 p_U(\bm r)-q_{d,U}(\bm r)
 },
 \label{eq:Delta-cf-definition}
\end{align}
where $[M]_{\neq}^{n} $ is the set of collision-free strings defined in Eq.~\eqref{eq:collision-free-strings}.


\subsection{From Gaussians to rectangular Haar submatrices}

We first establish a bridge between the Gaussian ensemble used in Sec.~\ref{sec:hoeffding} and the Haar-random interferometers used in the sampling algorithm.

\begin{lemma}[Rectangular Haar--Gaussian density domination]
\label{lem:rectangular-domination}
There exist constants $C_0, C_H>0$ such that the following holds.  
Let $1\leq n\leq N\leq M$ and suppose
\begin{align}
    M \geq C_0 n^2N.
\end{align}
Let $p_{\mathrm H}$ be the density of $\sqrt M\,U_{R,[N]}$,
where $U\sim\Haar(M)$ and $R\subseteq[M]$ is a fixed set of size $n$.
Let $p_{\mathrm G}$ be the density of an $n\times N$ matrix
with i.i.d.\ complex Gaussian entries.
Then
\begin{align}
p_{\mathrm H}(Y)
&\leq C_Hp_{\mathrm G}(Y),
\label{eq:rectangular-Haar-density-domination}
\end{align}
for every $Y\in\mathbb C^{n\times N}$.
\end{lemma}

We defer the proof to Appendix~\ref{app:haar-details}, as it requires a direct comparison between the exact density of rectangular truncations of Haar-random unitary matrices and that of matrices with i.i.d.\ complex Gaussian entries, which lies beyond the main scope of this work.

\subsection{Pointwise error for one collision-free string}


Next, we use the density comparison in Lemma~\ref{lem:rectangular-domination} to control the squared truncation error between $p_U(\bm r)$ and $q_{d,U}(\bm r)$ for a fixed output string $\bm r$, with the expectation taken over $U\sim\Haar(M)$.
Here, we restrict to collision-free first-quantized output strings $\bm r \in [M]_{\neq}^{n}$ in Eq.~\eqref{eq:collision-free-strings}.
This ensures that the rows of $A_{\bm r}$ are distinct and hence that Lemma~\ref{lem:rectangular-domination} applies.



\begin{lemma}[Pointwise Haar truncation error]
\label{lem:pointwise-haar-error}
Let $C_0,C_H>0$ be the constants in Lemma~\ref{lem:rectangular-domination}, and let $ M \geq C_0 n^2N $.
For every $\bm r\in[M]_{\neq}^{n} $,
\begin{align}
 \E_{U\sim\Haar(M)}
 \abs{
 p_U(\bm r)-q_{d,U}(\bm r)
 }^2
 &\leq
 \frac{C_HR_d(n,N)}{M^{2n}}.
 \label{eq:pointwise-haar-error}
\end{align}
\end{lemma}

\begin{proof}

Because $\bm r\in[M]_{\neq}^{n}$ is collision-free, the submatrix $A_{\bm r}=\sqrt M \,U_{\bm r,[N]}$ is obtained by selecting $n$ distinct rows of the Haar-random unitary $U$. 
Hence, $A_{\bm r}$ follows the Haar-submatrix density $p_{\mathrm H}$ from Lemma~\ref{lem:rectangular-domination}.

By the definitions of $p_U$ and $q_{d,U}$, we have 
\begin{align}
 p_U(\bm r)-q_{d,U}(\bm r)
 &=
 \frac{
 \Phi(A_{\bm r})-\Phi_{\leq d}(A_{\bm r})
 }{n!M^n}.
 \label{eq:p-minus-q-as-phi-error}
\end{align}
Therefore,
\begin{align}
&\E_{U\sim\Haar(M)}
 \abs{
 p_U(\bm r)-q_{d,U}(\bm r)
 }^2
 \nonumber\\
&=
 \frac{1}{(n!)^2M^{2n}}
 \int_{\mathbb C^{n\times N}}
 \abs{
 \Phi(X)-\Phi_{\leq d}(X)
 }^2
 p_{\mathrm H}(X)\,dX
 \nonumber\\
&\leq
 \frac{C_H}{(n!)^2M^{2n}}
 \int_{\mathbb C^{n\times N}}
 \abs{
 \Phi(X)-\Phi_{\leq d}(X)
 }^2
 p_{\mathrm G}(X)\,dX
 \nonumber\\
&=
 \frac{C_H}{(n!)^2M^{2n}}
 \E_{X\sim p_{\mathrm G}}
 \abs{
 \Phi(X)-\Phi_{\leq d}(X)
 }^2.
 \label{eq:haar-to-gaussian-error-transfer}
\end{align}
Moreover, Proposition~\ref{prop:gaussian-truncation} gives
\begin{align}
 \E_{X\sim p_{\mathrm G}}
 \abs{
 \Phi(X)-\Phi_{\leq d}(X)
 }^2
 &=
 (n!)^2R_d(n,N).
 \label{eq:gaussian-truncation-used}
\end{align}
Substituting Eq.~\eqref{eq:gaussian-truncation-used} into Eq.~\eqref{eq:haar-to-gaussian-error-transfer} proves Eq.~\eqref{eq:pointwise-haar-error}.
\end{proof}

\subsection{Global error over collision-free strings}

Finally, we convert the pointwise second-moment bound in Lemma~\ref{lem:pointwise-haar-error} into an $L^1$-error bound over the collision-free output set $[M]_{\neq}^{n}$. 


\begin{proposition}[Collision-free $L^1$ error]
\label{prop:collision-free-L1}
Let $C_0,C_H>0$ be the constants in Lemma~\ref{lem:rectangular-domination}, and let $ M \geq C_0 n^2N $. 
Then the collision-free $L^1$ error $\Delta_{\mathrm{cf}}(U)$ in Eq.~\eqref{eq:Delta-cf-definition} satisfies
\begin{align}
 \E_{U\sim\Haar(M)}
 \left[
 \Delta_{\mathrm{cf}}(U)^2
\right]
 &\leq
 C_HR_d(n,N).
 \label{eq:collision-free-L1-second-moment}
\end{align}
Consequently, for every $\varepsilon>0$, it holds that 
\begin{align}
 \Prob_{U\sim\Haar(M)}
 \left[
 \Delta_{\mathrm{cf}}(U)>\varepsilon
\right]
 &\leq
 \frac{C_HR_d(n,N)}{\varepsilon^2}.
 \label{eq:collision-free-L1-high-prob}
\end{align}
\end{proposition}

\begin{proof}
Applying the Cauchy--Schwarz inequality to the sum defining $\Delta_{\mathrm{cf}}(U)$ gives
\begin{align}
\Delta_{\mathrm{cf}}(U)^2
&=
\left(
\sum_{\bm r\in[M]_{\neq}^{n} }
\abs{p_U(\bm r)-q_{d,U}(\bm r)}
\right)^2 
\nonumber \\
&\leq
\abs{[M]_{\neq}^{n} }
\sum_{\bm r\in[M]_{\neq}^{n} }
\abs{p_U(\bm r)-q_{d,U}(\bm r)}^2.
\label{eq:cs-collision-free-error}
\end{align}
Here, the number of collision-free strings in $[M]_{\neq}^{n} $ is given by a falling factorial:
\begin{align}
 \abs{[M]_{\neq}^{n} }
 &=
 (M)_n
 =
 M(M-1)\cdots(M-n+1)
 \leq
 M^n.
 \label{eq:number-collision-free-strings}
\end{align}
Taking expectation of $\Delta_{\mathrm{cf}}(U)^2$ over $U \sim \Haar(M)$ and applying Lemma~\ref{lem:pointwise-haar-error} gives
\begin{align}
 \E_U[\Delta_{\mathrm{cf}}(U)^2]
 &\leq
 (M)_n
 \sum_{\bm r\in[M]_{\neq}^{n} }
 \E_U\left[\abs{p_U(\bm r)-q_{d,U}(\bm r)}^2\right]
 \nonumber\\
 &\leq
 (M)_n^2
 \frac{C_HR_d(n,N)}{M^{2n}}
 \nonumber\\
 &\leq
 C_HR_d(n,N).
 \label{eq:Delta-cf-second-moment-proof}
\end{align}
The two collision-free counting factors contribute at most $(M)_n^2\leq M^{2n}$ and cancel the pointwise scale $M^{-2n}$, so no union bound over the collision-free strings is required.
This proves Eq.~\eqref{eq:collision-free-L1-second-moment}.  
Moreover, Markov's inequality applied to $\Delta_{\mathrm{cf}}(U)^2$ gives
\begin{align}
 \Prob_U
 \left[
 \Delta_{\mathrm{cf}}(U)> \varepsilon
\right]
 &=
 \Prob_U
 \left[
 \Delta_{\mathrm{cf}}(U)^2> \varepsilon^2
\right]
 \nonumber\\
 &\leq
 \frac{\E_U[\Delta_{\mathrm{cf}}(U)^2]}{\varepsilon^2}
 \nonumber\\
 &\leq
 \frac{C_HR_d(n,N)}{\varepsilon^2}.
\end{align}
This proves Eq.~\eqref{eq:collision-free-L1-high-prob}.
\end{proof}

\section{Marginal evaluation and the collision-collapsed signed function}
\label{sec:marginals}

The previous section bounded the Haar-averaged squared $L^1$ error between $p_U$ and $q_{d,U}$ over the collision-free outcome space $[M]_{\neq}^{n}$.
This approximation bound alone, however, does not directly provide an efficient sampling algorithm.
The sequential sampling procedure used in our proof, adapted from Ref.~\cite{bremner2017achieving}, requires efficiently computable prefix marginals of a normalized signed function.

This section proceeds as follows.
We first define the collision-free prefix marginals and explain the difficulty of evaluating them directly.
We then construct a randomized signed approximation with efficiently computable marginals.
Lastly, we extend this randomized approximation to a normalized signed function by assigning the remaining signed mass to a dummy collision outcome.

\subsection{Prefix marginals of the signed approximation}

For a collision-free prefix $\mathbf a=(a_1,\ldots,a_\ell)\in[M]_{\neq}^{\ell}$ with $0\leq\ell\leq n$, define
\begin{align}
 m_{d,U}(\mathbf a)
 &:={}
 \sum_{\substack{r_{\ell+1},\ldots,r_n:\\
 (\mathbf a,r_{\ell+1},\ldots,r_n)\in[M]_{\neq}^{n}}}
 q_{d,U}(\mathbf a,r_{\ell+1},\ldots,r_n),
 \label{eq:Qcf-marginal-identity}
\end{align}
where $q_{d,U}$ is the signed approximation in Definition~\ref{def:signed-ordered-approximation}.
Here, $\mathbf a$ records the first $\ell$ output-mode labels already fixed by the sequential sampler, while the sum ranges over all collision-free choices of the remaining $n-\ell$ labels.
For $\ell=0$, we write $\mathbf a=\varnothing$, and $m_{d,U}(\varnothing)$ is the total collision-free signed mass, which need not equal one.
For $\ell=n$, the marginal is simply the pointwise value $q_{d,U}(\mathbf a)$.

Directly evaluating Eq.~\eqref{eq:Qcf-marginal-identity} involves $(M-\ell)_{n-\ell}\leq M^n$ collision-free completions, even when $d$ is small.
Although the completion sum can be reduced to an explicit summation over the physical output modes associated with the selected terms, together with combinatorial counting of the remaining completions, enumerating the selected input-column and physical-mode sets remains costly.
Specifically, at order $j$, the centered expansion in Lemma~\ref{lem:centered-expansion} selects an input-column set $J\subseteq[N]$ and a set $R\subseteq[n]$ of output-string positions.
As the output string ranges over completions of a prefix, the corresponding physical-mode set $\{r_a:a\in R\}$ ranges over subsets of $[M]$.
Separately enumerating the input and output mode sets can therefore involve up to $\binom Mj\binom Nj$ choices at order $j$, leading to a direct-evaluation cost of the form $(MN)^{O(d)}$.
For a fixed surviving fraction $n/N\in(0,1)$, a truncation degree logarithmic in the inverse error parameters suffices, as indicated by Eq.~\eqref{eq:degree-choice}.
The exponent of the direct-evaluation cost therefore grows with the truncation degree, constituting the remaining computational bottleneck.

\subsection{Randomized prefix-marginal evaluation}
\label{subsec:randomized-marginals}

We overcome this obstacle using random color coding~\cite{alon1995color, alon2008biomolecular}.
Specifically, for each retained order $1\leq j\leq d$ in $q_{d,U}$, we randomly assign $j$ colors to every input column and physical output mode.
We then retain only contributions whose selected input columns have distinct colors and whose selected physical output modes likewise have distinct colors, allowing these choices to be summed efficiently.
Reweighting the retained contributions yields an unbiased signed estimator of $q_{d,U}$, and averaging independent estimators gives $\widehat q_{d,U}$ with controlled additional error.
Here, one coloring refers to the full collection of color assignments for all $1\leq j\leq d$, held fixed across all output strings and prefix queries.
We write $\E_{\mathrm{col}}$ for expectation over the colorings used in this average, with $U$ held fixed.

Let $\widehat m_{d,U}(\mathbf a)$ denote the prefix marginal obtained by replacing $q_{d,U}$ with $\widehat q_{d,U}$ in Eq.~\eqref{eq:Qcf-marginal-identity}.
We also define the collision-free $L^1$ error of $\widehat q_{d,U}$ relative to $p_U$ by
\begin{align}
 \widehat\Delta_{\mathrm{cf}}(U)
 &:={}
 \sum_{\bm r\in[M]_{\neq}^{n}}
 \abs{p_U(\bm r)-\widehat q_{d,U}(\bm r)}.
 \label{eq:randomized-Delta-cf-definition}
\end{align}
The following proposition gives the evaluation cost and the error bound for this randomized approximation.

\begin{proposition}[Randomized marginal-evaluation complexity]
\label{prop:randomized-marginal-complexity}
Let $1\leq n\leq N\leq M$, $d\in\{0,\ldots,n\}$, and $\varepsilon,\delta\in(0,1)$.
For every $U\in\mathrm U(M)$, there is a randomized construction of a real signed function $\widehat q_{d,U}$ on $[M]_{\neq}^{n}$ such that
$\E_{\mathrm{col}}[\widehat q_{d,U}(\bm r)]=q_{d,U}(\bm r)$ for every collision-free string $\bm r$.
For each fixed collection of colorings, after $MN^2\varepsilon^{-2}\delta^{-1}2^{O(d)}$ arithmetic operations of preprocessing, every collision-free prefix marginal of $\widehat q_{d,U}$ can be evaluated exactly in $N^3\varepsilon^{-2}\delta^{-1}2^{O(d)}$ arithmetic operations.
Moreover, if $M\geq C_0n^2N$, then
\begin{align}
 \E_{U,\mathrm{col}}\!\left[\widehat\Delta_{\mathrm{cf}}(U)^2\right]
 &\leq C_HR_d(n,N)+\frac{\varepsilon^2\delta}{8},
 \label{eq:randomized-collision-free-error}
\end{align}
where $U\sim\Haar(M)$ and $C_0,C_H$ are the constants in Lemma~\ref{lem:rectangular-domination}.
\end{proposition}


\begin{proof}[Proof sketch]

Assume first that $1\leq d\leq n$.
We apply color coding~\cite{alon1995color, alon2008biomolecular} to the centered expansion of $q_{d,U}$ obtained from Lemma~\ref{lem:centered-expansion}.
Specifically, for each $1\leq j\leq d$, we independently assign a uniform color from $[j]$ to every input column and physical output mode, and retain only the terms that select columns from different colors and output modes of different colors.
Dividing each retained term by its inclusion probability gives a signed estimate with expectation $q_{d,U}$.
We define $\widehat q_{d,U}$ as the average of the signed estimates from $L$ independent colorings.
By linearity, it follows that $\E_{\mathrm{col}}[\widehat q_{d,U}(\bm r)]=q_{d,U}(\bm r)$ for every collision-free string $\bm r$.

For a fixed coloring, the marginal formula uses sums over output modes of products involving two input columns.
The two column indices range independently over $[N]$, giving $N^2$ ordered pairs.
At each retained order, accumulating the contributions of the $M$ output modes for these pairs, separately by output color, costs $O(MN^2)$.
For a prefix query, we adjust the contributions of at most $n$ prefix modes, each requiring updates for $N^2$ pairs, at cost $O(nN^2)$.
The remaining calculation uses $N\times N$ matrix multiplications, giving the $N^3$ factor.
To exclude repeated colors, we keep separate contributions according to the sets of input and output colors already used.
With at most $d$ colors of each type, there are at most $2^d\times2^d=4^d$ such pairs of sets, and storing and combining their contributions incurs a $2^{O(d)}$ overhead~\cite{koutis2008faster}.
Summing over the retained orders and using $n\leq N$ gives $N^3 2^{O(d)}$ arithmetic operations per query after $MN^2 2^{O(d)}$ preprocessing.

We finally bound the coloring error.
For each fixed collision-free string, the orthogonality of the centered products and the inclusion probabilities give a mean-square coloring error of at most $M^{-2n}2^{O(d)}$, where the expectation is taken over the Gaussian matrix and the coloring.
Averaging $L$ independent estimates divides this bound by $L$.
Moreover, pointwise unbiasedness removes the cross term with the truncation error after averaging over the colorings.
Applying the Haar--Gaussian comparison and the collision-free Cauchy--Schwarz estimate of Sec.~\ref{sec:haar-error} then gives
\begin{align}
\E_{U,\mathrm{col}}\!\left[\widehat\Delta_{\mathrm{cf}}(U)^2\right]
&\leq C_HR_d(n,N)+\frac{C_H2^{O(d)}}{L}.
\label{eq:coloring-average-budget-sketch}
\end{align}
Choosing $L=\varepsilon^{-2}\delta^{-1}2^{O(d)}$ makes the second term at most $\varepsilon^2\delta/8$, which proves Eq.~\eqref{eq:randomized-collision-free-error}.
Since each marginal of $\widehat q_{d,U}$ is the average of the exact marginals for the same $L$ fixed colorings, multiplying both costs by $L$ gives the stated complexity bounds.
For $d=0$, no coloring is needed, and the constant approximation has directly computable marginals.
Appendix~\ref{app:randomized-marginals} gives the full construction and proof.

\end{proof}

\subsection{Collision collapse}

Proposition~\ref{prop:randomized-marginal-complexity} gives efficiently computable collision-free marginals of $\widehat q_{d,U}$, but its total collision-free mass need not equal one.
To apply the signed marginal sampling lemma, we extend $\widehat q_{d,U}$ to a normalized signed function on an enlarged ordered-string space.
We retain the collision-free values of $\widehat q_{d,U}$ and assign the remaining signed mass to one dummy outcome representing the entire collision sector of $p_U$.
Throughout this subsection, fix $U$ and the full collection of colorings used to construct $\widehat q_{d,U}$.

Concretely, to retain the product-space structure needed for sequential sampling, we adjoin a symbol $0\notin[M]$ and represent the dummy outcome by the string $0^n$.
Define the collision-collapsed signed function $\widetilde q_{d,U}$ on $\{0,1,\ldots,M\}^n$ by
\begin{align}
 \widetilde q_{d,U}(\bm r)
 &:={}
 \begin{cases}
 \widehat q_{d,U}(\bm r),&\bm r\in[M]_{\neq}^{n},\\[1mm]
 1-\widehat m_{d,U}(\varnothing),&\bm r=0^n,\\[1mm]
 0,&\text{otherwise}.
 \end{cases}
 \label{eq:q-tilde-definition}
\end{align}
The corresponding ideal collision-collapsed ordered-string distribution is
\begin{align}
 \widetilde p_U(\bm r)
 &:={}
 \begin{cases}
 p_U(\bm r),&\bm r\in[M]_{\neq}^{n},\\[1mm]
 p_{\mathrm{coll}}(U),&\bm r=0^n,\\[1mm]
 0,&\text{otherwise},
 \end{cases}
 \label{eq:p-tilde-definition}
\end{align}
where
\begin{align}
 p_{\mathrm{coll}}(U)
 &:={}
 \sum_{\bm r\in[M]^n\setminus[M]_{\neq}^{n}}p_U(\bm r).
 \label{eq:collision-probability}
\end{align}
Note that $\widetilde q_{d,U}$ has total mass one but may take negative values, including at the dummy outcome.
The following lemma shows that normalization preserves the efficient evaluation of exact prefix marginals.

\begin{lemma}[Normalization and prefix marginals]
\label{lem:q-tilde-marginals}
For each fixed collection of colorings, $\widetilde q_{d,U}$ in Eq.~\eqref{eq:q-tilde-definition} is a normalized signed function.
After $MN^2\varepsilon^{-2}\delta^{-1}2^{O(d)}$ arithmetic operations of preprocessing, every prefix marginal of $\widetilde q_{d,U}$ can be evaluated exactly in $N^3\varepsilon^{-2}\delta^{-1}2^{O(d)}$ arithmetic operations.
\end{lemma}

\begin{proof}
Since $\widehat m_{d,U}(\varnothing)$ is the total collision-free signed mass, Eq.~\eqref{eq:q-tilde-definition} gives
\begin{align}
 \sum_{\bm r}\widetilde q_{d,U}(\bm r)
 &=\widehat m_{d,U}(\varnothing)
   +1-\widehat m_{d,U}(\varnothing)=1.
 \label{eq:q-tilde-normalization}
\end{align}
We evaluate $\widehat m_{d,U}(\varnothing)$ once using Proposition~\ref{prop:randomized-marginal-complexity} and reuse it, with its cost absorbed into preprocessing because $N\leq M$.

The empty-prefix marginal of $\widetilde q_{d,U}$ is one.
For $1\leq\ell\leq n$, a prefix $0^\ell$ has marginal $1-\widehat m_{d,U}(\varnothing)$, while a collision-free nonzero prefix $\mathbf a$ has marginal $\widehat m_{d,U}(\mathbf a)$.
Every other prefix has marginal zero: a prefix containing both zero and a nonzero symbol cannot extend to $0^n$, and a prefix with repeated nonzero symbols has no collision-free completion.
Thus the stored residual mass and the exact prefix evaluations from Proposition~\ref{prop:randomized-marginal-complexity} give all required marginals within the stated cost.
\end{proof}

We next bound the error of $\widetilde q_{d,U}$ relative to $\widetilde p_U$.
Because both functions have total mass one and are supported on the collision-free strings together with $0^n$, the error at the dummy outcome is controlled by their difference on the collision-free sector.

\begin{lemma}[Error after collision collapse]
\label{lem:collapse-error}
For every $U\in\mathrm U(M)$ and every fixed collection of colorings,
\begin{align}
 \norm{\widetilde p_U-\widetilde q_{d,U}}_1
 &\leq 2\widehat\Delta_{\mathrm{cf}}(U).
 \label{eq:collapse-error}
\end{align}
\end{lemma}

\begin{proof}
On collision-free strings, the $L^1$ error is exactly $\widehat\Delta_{\mathrm{cf}}(U)$.
By Eqs.~\eqref{eq:q-tilde-definition} and~\eqref{eq:p-tilde-definition}, the remaining contribution satisfies
\begin{align}
 \abs{\widetilde p_U(0^n)-\widetilde q_{d,U}(0^n)}
 &=
 \abs{\sum_{\bm r\in[M]_{\neq}^{n}}
\bigl(\widehat q_{d,U}(\bm r)-p_U(\bm r)\bigr)}
 \nonumber\\
 &\leq \widehat\Delta_{\mathrm{cf}}(U).
 \label{eq:dummy-error}
\end{align}
Adding the collision-free and dummy contributions proves Eq.~\eqref{eq:collapse-error}.
\end{proof}

Finally, we relate this ordered-string construction to the collision-collapsed Fock space used in Theorem~\ref{thm:main}.
Analogously to the collision-collapsing map $\mathsf C$ in Eq.~\eqref{eq:collision-collapse-definition}, define
\begin{align}
 \widehat{\mathsf C}:
 \{0,1,\ldots,M\}^n
 &\longrightarrow \mathcal T_{M,n}^{\mathrm{cf}}\cup\{c\}
\end{align}
by
\begin{align}
 \widehat{\mathsf C}(\bm r)
 &:={}
 \begin{cases}
 \sort(\bm r),&\bm r\in[M]_{\neq}^{n},\\
 c,&\text{otherwise},
 \end{cases}
 \label{eq:ordered-to-collapsed-fock-map}
\end{align}
where $\sort$ is the sorting map in Eq.~\eqref{eq:sort-definition} and $c$ is the dummy collision symbol introduced in Sec.~\ref{sec:preliminaries}.
The pushforward of $\widetilde p_U$ under $\widehat{\mathsf C}$ is then exactly the collision-collapsed target distribution:
\begin{align}
 (\widehat{\mathsf C})_{\#}\widetilde p_U
 &=P_{U,\mathrm{coll}}^{N,n}.
 \label{eq:ptilde-pushes-to-collapsed-target}
\end{align}
In Sec.~\ref{sec:sampling}, we apply the signed marginal sampling lemma to $\widetilde q_{d,U}$ and use this map to obtain the final sampler on collision-collapsed Fock outcomes.


\section{Classical sampling algorithm and proof of the main theorem}
\label{sec:sampling}

This section completes the construction of the sampling algorithm and proves Theorem~\ref{thm:main}.
We first use the exact prefix marginals of $\widetilde q_{d,U}$ in Eq.~\eqref{eq:q-tilde-definition} to construct a proper sampler, and then map its output to the collision-collapsed Fock space.
The error and marginal-evaluation bounds in Sec.~\ref{sec:marginals} then give the claimed accuracy and runtime guarantees.

\subsection{Signed marginal sampling}

We begin with the signed marginal sampling lemma of Ref.~\cite{bremner2017achieving}.
The procedure sequentially constructs a proper distribution from the prefix marginals of a normalized signed function.
At each step, it evaluates the marginals associated with all one-symbol extensions of the current prefix, truncates negative values to zero, and normalizes the remaining nonnegative weights to sample the next symbol.

\begin{lemma}[Signed marginal sampler~\cite{bremner2017achieving, cirstoiu2024fourier}]
\label{lem:signed-marginal-sampler}
Let $p$ be a probability distribution on a finite product space $\mathcal X^n$, and let $q:\mathcal X^n\rightarrow\mathbb R$ be a normalized signed function such that $\sum_{\mathbf x\in\mathcal X^n}q(\mathbf x)= 1$.
Assume that all prefix marginals of $q$ can be computed by an oracle.  If
\begin{align}
 \norm{p-q}_1
 &\leq
 \xi,
\end{align}
then there is a randomized classical algorithm using $O(n\abs{\mathcal X})$ calls to this marginal oracle that samples from a proper probability distribution $q^{\mathrm{samp}}$ satisfying
\begin{align}
 \norm{q^{\mathrm{samp}}-p}_1
 &\leq
 2\xi.
 \label{eq:signed-sampler-l1-bound}
\end{align}
\end{lemma}

We apply this lemma to $\widetilde p_U$ and $\widetilde q_{d,U}$ from Sec.~\ref{sec:marginals}.
For this application, we fix $U$ and the full collection of colorings used to construct $\widehat q_{d,U}$, so every queried prefix marginal belongs to the same signed function $\widetilde q_{d,U}$.
We denote the resulting ordered-string output distribution, conditional on this collection, by $q_{d,U}^{\mathrm{samp}}$.

\begin{lemma}[Proper sampler for the collision-collapsed distribution]
\label{lem:proper-collapsed-sampler}
For every fixed unitary $U\in\mathrm U(M)$ and fixed collection of colorings, the signed marginal sampler uses $O(Mn)$ prefix-marginal queries to $\widetilde q_{d,U}$ and samples from a proper probability distribution $q_{d,U}^{\mathrm{samp}}$ on $\{0,1,\ldots,M\}^n$ satisfying
\begin{align}
 d_{\mathrm{TV}}\!\left(q_{d,U}^{\mathrm{samp}},\widetilde p_U\right)
 &\leq 2\widehat\Delta_{\mathrm{cf}}(U).
 \label{eq:proper-collapsed-sampler-error}
\end{align}
\end{lemma}

\begin{proof}
By Lemma~\ref{lem:q-tilde-marginals}, $\widetilde q_{d,U}$ is normalized and has exactly computable prefix marginals, while Lemma~\ref{lem:collapse-error} gives $\norm{\widetilde p_U-\widetilde q_{d,U}}_1\leq2\widehat\Delta_{\mathrm{cf}}(U)$.
Applying Lemma~\ref{lem:signed-marginal-sampler} with $\xi=2\widehat\Delta_{\mathrm{cf}}(U)$ therefore gives
\begin{align}
 \norm{q_{d,U}^{\mathrm{samp}}-\widetilde p_U}_1
 &\leq4\widehat\Delta_{\mathrm{cf}}(U).
\end{align}
Dividing by two gives Eq.~\eqref{eq:proper-collapsed-sampler-error}.
Since the alphabet is $\{0,1,\ldots,M\}$, the number of marginal queries is $O(n(M+1))=O(Mn)$.
\end{proof}

\subsection{Returning to collision-collapsed Fock outcomes}

We now apply the map $\widehat{\mathsf C}$ in Eq.~\eqref{eq:ordered-to-collapsed-fock-map} to convert the ordered-string samples into collision-collapsed Fock outcomes.
For each fixed collection of colorings, the resulting output distribution is $(\widehat{\mathsf C})_{\#}q_{d,U}^{\mathrm{samp}}$.
The algorithm draws this collection at random when constructing $\widehat q_{d,U}$, so its output distribution is
\begin{align}
 \mathcal A_d(U)
 &:={}
 \E_{\mathrm{col}}\!\left[
 (\widehat{\mathsf C})_{\#}q_{d,U}^{\mathrm{samp}}\right].
 \label{eq:final-collapsed-sampler-distribution}
\end{align}
This is the distribution $\mathcal A_d(U)$ appearing in Theorem~\ref{thm:main}.
The following lemma transfers the conditional sampling bound to this distribution.

\begin{lemma}[Final sampler error bound]
\label{lem:final-pushforward-bound}
Let $\mathcal A_d(U)$ be the output distribution in Eq.~\eqref{eq:final-collapsed-sampler-distribution}.
For every fixed unitary $U\in\mathrm U(M)$,
\begin{align}
 d_{\mathrm{TV}}\!\left(
 \mathcal A_d(U),P_{U,\mathrm{coll}}^{N,n}\right)
 &\leq2\E_{\mathrm{col}}\widehat\Delta_{\mathrm{cf}}(U).
 \label{eq:final-pushforward-bound}
\end{align}
\end{lemma}

\begin{proof}
By Eq.~\eqref{eq:ptilde-pushes-to-collapsed-target}, the same map sends $\widetilde p_U$ to $P_{U,\mathrm{coll}}^{N,n}$.
Total variation distance is convex in its first argument and cannot increase under deterministic postprocessing.
Applying these properties and Lemma~\ref{lem:proper-collapsed-sampler} gives
\begin{align}
 d_{\mathrm{TV}}\!\left(\mathcal A_d(U),P_{U,\mathrm{coll}}^{N,n}\right)
 &\leq\E_{\mathrm{col}}\!\left[
 d_{\mathrm{TV}}\!\left(q_{d,U}^{\mathrm{samp}},\widetilde p_U\right)\right]
 \nonumber \\
 &\leq2\E_{\mathrm{col}}\widehat\Delta_{\mathrm{cf}}(U).
\end{align}
\end{proof}

\subsection{The sampling algorithm}

Combining the preceding ingredients gives Algorithm~\ref{alg:fixed-loss-sampler}, which generates one output sample for given $U,M,N,n,\varepsilon$, and $\delta$ satisfying the assumptions of Theorem~\ref{thm:main}.
Each execution draws a fresh collection of colorings and keeps it fixed throughout the marginal queries.
Note that for $d=0$, no such coloring is required.

\begin{algorithm}[H]
\caption{Sampler for fixed-loss collision-collapsed BosonSampling}
\label{alg:fixed-loss-sampler}
\begin{algorithmic}[1]
\STATE Choose the truncation degree $d$ according to Theorem~\ref{thm:main}.
\STATE Draw the full collection of colorings defining $\widehat q_{d,U}$ and perform the preprocessing for exact prefix-marginal evaluation of $\widetilde q_{d,U}$ in Lemma~\ref{lem:q-tilde-marginals}.
\STATE Keeping this collection fixed, generate an ordered string $\bm r$ one symbol at a time using these prefix marginals and the signed marginal sampler in Lemma~\ref{lem:proper-collapsed-sampler}.
\STATE Return $\sort(\bm r)$ if $\bm r\in[M]_{\neq}^{n}$, and the collision symbol $c$ otherwise.
\end{algorithmic}
\end{algorithm}

\subsection{Proof of the main theorem}

We now combine the final sampler bound with Proposition~\ref{prop:randomized-marginal-complexity} to prove the accuracy guarantee, and then account for the cost of generating one sample.

\begin{proof}[Proof of Theorem~\ref{thm:main}]
Squaring the bound in Lemma~\ref{lem:final-pushforward-bound} and taking expectation over $U$ gives
\begin{align}
 \E_U\!\left[
 d_{\mathrm{TV}}\!\left(\mathcal A_d(U),P_{U,\mathrm{coll}}^{N,n}\right)^2
 \right]
 &\leq4\E_U\!\left[
 \left(\E_{\mathrm{col}}\widehat\Delta_{\mathrm{cf}}(U)\right)^2
 \right]
 \nonumber\\
 &\leq4\E_{U,\mathrm{col}}\!\left[\widehat\Delta_{\mathrm{cf}}(U)^2\right]
 \nonumber\\
 &\leq4C_HR_d(n,N)+\frac{\varepsilon^2\delta}{2}.
 \label{eq:averaged-randomized-error}
\end{align}
Here, the second inequality is Jensen's inequality over the colorings, and the last follows from Proposition~\ref{prop:randomized-marginal-complexity} under $M\geq C_0n^2N$.

The coloring contribution in Eq.~\eqref{eq:averaged-randomized-error} is at most $\varepsilon^2\delta/2$.
To bound the truncation contribution by the same amount, set $C=8C_H$ in Theorem~\ref{thm:main}.
If $d<n$, the degree condition in Eq.~\eqref{eq:main-degree-condition} and Proposition~\ref{prop:gaussian-truncation} give
\begin{align}
 R_d(n,N)&\leq\frac{\varepsilon^2\delta}{8C_H}.
 \label{eq:Rd-choice-main-proof}
\end{align}
For $d=n$, the truncation is exact and $R_d(n,N)=0$, so the same bound holds.
Hence, the second moment in Eq.~\eqref{eq:averaged-randomized-error} is at most $\varepsilon^2\delta$.
Applying Markov's inequality to the squared total variation distance gives
\begin{align}
 &\Prob_{U\sim\Haar(M)}\!\left[
 d_{\mathrm{TV}}\!\left(\mathcal A_d(U),P_{U,\mathrm{coll}}^{N,n}\right)>\varepsilon
 \right]
 \leq\delta.
 \label{eq:main-tv-completion}
\end{align}
This proves Eq.~\eqref{eq:main-average-case}.
The probability is over $U$ alone, since $\mathcal A_d(U)$ already includes the coloring randomness.

It remains to verify the runtime.
By Lemma~\ref{lem:q-tilde-marginals}, preprocessing the marginal oracle for $\widetilde q_{d,U}$ costs $MN^2\varepsilon^{-2}\delta^{-1}2^{O(d)}$ arithmetic operations, including the evaluation of the remaining signed mass.
Each prefix marginal then costs $N^3\varepsilon^{-2}\delta^{-1}2^{O(d)}$ arithmetic operations.
Lemma~\ref{lem:proper-collapsed-sampler} requires $O(Mn)$ such queries.
Combining the preprocessing and query costs gives $MnN^3\varepsilon^{-2}\delta^{-1}2^{O(d)}$, where the query cost dominates preprocessing since $nN\geq1$.
The costs of the sequential choices and the final sorting map are also absorbed into this bound.
Thus the total arithmetic runtime per output sample is as stated in Theorem~\ref{thm:main}.
\end{proof}

\section{Full Fock-output approximation under the bosonic birthday bound}
\label{sec:full-fock}

We now extend the guarantee to the full Fock-output distribution.
To do so, we map the dummy collision outcome $c$ to one fixed Fock outcome and leave collision-free output labels unchanged.
The additional error is bounded by the total collision probability, which can be controlled by the bosonic birthday bound under an additional condition on the number of modes.

We first recall the bosonic birthday bound established in Refs.~\cite{aaronson2011computational, arkhipov2012bosonic}.

\begin{lemma}[Bosonic birthday bound~\cite{aaronson2011computational, arkhipov2012bosonic}]
\label{lem:bosonic-birthday-bound}
Assume that $M>2n^2$.
For fixed-loss BosonSampling, the expected total probability of collision outcomes over Haar-random interferometers satisfies
\begin{align}
 \E_{U\sim\Haar(M)}\!\left[p_{\mathrm{coll}}(U)\right]
 &\leq\frac{2n^2}{M}.
 \label{eq:expected-collision-probability}
\end{align}
Consequently, for every $\varepsilon>0$,
\begin{align}
 \Prob_{U\sim\Haar(M)}\!\left[
 p_{\mathrm{coll}}(U)>\frac{\varepsilon}{2}
 \right]
 &\leq\frac{4n^2}{M\varepsilon}.
 \label{eq:collision-markov}
\end{align}
\end{lemma}

\begin{proof}
For each fixed collision-free $n$-photon input subset $S\subseteq[N]$, the bosonic birthday bound gives $\E_U[p_{\mathrm{coll}}(U\mid S)]\leq2n^2/M$ when $M>2n^2$~\cite{aaronson2011computational, arkhipov2012bosonic}.
The fixed-loss input is the uniform mixture over these subsets.
Averaging the bound over $S$ gives Eq.~\eqref{eq:expected-collision-probability}, and Markov's inequality then gives Eq.~\eqref{eq:collision-markov}.
\end{proof}

Combining this bound with the collision-collapsed sampler gives the following approximation guarantee for the full Fock-output distribution $P_U^{N,n}$ in Eq.~\eqref{eq:fock-fixed-loss}.

\begin{corollary}[Full fixed-loss Fock-output approximation]
\label{cor:full-fock-output}
Let $C_0,C$ be the constants in Theorem~\ref{thm:main}.
Let $1\leq n<N$ and $\varepsilon,\delta\in(0,1)$, and suppose
$M\geq C_0n^2N$ and $M\geq C_Bn^2/(\varepsilon\delta)$ for a sufficiently large universal constant $C_B$.
Choose $d\in\{0,\ldots,n\}$ so that either $d=n$ or
\begin{align}
 \frac{(n/N)^{d+1}}{1-n/N}
 &\leq\frac{\varepsilon^2\delta}{8C}.
 \label{eq:full-output-degree-condition}
\end{align}
Then there exists a randomized classical algorithm that samples from a distribution $Q_U$ on the full Fock-output space $\mathcal T_{M,n}$ satisfying
\begin{align}
 \Prob_{U\sim\Haar(M)}\!\left[
 d_{\mathrm{TV}}\!\left(Q_U,P_U^{N,n}\right)\leq\varepsilon
 \right]
 &\geq1-\delta,
 \label{eq:full-fock-output-guarantee}
\end{align}
with runtime $MnN^3\varepsilon^{-2}\delta^{-1}2^{O(d)}$ arithmetic operations per output sample.
\end{corollary}

\begin{proof}
We first run the fixed-loss sampling algorithm of Sec.~\ref{sec:sampling} with degree $d$, accuracy parameter $\varepsilon/2$, and failure parameter $\delta/2$.
We then choose any fixed Fock outcome $\bm t_0\in\mathcal T_{M,n}$, and map the dummy outcome $c$ to $\bm t_0$ while leaving collision-free output labels unchanged.
Let $Q_U$ be the resulting output distribution, and let $\widehat P_U$ be the distribution obtained by applying the same map to $P_{U,\mathrm{coll}}^{N,n}$.
Since total variation distance cannot increase under this map,
\begin{align}
 d_{\mathrm{TV}}(Q_U,\widehat P_U)
 &\leq d_{\mathrm{TV}}\!\left(\mathcal A_d(U),P_{U,\mathrm{coll}}^{N,n}\right).
 \label{eq:full-corollary-first-term}
\end{align}
Moreover, moving the ideal collision mass to $\bm t_0$ gives
\begin{align}
 d_{\mathrm{TV}}(\widehat P_U,P_U^{N,n})
 &\leq p_{\mathrm{coll}}(U).
 \label{eq:full-corollary-collision-term}
\end{align}
The triangle inequality therefore bounds the full Fock-output error by the sum of the collision-collapsed sampling error and the ideal collision probability.

Specifically, note that Eq.~\eqref{eq:full-output-degree-condition} is the degree condition of Theorem~\ref{thm:main} with parameters $\varepsilon/2$ and $\delta/2$.
For $d=n$, the truncation error is zero.
Thus the proof of Theorem~\ref{thm:main} gives
\begin{align}
 \Prob_{U\sim\Haar(M)}\!\left[
 d_{\mathrm{TV}}\!\left(\mathcal A_d(U),P_{U,\mathrm{coll}}^{N,n}\right)
 >\frac{\varepsilon}{2}
 \right]
 &\leq\frac{\delta}{2}.
 \label{eq:full-corollary-collapsed-good}
\end{align}
By Lemma~\ref{lem:bosonic-birthday-bound}, choosing $C_B\geq8$ also gives
\begin{align}
 \Prob_U\!\left[p_{\mathrm{coll}}(U)>\frac{\varepsilon}{2}\right]
 &\leq\frac{\delta}{2}.
 \label{eq:full-corollary-birthday-good}
\end{align}
Outside these two exceptional events, Eqs.~\eqref{eq:full-corollary-first-term} and~\eqref{eq:full-corollary-collision-term} imply
\begin{align}
 d_{\mathrm{TV}}(Q_U,P_U^{N,n})
 &\leq d_{\mathrm{TV}}(Q_U,\widehat P_U)
       +d_{\mathrm{TV}}(\widehat P_U,P_U^{N,n})
 \nonumber\\
 &\leq\frac{\varepsilon}{2}+\frac{\varepsilon}{2}=\varepsilon.
\end{align}
Therefore, a union bound proves Eq.~\eqref{eq:full-fock-output-guarantee}.
Rescaling the accuracy and failure parameters changes the sampling cost only by a constant factor, and the final relabeling is absorbed into the stated runtime.
\end{proof}

\section{Lossy BosonSampling with uniform transmission loss}
\label{sec:extensions}

This section extends our fixed-loss simulation result to the uniform-transmission-loss model, in which each photon is transmitted independently with the same transmission rate $\eta\in(0,1)$.
We obtain the extension by conditioning on the number of surviving photons and applying the fixed-loss sampler within each sector.

Let $S_\eta\sim\operatorname{Bin}(N,\eta)$ be the number of surviving photons, and let $s$ denote its realized value.
Conditioned on $S_\eta=s$, all surviving input subsets of size $s$ are equally likely, each with probability $1/\binom Ns$.
Consequently, the conditional output distribution is the fixed-loss distribution with $s$ surviving photons.
We use its collision-collapsed version $P_{U,\mathrm{coll}}^{N,s}$ as the target within each sector, with $P_{U,\mathrm{coll}}^{N,0}$ denoting the point mass on the vacuum.

To distinguish outcomes from different sectors, let $c_s$ be the dummy collision symbol in the $s$-photon sector and define
\begin{align}
 \mathcal Y_{M,s}^{\mathrm{coll}}
 &:={}\mathcal T_{M,s}^{\mathrm{cf}}\cup\{c_s\},
 \qquad 1\leq s\leq N.
 \label{eq:uniform-sector-outcome-spaces}
\end{align}
For $s=0$, set $\mathcal Y_{M,0}^{\mathrm{coll}}:=\{\mathrm{vac}\}$.
The common collision-collapsed outcome space is
\begin{align}
 \mathcal Y_{M,\leq N}^{\mathrm{coll}}
 &:={}\bigsqcup_{s=0}^{N}\mathcal Y_{M,s}^{\mathrm{coll}}.
 \label{eq:uniform-common-outcome-space}
\end{align}
We extend each sector distribution by zero outside its sector, using $c_s$ as the collision symbol for $s\geq1$.
The collision-collapsed target for uniform transmission loss is then the binomial mixture
\begin{align}
 P_{U,\eta,\mathrm{coll}}^{N}
 &:=\sum_{s=0}^{N}\Prob[S_\eta=s]P_{U,\mathrm{coll}}^{N,s}.
 \label{eq:collapsed-binomial-sector-decomposition}
\end{align}

To control the runtime, we apply the fixed-loss sampler only when the surviving photon number is at most a cutoff $s_*$.
For larger photon numbers, we return the corresponding dummy collision symbol.
These sectors contribute at most their total binomial weight $\Prob[S_\eta>s_*]$ to the approximation error.

\begin{theorem}[Uniform transmission loss via fixed-survivor-number sectors]
\label{thm:uniform-transmission}
Let $C_0,C$ be the constants in Theorem~\ref{thm:main}.
Let $N\geq2$, $\eta,\varepsilon,\delta\in(0,1)$, $S_\eta\sim\operatorname{Bin}(N,\eta)$, and $s_*\in\{1,\ldots,N-1\}$.
Suppose $M\geq C_0(s_*)^2N$ and
\begin{align}
 \Prob[S_\eta>s_*]&\leq\frac{\varepsilon}{2}.
\end{align}
Choose any degree $d_*\in\{0,\ldots,s_*-1\}$ satisfying
\begin{align}
 \frac{(s_*/N)^{d_*+1}}{1-s_*/N}
 &\leq\frac{\varepsilon^2\delta}{4C}
 \label{eq:uniform-degree-condition}
\end{align}
if such a value exists; otherwise, set $d_*=s_*$.
Then there exists a randomized classical algorithm that runs in $Ms_*N^3\varepsilon^{-2}\delta^{-1}2^{O(d_*)}$ arithmetic operations per output sample and samples from a distribution $\mathcal A_\eta(U)$ on $\mathcal Y_{M,\leq N}^{\mathrm{coll}}$ satisfying
\begin{align}
 \Prob_{U\sim\Haar(M)}\!\left[
 d_{\mathrm{TV}}\!\left(\mathcal A_\eta(U),P_{U,\eta,\mathrm{coll}}^{N}\right)
 \leq\varepsilon
 \right]
 &\geq1-\delta.
 \label{eq:uniform-average-case}
\end{align}
\end{theorem}

\begin{proof}
The algorithm first samples $s\sim\operatorname{Bin}(N,\eta)$ using $N$ independent Bernoulli trials.
If $s=0$, it outputs the vacuum; if $s=1$, it samples the one-photon distribution exactly.
For $2\leq s\leq s_*$, it runs the fixed-loss sampler with $s$ surviving photons, degree $d_s:=\min\{d_*,s\}$, accuracy parameter $\varepsilon/2$, and failure parameter $\delta$.
Denote its output distribution, including the coloring randomness, by $\mathcal A_{d_s}^{(s)}(U)$.
For $s>s_*$, the algorithm returns $c_s$.
This defines $\mathcal A_\eta(U)$ with the same binomial sector weights as the target distribution.

We bound the error of this mixture directly.
For $2\leq s\leq s_*$, write
\begin{align}
 \Delta_s(U)
 &:={}d_{\mathrm{TV}}\!\left(\mathcal A_{d_s}^{(s)}(U),P_{U,\mathrm{coll}}^{N,s}\right).
 \label{eq:uniform-sector-collapsed-error}
\end{align}
The sectors $s=0,1$ are sampled exactly, while the error in each sector above $s_*$ is at most one.
Since the sector outcome spaces are disjoint,
\begin{align}
 d_{\mathrm{TV}}\!\left(\mathcal A_\eta(U),P_{U,\eta,\mathrm{coll}}^{N}\right)
 &\leq\Prob[S_\eta>s_*]+\Delta_\eta(U),
 \label{eq:uniform-deterministic-error}
\end{align}
where
\begin{align}
 \Delta_\eta(U)
 &:={}\sum_{s=2}^{s_*}\Prob[S_\eta=s]\Delta_s(U).
 \label{eq:uniform-weighted-truncation}
\end{align}

For $2\leq s\leq s_*$, $M\geq C_0(s_*)^2N$ implies $M\geq C_0s^2N$.
Applying Eq.~\eqref{eq:averaged-randomized-error} with accuracy parameter $\varepsilon/2$ gives
\begin{align}
 \E_U[\Delta_s(U)^2]
 &\leq4C_HR_{d_s}(s,N)+\frac{\varepsilon^2\delta}{8}.
 \label{eq:uniform-sector-second-moment}
\end{align}
If $d_s=s$, the truncation error $R_{d_s}(s,N)$ is zero.
Otherwise, $d_s=d_*<s$, and
$R_{d_s}(s,N)\leq(s_*/N)^{d_*+1}/(1-s_*/N)$.
Equation~\eqref{eq:uniform-degree-condition}, with $C=8C_H$ as in the proof of Theorem~\ref{thm:main}, therefore bounds the truncation contribution by $\varepsilon^2\delta/8$.
Thus $\E_U[\Delta_s(U)^2]\leq\varepsilon^2\delta/4$ for every such sector, including the full-degree case.

The weights in Eq.~\eqref{eq:uniform-weighted-truncation} sum to at most one.
Applying Jensen's inequality, with the remaining weight placed at zero, gives
\begin{align}
 \E_U[\Delta_\eta(U)^2]
 &\leq\sum_{s=2}^{s_*}\Prob[S_\eta=s]\E_U[\Delta_s(U)^2]
 \leq\frac{\varepsilon^2\delta}{4}.
 \label{eq:uniform-weighted-second-moment}
\end{align}
Markov's inequality then yields
\begin{align}
 \Prob_U\!\left[\Delta_\eta(U)>\frac{\varepsilon}{2}\right]
 &\leq\delta.
 \label{eq:uniform-truncation-good-event}
\end{align}
Together with the binomial-tail assumption and Eq.~\eqref{eq:uniform-deterministic-error}, this proves Eq.~\eqref{eq:uniform-average-case}.

Finally, each execution samples only one sector.
For $2\leq s\leq s_*$, the fixed-loss cost is bounded by $Ms_*N^3\varepsilon^{-2}\delta^{-1}2^{O(d_*)}$ because $s\leq s_*$ and $d_s\leq d_*$.
The exact vacuum and one-photon cases, the dummy output, and the $O(N)$ Bernoulli trials are also absorbed into this bound.
\end{proof}

For a transmission rate bounded away from one, a cutoff proportional to $N$ makes the binomial tail exponentially small.
The following corollary uses this observation to obtain polynomial runtime for fixed and decreasing transmission rates.

\begin{corollary}[Polynomial-time simulation under uniform transmission loss]
\label{cor:constant-transmission}
Let $\varepsilon,\delta\in(0,1)$ and suppose $M\geq C_0N^3$.
For any fixed transmission rate $\eta\in(0,1)$, collision-collapsed uniform-transmission-loss BosonSampling admits a classical simulation satisfying Eq.~\eqref{eq:uniform-average-case} with runtime $\operatorname{poly}(M,N,\varepsilon^{-1},\delta^{-1})$.
\end{corollary}

\begin{proof}
For a given $\eta\in(0,1)$, choose a constant $\rho\in(\eta,1)$ and set $s_*:=\lfloor\rho N\rfloor$.
A Chernoff bound gives a constant $c_\eta>0$ such that
$\Prob[S_\eta>s_*]\leq e^{-c_\eta N}$ for sufficiently large $N$.
This cutoff therefore satisfies the tail condition of Theorem~\ref{thm:uniform-transmission} whenever $\varepsilon\geq2e^{-c_\eta N}$.
Since $s_*/N\leq\rho<1$, choosing the smallest degree allowed by the theorem, including the full-degree fallback, gives
$d_*=O(1+\log(\varepsilon^{-1}\delta^{-1}))$.
The runtime is consequently polynomial in $M,N,\varepsilon^{-1}$, and $\delta^{-1}$, with the polynomial degree allowed to depend on the fixed transmission rate.

For the complementary regime $\varepsilon<2e^{-c_\eta N}$, we sample $s$ without imposing a cutoff and use the full degree $d=s$ in every sector with $s\geq2$, with accuracy parameter $\varepsilon/2$ and failure parameter $\delta$ as above.
Proposition~\ref{prop:randomized-marginal-complexity} also permits $s=N$, and $M\geq C_0N^3$ ensures the Haar--Gaussian comparison condition in every sector.
All truncation errors then vanish, while the coloring errors remain controlled by the chosen accuracy and failure parameters.
The same weighted second-moment argument therefore gives the required guarantee without any binomial-tail contribution.
The total cost is at most $MN^4\varepsilon^{-2}\delta^{-1}2^{O(N)}$.
Since $\varepsilon<2e^{-c_\eta N}$ implies $N=O(1+\log\varepsilon^{-1})$, this cost is in $\operatorname{poly}(M,N,\varepsilon^{-1},\delta^{-1})$.

\end{proof}

We now conclude with the analogous extension to the full Fock-output distribution.

\begin{remark}[Full Fock-output approximation under uniform transmission loss]
Although the simulation guarantee in Theorem~\ref{thm:uniform-transmission} is collision-collapsed, we can achieve an analogous full Fock-output approximation guarantee by mapping each dummy collision symbol $c_s$ to a fixed Fock outcome in its sector, as in Corollary~\ref{cor:full-fock-output}.
In addition to $M = \Omega((s_*)^2N)$, a sufficient collision-suppression condition is
\begin{align}
 M
 =
 \Omega\left(\frac{\E[S_\eta^2]}{\varepsilon\delta}\right).
\end{align}
Under this additional condition, the binomially averaged collision probability is small with probability at least $1-\delta$ over $U\sim\Haar(M)$, by applying the bosonic birthday bound conditionally in retained survivor-number sectors and then using Markov's inequality.
\end{remark}

\section{Conclusion}
\label{sec:discussion}

In this work, we developed a polynomial-time classical simulation algorithm for collision-collapsed fixed-loss BosonSampling over Haar-random interferometers in the dilute-mode regime $M=\Omega(n^2N)$.
For every fixed surviving fraction $n/N\in(0,1)$, the algorithm produces samples from a distribution within total variation distance $\varepsilon$ of the target distribution for at least a $1-\delta$ fraction of Haar-random interferometers, with runtime $\operatorname{poly}(M,N,\varepsilon^{-1},\delta^{-1})$.
The central technical ingredients are a Hoeffding decomposition of the fixed-loss output probabilities, which exposes a controlled low-degree approximation, and an efficient randomized color-coding procedure for evaluating the required prefix marginals.
Under the additional bosonic-birthday-bound condition $M=\Omega(n^2/(\varepsilon\delta))$, the simulation guarantee extends from the collision-collapsed distribution to the full Fock-output distribution.
We further generalize the fixed-loss analysis to the uniform-transmission-loss model, establishing analogous polynomial-time classical simulability for every fixed transmission probability $\eta\in(0,1)$.




We conclude with several open directions.
The first is to extend the full Fock-output guarantee to the saturated-mode regime, where collisions need not be rare~\cite{bouland2026complexity, mhiri2026boson}.
Such an extension requires both resolving individual collision outcomes and controlling the approximation error beyond the dilute-mode condition used in the Haar--Gaussian comparison (Lemma~\ref{lem:rectangular-domination}).
This would bring the analysis closer to experimental settings in which the number of modes is comparable to the surviving photon number~\cite{zhong2020quantum, zhong2021phase, madsen2022quantum, deng2023gaussian, liu2026gaussian}.

Beyond relaxing the dilute-mode restriction, another direction is to incorporate additional experimentally relevant conditions into the simulation framework, including partial distinguishability noise~\cite{renema2018efficient, renema2018classical, moylett2019classically, shchesnovich2019noise, renema2020simulability, shi2022effect, van2024efficient} or non-ideal measurements beyond photon-number-resolving detection~\cite{quesada2018gaussian, bressanini2024gaussian, bulmer2022threshold, bradler2021certain, kolarovszki2026generative, kurkin2026universality, go2026threshold}.
Extending the algorithm to these more realistic settings could further strengthen classical benchmarks for photonic quantum-advantage experiments.


\acknowledgements
This work was supported by the Korean government (Ministry of Science and ICT~(MSIT)), the NRF grants funded by the Korea government~(MSIT)~(Nos.~RS-2024-00413957 and RS-2024-00438415),~and the Institute of Information \& Communications Technology Planning \& Evaluation (IITP) grant funded by the Korea government (MSIT) (IITP-2025-RS-2020-II201606 and  IITP-2025-RS-2024-00437191).
Generative AI assisted with analytical cross-checks and manuscript editing.
The authors reviewed and edited all outputs and take full responsibility for the content.

\bibliographystyle{unsrt}
\bibliography{reference.bib}

\onecolumngrid
\appendix

\section{Rectangular Haar--Gaussian density domination (proof of Lemma~\ref{lem:rectangular-domination})}\label{app:haar-details}


We first recall the exact density $p_{\mathrm H}$ of rectangular truncations of Haar unitary matrices.
With respect to Lebesgue measure on $\mathbb C^{n\times N}$, this density, evaluated at $Y\in\mathbb C^{n\times N}$, is given by~\cite{jiang2009approximation}
\begin{align}
p_{\mathrm H}(Y)
=
\pi^{-nN}M^{-nN}
\prod_{j=0}^{n-1}
\frac{(M-j-1)!}{(M-N-j-1)!}
\det\!\left(I_n-\frac{YY^\dagger}{M}\right)^{M-n-N}
\mathbf 1_{\{YY^\dagger \prec MI_n\}} ,
\label{eq:rectangular-truncation-density}
\end{align}
where $\mathbf 1$ denotes the indicator function. 
The formula is valid provided that $n+N\le M$, which is implied by our assumption $M\geq C_0n^2N$ for sufficiently large $C_0$.
Here $YY^\dagger \prec MI_n$ means that all eigenvalues of $YY^\dagger$ are strictly smaller than $M$; the boundary has measure zero and is irrelevant for the density comparison.
Also, the i.i.d.\ standard complex Gaussian density $p_{\mathrm G}$ is given by
\begin{align}
p_{\mathrm G}(Y)
=
\prod_{a=1}^{n}
\prod_{b=1}^{N}
\pi^{-1}e^{-\abs{Y_{a,b}}^2}
=
\pi^{-nN}\exp\!\left[-\Tr(YY^\dagger)\right].
\label{eq:rectangular-gaussian-density}
\end{align}
On the support of $p_{\mathrm H}$, the likelihood ratio is
\begin{align}
\frac{p_{\mathrm H}(Y)}{p_{\mathrm G}(Y)}
&=
M^{-nN}
\prod_{j=0}^{n-1}
\frac{(M-j-1)!}{(M-N-j-1)!}
\exp \left[\Tr (YY^\dagger)\right]
\det\!\left(I_n-\frac{1}{M}YY^\dagger\right)^{M-n-N}.
\label{eq:rectangular-likelihood-ratio}
\end{align}
Here, the normalizing factor can be written as
\begin{align}
M^{-nN}
\prod_{j=0}^{n-1}
\frac{(M-j-1)!}{(M-N-j-1)!}
&=
\prod_{j=0}^{n-1}
\prod_{r=1}^{N}
\left(1-\frac{j+r}{M}\right).
\label{eq:normalizing-ratio-bound}
\end{align}

Let $\lambda_1,\ldots,\lambda_n$ be the eigenvalues of $YY^\dagger$, and set $a:=n+N$.
On the support of the Haar density, $0\leq\lambda_i<M$.
We then obtain
\begin{align}
\log
\frac{p_{\mathrm H}(Y)}{p_{\mathrm G}(Y)}
= \sum_{j=0}^{n-1}
\sum_{r=1}^{N}
\log\!\left(1-\frac{j+r}{M}\right) + 
\sum_{i=1}^{n}
\left[
\lambda_i
+
(M - n - N)\log\!\left(1-\frac{\lambda_i}{M}\right)
\right]. \label{eq:log-ratio-equality}
\end{align}
To bound the latter term, define, for $0\leq\lambda<M$,
\begin{align}
f(\lambda)
&:=
\lambda
+
(M-a)\log\!\left(1-\frac{\lambda}{M}\right)  ,
\label{eq:eigenvalue-function}
\end{align}
whose derivative is given by
\begin{align}
f'(\lambda)
&=
1-\frac{M-a}{M-\lambda}
=
\frac{a-\lambda}{M-\lambda}.
\label{eq:eigenvalue-function-derivative}
\end{align}
Thus, $f$ is maximized at $\lambda=a$ whenever $a<M$, which is guaranteed by choosing $C_0$ sufficiently large in the condition $M\geq C_0 n^2N$.
Then, it holds that
\begin{align}
\log\frac{p_{\mathrm H}(Y)}{p_{\mathrm G}(Y)}
&\leq
\sum_{j=0}^{n-1}
\sum_{r=1}^{N}
\log\!\left(1-\frac{j+r}{M}\right)
+
nf(a)
=
\sum_{k=1}^{\infty}
\frac{1}{kM^k}
\left[
\frac{na^{k+1}}{k+1}
-
\sum_{j=0}^{n-1}
\sum_{r=1}^{N}
(j+r)^k
\right].
\label{eq:likelihood-power-series}
\end{align}
The equality follows from the power-series expansion $\log(1-x)=-\sum_{k=1}^{\infty}x^k/k$, which is absolutely convergent for $|x|<1$. 
Indeed, since $M\geq2a$ and $j+r\leq a-1$, we have
\begin{align}
\sum_{j=0}^{n-1}\sum_{r=1}^{N}
\log\!\left(1-\frac{j+r}{M}\right)
&=
-\sum_{k=1}^{\infty}
\frac{1}{kM^k}
\sum_{j=0}^{n-1}\sum_{r=1}^{N}(j+r)^k,\label{eq:log-expansion}
\end{align}
Moreover, using the same expansion, 
\begin{align}
f(a)
=
a-(M-a)
\sum_{k=1}^{\infty} \frac{a^k}{kM^k}
=
a - \sum_{k=1}^{\infty}\frac{a^k}{kM^{k-1}}
+ \sum_{k=1}^{\infty}\frac{a^{k+1}}{kM^k}
=
\sum_{k=1}^{\infty} \frac{a^{k+1}}{k(k+1)M^k}.
\end{align}
Multiplying the expansion of $f(a)$ by $n$ and combining the two expressions then gives
\begin{align}
\sum_{j=0}^{n-1}\sum_{r=1}^{N}
\log\!\left(1-\frac{j+r}{M}\right)
+nf(a)
=
\sum_{k=1}^{\infty}
\frac{1}{kM^k}
\left[
\frac{na^{k+1}}{k+1}
-
\sum_{j=0}^{n-1}\sum_{r=1}^{N}(j+r)^k
\right].
\end{align}

Next, for every $j\in\{0,\ldots,n-1\}$ and $k\geq1$, monotonicity of $x\mapsto(j+x)^k$ implies 
\begin{align}
\int_{r-1}^{r}(j+x)^k dx
\leq
(j+r)^k
\qquad \Rightarrow \qquad
\sum_{r=1}^{N}(j+r)^k
\geq
\int_0^N(j+x)^k\,dx
=
\frac{(N+j)^{k+1}-j^{k+1}}{k+1}.
\end{align}
Consequently, we have
\begin{align}
\frac{a^{k+1}}{k+1}
-
\sum_{r=1}^{N}(j+r)^k
&\leq
\frac{
a^{k+1}-(N+j)^{k+1}+j^{k+1}
}{k+1}
\leq
(n-j)a^k+\frac{j}{k+1}a^k
\leq
na^k.
\end{align}
In the second inequality, we used the mean-value theorem and
$a-(N+j)=n-j$, together with $j^{k+1}\leq ja^k$.
Summing over $j$ therefore yields
\begin{align}
\frac{na^{k+1}}{k+1}
-
\sum_{j=0}^{n-1}
\sum_{r=1}^{N}(j+r)^k
&\leq
n^2a^k.
\end{align}
Substituting this bound into Eq.~\eqref{eq:likelihood-power-series}, we obtain
\begin{align}
\log\frac{p_{\mathrm H}(Y)}{p_{\mathrm G}(Y)}
&\leq
n^2
\sum_{k=1}^{\infty}
\frac{1}{k}
\left(\frac{a}{M}\right)^k
=
n^2\log\!\left(\frac{1}{1-a/M}\right).
\label{eq:likelihood-log-bound}
\end{align}

Finally, for $M \geq C_0n^2N$ with sufficiently large $C_0$, we have $M\geq2a$, which implies that 
\begin{align}
\log\!\left(\frac{1}{1-a/M}\right)
&\leq
\frac{2a}{M}.
\end{align}
Then it follows that
\begin{align}
\log\frac{p_{\mathrm H}(Y)}{p_{\mathrm G}(Y)}
&\leq
\frac{2n^2(n+N)}{M}
\leq
\frac{4n^2N}{M},
\label{eq:improved-likelihood-log-bound}
\end{align}
where the last inequality uses $n\leq N$.
Under the mode-number condition $M\geq C_0n^2N$, we therefore have
\begin{align}
\frac{p_{\mathrm H}(Y)}{p_{\mathrm G}(Y)}
&\leq
\exp\!\left(\frac{4}{C_0}\right).
\label{eq:likelihood-ratio-final-bound}
\end{align}
After increasing $C_0$ if necessary, the right-hand side is bounded by a universal constant, which we denote by $C_H$.
This proves Lemma~\ref{lem:rectangular-domination}.

\section{Randomized marginal evaluation (proof of Proposition~\ref{prop:randomized-marginal-complexity})}
\label{app:randomized-marginals}

In this appendix, we prove Proposition~\ref{prop:randomized-marginal-complexity}.
We first construct, from one random coloring, a real signed function whose expectation over the coloring is $q_{d,U}$.
For a fixed coloring, we derive its exact prefix marginals in Lemma~\ref{lem:colored-exact-prefix} and bound the corresponding preprocessing and evaluation costs in Lemma~\ref{lem:colored-one-realization-cost}.
We then determine in Lemma~\ref{lem:colored-average-error} how many independent colorings to be averaged are necessary to obtain $\widehat q_{d,U}$ with the desired error bound.
Multiplying the single-coloring preprocessing and evaluation costs by this number of colorings yields the claimed complexity bounds.
The coefficient identity used in the exact prefix formula is proved in the final subsection.

We begin with the centered expansion of $q_{d,U}$ to identify the terms to which color coding will be applied.
Put $V:=\sqrt M\,U_{[M],[N]}$ and $B_{uv}^{(l)}:=V_{u,l}\overline{V_{v,l}}-\delta_{uv}$ for $u,v\in[M]$.
We first assume $1\leq d\leq n$; the case $d=0$ is treated in the proof of Proposition~\ref{prop:randomized-marginal-complexity}.
For nonnegative integers $m,s$, write $(m)_s=m(m-1)\cdots(m-s+1)$, with $(m)_0=1$ and $(m)_s=0$ for $s>m$.
For a collision-free string $\bm r$, Proposition~\ref{prop:exact-decomposition} and Lemma~\ref{lem:centered-expansion} give the following expansion of the target function $q_{d,U}$:
\begin{align}
q_{d,U}(\bm r)
={}&\frac1{M^n}
+\frac1{M^n}\sum_{j=1}^d\frac1{(N)_j}
\sum_{\substack{J\subseteq[N],\,R\subseteq[n]\\\abs{J}=\abs{R}=j}}
\sum_{\sigma,\tau\in\operatorname{Bij}(J,R)}
\prod_{l\in J} B_{r_{\sigma(l)},r_{\tau(l)}}^{(l)}.
\label{eq:uncolored-q-expansion}
\end{align}

Here $J$ selects $j$ input columns, $R$ selects $j$ positions in $\bm r$, and the two bijections specify the row indices in each centered factor.
Thus, for $l\in J$, the values $\sigma(l)$ and $\tau(l)$ are positions in the output string, whereas $r_{\sigma(l)}$ and $r_{\tau(l)}$ are the corresponding physical-mode labels.
Because $\bm r$ is collision-free, the modes $\{r_a:a\in R\}$ are distinct and $\delta_{r_a,r_b}=\delta_{ab}$; hence the algebraic identity in Lemma~\ref{lem:centered-expansion} applies to these rows of $V$.
Each choice of $(J,R,\sigma,\tau)$ determines one product in Eq.~\eqref{eq:uncolored-q-expansion}.
The Hoeffding coefficient from Proposition~\ref{prop:exact-decomposition} and the normalization in Lemma~\ref{lem:centered-expansion} give the common coefficient of these order-$j$ products:
\begin{align}
\frac{\binom nj}{\binom Nj}\frac{(n-j)!}{n!M^n}
&=\frac1{M^n(N)_j}.
\label{eq:colored-coefficient-normalization}
\end{align}

The prefix marginal sums over all collision-free completions, and the physical-mode set $\{r_a:a\in R\}$ ranges over subsets of $[M]$.
Separately enumerating this set and the input set $J$ can therefore involve $\binom Mj\binom Nj$ choices even before summing over the two bijections.
A direct evaluation through order $d$ therefore has a bound of the form $(MN)^{O(d)}$, whose exponent grows with the truncation degree.

To avoid enumerating all selected input and output sets, we use random color coding~\cite{alon1995color, alon2008biomolecular}.
Specifically, for each order $j$, we randomly assign $j$ colors to all input columns and physical output modes, and retain only those products for which the selected input columns have distinct colors and, likewise, the selected physical output modes have distinct colors.
Dividing each retained summand by the probability of satisfying these color conditions preserves its original coefficient in expectation.
The color constraints allow us to evaluate the retained sums in time $\poly(M,N)2^{O(d)}$, at the cost of introducing randomness into the resulting signed approximation to $q_{d,U}$.
We derive exact expressions for its prefix marginals together with their evaluation cost, and then bound the additional randomized error by averaging over independent colorings, yielding the claimed complexity and error guarantees.

\subsection{Constructing the randomized signed function}

For each order $1\leq j\leq d$, assign independently to every input column $l\in[N]$ a uniformly random color $h_j(l)\in[j]$, and to every physical output mode $u\in[M]$ a uniformly random color $g_j(u)\in[j]$.
One draw of the entire family of maps $\{(h_j,g_j)\}_{j=1}^d$ is called a \textit{coloring}.
A set $J$ of $j$ input columns is called $h_j$-\textit{colorful} if $h_j$ is injective on $J$, or equivalently if $h_j(J)=[j]$.
A set of $j$ output modes is called $g_j$-\textit{colorful} under the analogous condition.
The modes $\{r_a:a\in R\}$ are distinct because $\bm r$ is collision-free, so their colors are independent and uniform.
For either selected set, exactly $j!$ of the $j^j$ color assignments are colorful.
Since the maps $h_j$ and $g_j$ are independent, we obtain
\begin{align}
\Prob_{\mathrm{col}}\!\left[
h_j(J)=[j],\ 
g_j(\{r_a:a\in R\})=[j]
\right]
&=\left(\frac{j!}{j^j}\right)^2.
\label{eq:colored-retention-probability}
\end{align}

Each map $g_j:[M]\to[j]$ is drawn once on the physical-mode labels and is subsequently evaluated at the modes appearing in every output string and every prefix query.
For a fixed coloring, retain in Eq.~\eqref{eq:uncolored-q-expansion} the summands whose selected input and output sets are both colorful.
Define the corresponding signed function by
\begin{align}
\widehat q_{d,U}^{(1)}(\bm r)
={}&\frac1{M^n}
+\frac1{M^n}\sum_{j=1}^d
\frac{j^{2j}}{(j!)^2(N)_j}
\sum_{\substack{J\subseteq[N],\,\abs{J}=j\\h_j(J)=[j]}}
\sum_{\substack{R\subseteq[n],\,\abs{R}=j\\
g_j(\{r_a:a\in R\})=[j]}}
\sum_{\sigma,\tau\in\operatorname{Bij}(J,R)}
\prod_{l\in J}
B_{r_{\sigma(l)},r_{\tau(l)}}^{(l)}.
\label{eq:colored-q-one}
\end{align}
The factor $j^{2j}/(j!)^2$ is the reciprocal of the retention probability in Eq.~\eqref{eq:colored-retention-probability}.
Here, $\widehat q_{d,U}^{(1)}$ is real, because complex conjugation exchanges $\sigma$ and $\tau$, and the sum contains both ordered pairs.

For fixed $U,\bm r$, write the restricted sums in Eq.~\eqref{eq:colored-q-one} as unrestricted sums multiplied by the two color indicators.
By linearity of expectation, Eq.~\eqref{eq:colored-retention-probability} gives
\begin{align}
\E_{\mathrm{col}}\!\left[\widehat q_{d,U}^{(1)}(\bm r)\right] 
&=\frac1{M^n}
+\frac1{M^n}\sum_{j=1}^d\frac{j^{2j}}{(j!)^2(N)_j}
\sum_{\substack{J\subseteq[N],\,R\subseteq[n]\\\abs{J}=\abs{R}=j}}
\sum_{\sigma,\tau\in\operatorname{Bij}(J,R)}
\E_{\mathrm{col}}\!\left[
\mathbf 1_{\{h_j(J)=[j]\}}
\mathbf 1_{\{g_j(\{r_a:a\in R\})=[j]\}}
\right]
\prod_{l\in J}B_{r_{\sigma(l)},r_{\tau(l)}}^{(l)}
\nonumber\\
&=\frac1{M^n}
+\frac1{M^n}\sum_{j=1}^d\frac{j^{2j}}{(j!)^2(N)_j}
\left(\frac{j!}{j^j}\right)^2
\sum_{\substack{J\subseteq[N],\,R\subseteq[n]\\\abs{J}=\abs{R}=j}}
\sum_{\sigma,\tau\in\operatorname{Bij}(J,R)}
\prod_{l\in J}B_{r_{\sigma(l)},r_{\tau(l)}}^{(l)}
\nonumber\\
&=q_{d,U}(\bm r).  \label{eq:unbiasedness}
\end{align}

In the next two subsections, we fix one coloring $\{(h_j,g_j)\}_{j=1}^d$ and compute the exact prefix marginals of the corresponding function $\widehat q_{d,U}^{(1)}$ in Eq.~\eqref{eq:colored-q-one}.
The sampler will use $\widehat q_{d,U}$, obtained later by averaging functions generated by independent colorings.
By linearity, its marginal at any prefix is the average of the exact marginals computed for these colorings.
We use the same collection of colorings for every prefix query, so all queried marginals belong to this single function $\widehat q_{d,U}$.

\subsection{Exact prefix marginal for a fixed coloring}

To sum Eq.~\eqref{eq:colored-q-one} over the completions of a prefix, we group its products by the selected physical modes.
The products in each group can then be summed by a coefficient identity, while a completion count determines how often each product occurs.

Fix a collision-free prefix $\mathbf a=(a_1,\ldots,a_\ell)$, where $0\leq\ell\leq n$, and write
$\Omega_{\mathbf a}:=[M]\setminus\{a_1,\ldots,a_\ell\}$ for the physical modes available to the suffix $(r_{\ell+1},\ldots,r_n)$.
For a collision-free string $\bm r=(a_1,\ldots,a_\ell,r_{\ell+1},\ldots,r_n)$ extending the prefix and a set $R\subseteq[n]$ of $j$ selected positions, put $\Omega:=\{r_s:s\in R\}$ (here $r_i = a_i$ when $i \in [\ell])$.
The selected positions may lie in either the prefix or the suffix, so $\Omega$ may contain modes already present in the prefix.
In expressions indexed by $\Omega$, we use $\sigma,\tau:J\to\Omega$ for the bijections obtained by replacing each selected position by its physical mode, so the corresponding product is $\prod_{l\in J}B_{\sigma(l),\tau(l)}^{(l)}$.
For fixed $J,\Omega,\sigma,\tau$, every completion containing $\Omega$ contributes the same product, regardless of the positions of these modes or the other entries.
Hence, summing this product over such completions amounts to multiplying it by their number.

Let $p:=\abs{\Omega\cap\{a_1,\ldots,a_\ell\}}=\abs{\Omega\setminus\Omega_{\mathbf a}}$ be the number of selected modes already present in the prefix.
For fixed $j$, the number of completions containing $\Omega$ depends on this set only through $p$.
Thus, for each $p$, we need the sum of the matrix-dependent products over all colorful $J,\Omega$ and their two bijections; this sum can then be multiplied by their common completion count.

Fix an order $j$ and abbreviate $h=h_j$ and $g=g_j$.
To evaluate the sum required for the marginal, we use coefficient extraction with formal variables that track the selected colors and the number of selected prefix modes.
Specifically, we introduce commuting formal variables $\bm x = (x_1,\ldots,x_j)$ for the output colors and $\bm y  = (y_1,\ldots,y_j)$ for the input colors, together with one factor of $z$ for each selected prefix mode.
We impose $x_c^2=y_c^2=0$ and discard powers of $z$ above degree $j$~\cite{koutis2008faster, feinsilver2011zeons}; these are formal multiplication rules rather than numerical assignments.
These operations preserve the coefficients we seek, since multiplication cannot remove a repeated color variable or decrease the $z$-degree, and therefore introduce no additional approximation error.
For a monomial $w$, we write $[w]F$ for its coefficient in $F$.

The coefficient identity below expresses the required sums for marginal through one inverse determinant.
Its permanent expansion, proved in the final subsection of this appendix, accounts for both the selected sets and the two bijections~\cite{goulden1992immanants, chabaud2022quantum, macmahon2004combinatory}.
Define
\begin{align}
\mathcal P_j(\bm x,\bm y,z)
:={}&
\det\!\left[
\delta_{uv}
-x_{g(u)}z^{\mathbf 1_{\{u\notin\Omega_{\mathbf a}\}}}
\sum_{l=1}^N y_{h(l)}B_{uv}^{(l)}
\right]_{u,v\in[M]}^{-1}.
\label{eq:colored-prefix-generating-polynomial}
\end{align}
Here $u\notin\Omega_{\mathbf a}$ means that $u$ is a fixed prefix mode, so selecting it contributes one power of $z$.
The determinant has the form $1+Q$, where every monomial in $Q$ contains an $x$ variable.
A product of $j+1$ such monomials must repeat one of the $j$ variables $x_1,\ldots,x_j$, so $Q^{j+1}=0$.
The reciprocal is therefore the finite polynomial $\sum_{k=0}^j(-Q)^k$.

For $0\leq p\leq j$, the required coefficient is
\begin{align}
[x_1\cdots x_jy_1\cdots y_jz^p]\mathcal P_j
&={}
\sum_{\substack{\Omega\subseteq[M],\,\abs{\Omega}=j\\
\abs{\Omega\setminus\Omega_{\mathbf a}}=p,\ g(\Omega)=[j]}}
\sum_{\substack{J\subseteq[N],\,\abs{J}=j\\h(J)=[j]}}
\sum_{\sigma,\tau\in\operatorname{Bij}(J,\Omega)}
\prod_{l\in J} B_{\sigma(l),\tau(l)}^{(l)}.
\label{eq:colored-coefficient-identity}
\end{align}
The $x$- and $y$-coefficients select one mode and one column of each color, while $z^p$ selects products with exactly $p$ prefix modes.
The following lemma combines these coefficients with the completion count for each $j,p$ and the prefactors in Eq.~\eqref{eq:colored-q-one} to obtain the exact prefix marginal.

\begin{lemma}[Exact prefix formula]
\label{lem:colored-exact-prefix}
Fix a coloring $\{(h_j,g_j)\}_{j=1}^d$ and a collision-free prefix $\mathbf a=(a_1,\ldots,a_\ell)$.
With $\mathcal P_j$ defined in Eq.~\eqref{eq:colored-prefix-generating-polynomial}, the corresponding prefix marginal is
\begin{align}
&\sum_{\substack{\bm r\in[M]_{\neq}^{n}\\(r_1,\ldots,r_\ell)=\mathbf a}}
\widehat q_{d,U}^{(1)}(\bm r)
\nonumber\\
&\qquad={}
\frac{(M-\ell)_{n-\ell}}{M^n}
+\sum_{j=1}^d
\frac{j^{2j}}{(j!)^2M^n(N)_j}
\sum_{p=\max\{0,j-(n-\ell)\}}^{\min\{j,\ell\}}
(n-\ell)_{j-p}(M-\ell-j+p)_{n-\ell-j+p}
\cdot [x_1\cdots x_jy_1\cdots y_jz^p]\mathcal P_j.
\label{eq:colored-prefix-marginal-formula}
\end{align}
\end{lemma}

\begin{proof}
Substituting Eq.~\eqref{eq:colored-q-one} into the left-hand side of Eq.~\eqref{eq:colored-prefix-marginal-formula}, we first sum the constant term over the $(M-\ell)_{n-\ell}$ collision-free completions, which gives $(M-\ell)_{n-\ell}/M^n$.
For each $1\leq j\leq d$, the prefactor $j^{2j}/((j!)^2M^n(N)_j)$ is independent of the completion and can be taken outside its sum.
We now evaluate the remaining sums for fixed $j$.

For a completion $\bm r$ and a selected position set $R$, put $\Omega=\{r_s:s\in R\}$.
Since $\bm r$ is collision-free, composing the two bijections with $s\mapsto r_s$ replaces their values in $R$ by the corresponding modes in $\Omega$, preserving both the product and its color conditions.
Conversely, fix $\Omega$ and two bijections from $J$ to $\Omega$.
Then, for any completion containing $\Omega$, the unique positions of its modes determine $R=\{s:r_s\in\Omega\}$ and recover the two original bijections.
The first equality below replaces $R$ by $\Omega$ for each fixed completion.
The second reverses the order of summation, so that we sum first over $\Omega$ and then over the completions containing it:
\begin{align}
&\sum_{\substack{\bm r\in[M]_{\neq}^{n}\\(r_1,\ldots,r_\ell)=\mathbf a}}
\sum_{\substack{J\subseteq[N],\,\abs{J}=j\\h_j(J)=[j]}}
\sum_{\substack{R\subseteq[n],\,\abs{R}=j\\g_j(\{r_s:s\in R\})=[j]}}
\sum_{\sigma,\tau\in\operatorname{Bij}(J,R)}
\prod_{l\in J}B_{r_{\sigma(l)},r_{\tau(l)}}^{(l)}
\nonumber\\
&\quad={}
\sum_{\substack{\bm r\in[M]_{\neq}^{n}\\(r_1,\ldots,r_\ell)=\mathbf a}}
\sum_{\substack{J\subseteq[N],\,\abs{J}=j\\h_j(J)=[j]}}
\sum_{\substack{\Omega\subseteq\{r_1,\ldots,r_n\},\,\abs{\Omega}=j\\g_j(\Omega)=[j]}}
\sum_{\sigma,\tau\in\operatorname{Bij}(J,\Omega)}
\prod_{l\in J}B_{\sigma(l),\tau(l)}^{(l)}
\nonumber\\
&\quad={}
\sum_{\substack{\Omega\subseteq[M],\,\abs{\Omega}=j\\g_j(\Omega)=[j]}}
\left(\sum_{\substack{\bm r\in[M]_{\neq}^{n}\\(r_1,\ldots,r_\ell)=\mathbf a\\\Omega\subseteq\{r_1,\ldots,r_n\}}}1\right)
\sum_{\substack{J\subseteq[N],\,\abs{J}=j\\h_j(J)=[j]}}
\sum_{\sigma,\tau\in\operatorname{Bij}(J,\Omega)}
\prod_{l\in J}B_{\sigma(l),\tau(l)}^{(l)}.
\label{eq:colored-prefix-reindexing}
\end{align}
Here, the sum of $1$ counts the completions for each fixed $\Omega$; the remaining factors no longer depend on the completion.

For a fixed set $\Omega$ of size $j$, we now count strings $\bm r$ that begin with $\mathbf a$ and contain every mode in $\Omega$, without repetitions.
The $p=\abs{\Omega\cap\{a_1,\ldots,a_\ell\}}$ modes occurring in the prefix are already included.
We therefore only need to place the other $j-p$ modes in the $n-\ell$ suffix positions.
If there are fewer than $j-p$ suffix positions, no such string exists.
Otherwise, placing these distinct modes one at a time gives $n-\ell$ choices for the first position, $n-\ell-1$ for the next, and so on, yielding $(n-\ell)_{j-p}$ placements.

Once these modes have been placed, every mode in $\Omega$ is present in the string.
The remaining positions can therefore contain any modes that have not yet been used.
The used modes consist of all $\ell$ prefix modes and the $j-p$ modes just placed, leaving $M-\ell-j+p$ available modes.
Filling the remaining $n-\ell-j+p$ positions in order, without repeating a mode, gives $(M-\ell-j+p)_{n-\ell-j+p}$ possibilities.
Every admissible string is obtained exactly once by these two steps, so their number is $(n-\ell)_{j-p}(M-\ell-j+p)_{n-\ell-j+p}$.

For example, consider $M=5$, $n=3$, $\mathbf a=(1)$, and $\Omega=\{1,4\}$.
Mode $1$ is fixed by the prefix, while mode $4$ can occupy either of the two suffix positions.
For each placement of mode $4$, the remaining entry can be $2$, $3$, or $5$, yielding $2\cdot3=6$ collision-free strings extending $\mathbf a$ and containing $\Omega$.
Here $\ell=1$, $j=2$, and $p=1$, so the general formula gives
$(n-\ell)_{j-p}(M-\ell-j+p)_{n-\ell-j+p}=(2)_1(3)_1=6$,
in agreement with the direct count.

For fixed $j$, the completion count depends on $\Omega$ only through $p$ and is independent of $J,\sigma,\tau$.
We therefore group the sets $\Omega$ in Eq.~\eqref{eq:colored-prefix-reindexing} by $p$ and take this count outside the sums over sets and bijections.
The conditions $0\leq p\leq\min\{j,\ell\}$ and $j-p\leq n-\ell$ give the range of the resulting $p$-sum.
After this grouping, the sums over $\Omega$, $J$, and the two bijections have the form in Eq.~\eqref{eq:colored-coefficient-identity}.
Restoring the prefactor and the sum over $j$ from Eq.~\eqref{eq:colored-q-one}, together with the constant term, makes this correspondence explicit:
\begin{align}
\sum_{\substack{\bm r\in[M]_{\neq}^{n}\\(r_1,\ldots,r_\ell)=\mathbf a}}
\widehat q_{d,U}^{(1)}(\bm r)
={}
\frac{(M-\ell)_{n-\ell}}{M^n}
&+\sum_{j=1}^d\frac{j^{2j}}{(j!)^2M^n(N)_j}
\sum_{p=\max\{0,j-(n-\ell)\}}^{\min\{j,\ell\}}
(n-\ell)_{j-p}(M-\ell-j+p)_{n-\ell-j+p}
\nonumber\\
&\qquad\qquad\qquad\qquad\qquad\times
\!\!\!\!\!\!\sum_{\substack{\Omega\subseteq[M],\,\abs{\Omega}=j\\\abs{\Omega\setminus\Omega_{\mathbf a}}=p,\ g_j(\Omega)=[j]}}
\sum_{\substack{J\subseteq[N],\,\abs{J}=j\\h_j(J)=[j]}}
\sum_{\sigma,\tau\in\operatorname{Bij}(J,\Omega)}
\prod_{l\in J}B_{\sigma(l),\tau(l)}^{(l)}.
\label{eq:colored-prefix-expanded-marginal}
\end{align}
For each $j$ and $p$, the three innermost sums in Eq.~\eqref{eq:colored-prefix-expanded-marginal} are exactly the right-hand side of Eq.~\eqref{eq:colored-coefficient-identity}, with $h=h_j$ and $g=g_j$.
Replacing them by $[x_1\cdots x_jy_1\cdots y_jz^p]\mathcal P_j$ proves Eq.~\eqref{eq:colored-prefix-marginal-formula}.

\end{proof}

\subsection{Marginal-evaluation complexity for a fixed coloring}

Lemma~\ref{lem:colored-exact-prefix} represents the prefix marginals of $\widehat q_{d,U}^{(1)}$ in terms of coefficients of the inverse determinant $\mathcal P_j$.
For use in the sampling algorithm, we now seek an efficient procedure for evaluating these coefficients, and hence the marginals of $\widehat q_{d,U}^{(1)}$.
The following lemma provides such a procedure.



\begin{lemma}[Evaluation cost for a fixed coloring]
\label{lem:colored-one-realization-cost}
Assume $1\leq d\leq n$ and fix the coloring maps $\{(h_j,g_j)\}_{j=1}^d$.
After $MN^2 2^{O(d)}$ arithmetic operations of preprocessing, the exact prefix marginal of $\widehat q_{d,U}^{(1)}$ at any collision-free prefix can be evaluated using $N^3 2^{O(d)}$ arithmetic operations.
\end{lemma}

\begin{proof}

The proof proceeds in three steps.
First, we reduce the determinant defining $\mathcal P_j$ in Eq.~\eqref{eq:colored-prefix-generating-polynomial} from dimension $M$ to $N$ to lower the cost of the subsequent matrix operations.
Second, we precompute the sums shared across prefixes, so that a new query requires corrections only from its prefix modes.
Finally, we compute the required coefficients by matrix and polynomial multiplications and bound their cost.

\emph{Reduction of the determinant.---}
Fix an order $1\leq j\leq d$ and a collision-free prefix $\mathbf a=(a_1,\ldots,a_\ell)$, write $h=h_j$ and $g=g_j$, and abbreviate $s_u:=x_{g(u)}z^{\mathbf 1_{\{u\notin\Omega_{\mathbf a}\}}}$.
Since $\Omega_{\mathbf a}$ contains the modes outside the prefix, $s_u=x_{g(u)}$ outside the prefix and $s_u=x_{g(u)}z$ on the prefix.

Expanding $B_{uv}^{(l)}=V_{u,l}\overline{V_{v,l}}-\delta_{uv}$ separates the diagonal term from the products of entries of $V$.
Accordingly, write the matrix inside the determinant in Eq.~\eqref{eq:colored-prefix-generating-polynomial} as $\mathsf D-\mathsf A\mathsf F$.
Here $\mathsf D$ is the $M\times M$ diagonal matrix with $\mathsf D_{uv}:=\delta_{uv}(1+s_u\sum_{l=1}^Ny_{h(l)})$.
The matrices $\mathsf A$ and $\mathsf F$ have dimensions $M\times N$ and $N\times M$, respectively, and their entries are $\mathsf A_{u\alpha}:=s_uV_{u,\alpha}$ and $\mathsf F_{\alpha v}:=y_{h(\alpha)}\overline{V_{v,\alpha}}$, where $\alpha\in[N]$ indexes an input column.

Indeed, $(\mathsf A\mathsf F)_{uv}=s_u\sum_{\alpha=1}^Ny_{h(\alpha)}V_{u,\alpha}\overline{V_{v,\alpha}}$.
Because the sum involving $B_{uv}^{(l)}$ is subtracted in Eq.~\eqref{eq:colored-prefix-generating-polynomial}, the term $-\delta_{uv}$ in $B_{uv}^{(l)}$ produces the positive diagonal correction in $\mathsf D$, while the term $V_{u,l}\overline{V_{v,l}}$ produces $-\mathsf A\mathsf F$.
Moreover, $s_u^2=0$, so the inverse of the $u$th diagonal entry of $\mathsf D$ is $1-s_u\sum_{l=1}^Ny_{h(l)}$.
Multiplying this inverse by $\mathsf A_{u\alpha}=s_uV_{u,\alpha}$ leaves $\mathsf A_{u\alpha}$ unchanged, because the correction contains $s_u^2$.
Thus $\mathsf D^{-1}\mathsf A=\mathsf A$.

We now factor out $\mathsf D$ and apply Sylvester's determinant theorem, which holds for the commuting polynomial entries used here~\cite{chabaud2022quantum}.
Using $\mathsf D^{-1}\mathsf A=\mathsf A$ in the last equality gives
\begin{align}
\mathcal P_j
&=\det(\mathsf D)^{-1}\det(I_M-\mathsf D^{-1}\mathsf A\mathsf F)^{-1}
\nonumber\\
&=\det(\mathsf D)^{-1}\det(I_N-\mathsf F\mathsf D^{-1}\mathsf A)^{-1}
\nonumber\\
&=\det(\mathsf D)^{-1}\det(I_N-\mathsf F\mathsf A)^{-1}.
\label{eq:colored-sylvester-factorization}
\end{align}
The determinant now has dimension $N$, but its entries and the diagonal prefactor still involve all $M$ output modes.
Evaluating these expressions afresh would still require a sum or product over all $M$ modes for every prefix query.
We therefore isolate the contributions shared across prefixes, so that only sums over the prefix modes need to be computed anew.

\emph{Constructing the reduced matrix and diagonal factors.---}
To reuse these common contributions across prefixes, we precompute sums over all output modes in $O(MN^2)$ operations for each order $j$.
For each prefix query, we then combine the stored sums with corrections involving only the prefix modes to construct the entries of $\mathsf F\mathsf A$ and the factors of $\det(\mathsf D)^{-1}$ in $O(nN^2)$ operations.

We begin with the reduced matrix.
Direct multiplication gives
\begin{align}
(\mathsf F\mathsf A)_{\alpha\beta}
&=\sum_{u=1}^M\mathsf F_{\alpha u}\mathsf A_{u\beta}
=y_{h(\alpha)}\sum_{u=1}^Ms_uV_{u,\beta}\overline{V_{u,\alpha}}.
\end{align}
Here $\alpha,\beta\in[N]$ label the input columns, and the sum runs over the $M$ physical output modes.
To separate the sums that can be reused across queries, we group the modes whose $s_u$ contain the same color variable.

For a color $c\in[j]$, let $\Omega_c:=\{u\in[M]:g(u)=c\}$ be the set of all modes of color $c$.
Then the matrix entry can be rewritten as
\begin{align}\label{eq:colored-reduced-matrix-by-color}
(\mathsf F\mathsf A)_{\alpha\beta}
&=y_{h(\alpha)}\sum_{c=1}^j\sum_{u\in\Omega_c}s_uV_{u,\beta}\overline{V_{u,\alpha}}.
\end{align}
For $u\in\Omega_c$, the coloring fixes the variable in $s_u$ to $x_c$, while the prefix determines whether the additional factor $z$ is present.
The color classes $\Omega_c$ and the matrix-entry products remain fixed across prefix queries.

The diagonal prefactor can also be expressed through sums over the modes of each color.
Specifically, in expanding $\det(\mathsf D)^{-1}=\prod_{u=1}^M(1-s_u\sum_{l=1}^Ny_{h(l)})$, each factor contributes either its constant term $1$ or the term $-s_u\sum_{l=1}^Ny_{h(l)}$.
Suppose the second term is chosen from two factors corresponding to distinct modes $u,v\in\Omega_c$.
Then, both $s_u$ and $s_v$ contain $x_c$, so $s_us_v=0$ by $x_c^2=0$, regardless of whether these modes belong to the prefix.
Thus, within one color class, only the choice of $1$ from every factor and the choices of the second term from exactly one factor remain.
Summing these contributions gives
\begin{align}\label{eq:colored-diagonal-color-factor}
\prod_{u\in\Omega_c}\left(1-s_u\sum_{l=1}^Ny_{h(l)}\right)
&=1-\left(\sum_{u\in\Omega_c}s_u\right)\left(\sum_{l=1}^Ny_{h(l)}\right).
\end{align}
Hence, the product over all $M$ modes is determined by the $j$ sums $\sum_{u\in\Omega_c}s_u$, one for each color.

The sums required by Eqs.~\eqref{eq:colored-reduced-matrix-by-color} and~\eqref{eq:colored-diagonal-color-factor} both have the form $\sum_{u\in\Omega_c}s_uf(u)$; $f(u)=1$ for the diagonal factor and $f(u)=V_{u,\beta}\overline{V_{u,\alpha}}$ for the $(\alpha,\beta)$ entry of the reduced matrix.
We start by assigning $s_u=x_c$ to every mode of the class and correct the prefix modes, for which the required change is $x_cz-x_c=x_c(z-1)$.
Thus,
\begin{align}
\sum_{u\in\Omega_c}s_uf(u)
&=x_c\left\{\sum_{u\in\Omega_c}f(u)
+(z-1)\sum_{u\in\Omega_c\cap\{a_1,\ldots,a_\ell\}}f(u)\right\}.
\label{eq:colored-prefix-total-correction}
\end{align}
The first sum inside the braces is independent of the prefix, and only the second sum must be computed for each query.

The second sum runs over the prefix modes of color $c$.
Since the prefix is collision-free, the positions $t$ satisfying $g(a_t)=c$ enumerate each such mode exactly once, giving
\begin{align}
\sum_{u\in\Omega_c\cap\{a_1,\ldots,a_\ell\}}f(u)
&=\sum_{\substack{1\leq t\leq\ell\\g(a_t)=c}}f(a_t).
\label{eq:colored-prefix-index-sum}
\end{align}

For $f(u)=1$, the full color-class sum is $\abs{\Omega_c}$, and the prefix sum $\sum_{1\leq t\leq\ell:\,g(a_t)=c}1$ counts the prefix entries of color $c$.
Using these counts for the diagonal prefactor and the corresponding matrix-entry sums for $\mathsf F\mathsf A$, Eq.~\eqref{eq:colored-sylvester-factorization} becomes
\begin{align}
\mathcal P_j
={}&\prod_{c=1}^j\left[1-x_c\left(\sum_{l=1}^Ny_{h(l)}\right)
\left\{\abs{\Omega_c}+(z-1)\sum_{\substack{1\leq t\leq\ell\\g(a_t)=c}}1\right\}\right]
\nonumber\\
&\times\det\!\left[\delta_{\alpha\beta}-y_{h(\alpha)}\sum_{c=1}^jx_c
\left\{\sum_{u\in\Omega_c}V_{u,\beta}\overline{V_{u,\alpha}}
+(z-1)\sum_{\substack{1\leq t\leq\ell\\g(a_t)=c}}V_{a_t,\beta}\overline{V_{a_t,\alpha}}\right\}
\right]_{\alpha,\beta\in[N]}^{-1}.
\label{eq:colored-first-order-reduction}
\end{align}
In Eq.~\eqref{eq:colored-first-order-reduction}, the full color-class sums can be stored during preprocessing, while the prefix sums (depending on $\bm a$) are computed for each query, as depicted follows.
\begin{enumerate}
\item For fixed $j,h,g$, precompute $\abs{\Omega_c}$ and $\sum_{u\in\Omega_c}V_{u,\beta}\overline{V_{u,\alpha}}$ for all $c\in[j]$ and $\alpha,\beta\in[N]$.
Each mode $u$ contributes one to the count for color $g(u)$ and contributes $V_{u,\beta}\overline{V_{u,\alpha}}$ to that color's sum for each pair $(\alpha,\beta)$.
Processing each of the $M$ modes once therefore requires $O(MN^2)$ operations.
The coefficients of $\sum_{l=1}^Ny_{h(l)}$ are the numbers of input columns of each color and can be computed once in $O(N)$ operations, within the same bound.
The precomputed values are then reused for every prefix query.
\item For a given prefix, we count, for each color, the modes appearing in the prefix and compute the corresponding matrix-entry sums.
Since the prefix contains $\ell$ modes, this requires $O(\ell N^2)$ operations.
We form the expressions in braces in Eq.~\eqref{eq:colored-first-order-reduction} by combining each stored total with $(z-1)$ times the corresponding prefix count or sum.
This requires $O(jN^2)$ operations over the $j$ colors and $N^2$ input-column pairs and specifies the coefficients of the reduced-matrix entries and the individual diagonal factors.
Since $\ell,j\leq n$, their construction costs $O(nN^2)$ per prefix, without summing over all $M$ modes again.
\end{enumerate}

\emph{Computing the coefficients of $\mathcal P_j$.---}
The preceding step specifies the entries of the reduced matrix and the individual diagonal factors.
We now combine them to compute the at most $j+1$ coefficients of $x_1\cdots x_jy_1\cdots y_jz^p$ in $\mathcal P_j$ that determine the order-$j$ contribution to the marginal.
We express the inverse determinant through finitely many matrix and polynomial multiplications and show that this calculation costs $N^3 2^{O(j)}$ arithmetic operations.

We use the polynomial arithmetic introduced above, with complex coefficients and formal variables $x_1,\ldots,x_j,y_1,\ldots,y_j,z$.
Thus we retain only monomials with no repeated $x$ or $y$ variable and with $z$-degree at most $j$.

We use the trace--log identity to express the inverse determinant through matrix powers; the square-zero relations make the required power series finite.
Put $\mathsf L:=\mathsf F\mathsf A$.
The $N^2$ entries of $\mathsf L$ are polynomials, and each of their monomials contains one $x$ and one $y$ variable.
The entries of $\mathsf L^{j+1}$ are sums of products of $j+1$ entries of $\mathsf L$.
Each monomial in such a product contains $j+1$ $x$ variables, so at least one of the $j$ available variables repeats and the term vanishes.
Thus $\mathsf L^{j+1}=0$ under the stated rules.

The entries of $\mathsf L$ commute and have no constant term in the $x$ variables, so the formal identity $-\log\det(I_N-\mathsf L)=\sum_{r\geq1}\operatorname{tr}(\mathsf L^r)/r$ applies~\cite[Eq.~(1)]{goulden1992immanants}.
Here $\log$ denotes its formal power series about one, and every nonconstant term of $\det(I_N-\mathsf L)$ contains an $x$ variable.
As above, a product of $j+1$ terms containing $x$ variables vanishes, so the logarithm is a finite polynomial and requires no analytic convergence assumption.
The trace sum terminates at $r=j$ because $\mathsf L^{j+1}=0$, and its $(j+1)$st power also vanishes because each of its monomials contains an $x$ variable.
Exponentiating the identity therefore gives the finite expression
\begin{align}
\det(I_N-\mathsf L)^{-1}
&=\sum_{s=0}^j\frac1{s!}\left(\sum_{r=1}^j\frac{\operatorname{tr}(\mathsf L^r)}r\right)^s.
\label{eq:colored-finite-trace-expansion}
\end{align}
The inner sum is a scalar polynomial formed from the matrix traces, and the outer sum gives the inverse determinant.
We next bound the cost of evaluating this expression and multiplying it by the diagonal prefactor in Eq.~\eqref{eq:colored-first-order-reduction}.

Each surviving monomial is specified by a subset of the $x$ variables, a subset of the $y$ variables, and a $z$-degree between zero and $j$.
There are $2^j$ choices for each subset and $j+1$ choices for the degree, giving $(j+1)4^j$ possible monomials.
We store one complex coefficient for each, including zero when the monomial is absent.
Thus a matrix power requires $N^2$ such lists of coefficients, one for each entry, while a single polynomial such as its trace or $\mathcal P_j$ requires one list.
Monomials containing only some of the colors must also be retained, since their products can contribute to the required coefficients.

Polynomial addition adds the coefficients of matching monomials.
For multiplication, we consider each pair of monomials, one from each input polynomial.
We discard the product if an $x$ or $y$ variable is repeated or the total $z$-degree exceeds $j$.
Otherwise, we multiply the two complex coefficients and add the result to the coefficient of the resulting monomial.
All contributions to the same monomial are combined, so every intermediate polynomial still requires only $(j+1)4^j$ coefficients.

There are at most $((j+1)4^j)^2$ monomial pairs, each requiring at most one multiplication and one addition of complex numbers.
For each fixed choice of the $x$ and $y$ variables in the two input monomials, repeated variables are checked once before varying the $z$-degrees.
Including these checks, a polynomial multiplication costs $O(((j+1)4^j)^2)=O(j^2 16^j)$ arithmetic operations.
Polynomial addition and initialization each cost $O((j+1)4^j)$ and are included in this bound.

We compute the matrix powers successively using
\begin{align}
(\mathsf L^{r+1})_{\alpha\beta}
&=\sum_{l=1}^N(\mathsf L^r)_{\alpha l}\mathsf L_{l\beta}
\end{align}
for $1\leq r<j$.
Each of the $N^2$ entries requires $N$ polynomial products and their sum, so one matrix multiplication costs $O(j^2N^3 16^j)$ arithmetic operations.
We combine contributions to the same monomial in each entry after every multiplication, without retaining separate terms for all sequences of intermediate matrix indices.
Computing the powers through $\mathsf L^j$ requires at most $j-1$ matrix multiplications, for a total cost of $O(j^3N^3 16^j)$.

For each $r$, taking $\operatorname{tr}(\mathsf L^r)$ adds the $N$ diagonal polynomials of $\mathsf L^r$.
Combining these traces with weights $1/r$ gives the inner sum in Eq.~\eqref{eq:colored-finite-trace-expansion}, a single polynomial, at a cost of $O(j^2N4^j)$.
Its powers through order $j$ are computed successively using $O(j)$ polynomial multiplications, and adding them with weights $1/s!$ gives the inverse determinant.
Multiplying the $j$ factors of the diagonal prefactor and then multiplying by this inverse determinant requires another $O(j)$ polynomial multiplications and gives all retained coefficients of $\mathcal P_j$.
Since each polynomial multiplication costs $O(j^2 16^j)$, these products require $O(j^3 16^j)$ additional operations.
Together with the coefficient additions, the total cost therefore remains $O(j^3N^3 16^j)$.

The required coefficients of $x_1\cdots x_jy_1\cdots y_jz^p$ can now be read directly from the stored coefficients of $\mathcal P_j$.
There are at most $j+1$ of them, since the allowed values of $p$ lie between zero and $j$.
Multiplying them by the corresponding factors in Eq.~\eqref{eq:colored-prefix-marginal-formula} and summing over $p$ gives the order-$j$ contribution to the marginal without increasing the bound.

Together with the preceding construction, preprocessing costs $O(MN^2)$ once per order $j$, and each prefix query costs $O(nN^2+j^3N^3 16^j)$.
Summing over $1\leq j\leq d$ gives $O(dMN^2)\leq MN^2 2^{O(d)}$ preprocessing operations.
Similarly, summing the query costs gives
\begin{align}
\sum_{j=1}^d O(nN^2+j^3N^3 16^j)
&=O(dnN^2+d^4N^3 16^d)
=N^3 2^{O(d)},
\end{align}
where $n\leq N$ and polynomial factors in $d$ are absorbed into $2^{O(d)}$.
Adding the constant contribution in Eq.~\eqref{eq:colored-prefix-marginal-formula} does not increase this bound.
This proves Lemma~\ref{lem:colored-one-realization-cost}.

\end{proof}

\subsection{Error after averaging colorings}
\label{app:coloring-error}

Lemmas~\ref{lem:colored-exact-prefix} and~\ref{lem:colored-one-realization-cost} give the exact prefix marginals and their evaluation cost for one fixed coloring.
It remains to determine how many independent colorings needs to be averaged to achieve required error bound.
We first bound the error introduced by one coloring under the Gaussian measure, where orthogonality reduces the calculation to coefficient variances.
Averaging independent colorings reduces this error, which we then combine with the truncation error and transfer to the Haar-random unitary setting.
The resulting bound concerns $\widehat\Delta_{\mathrm{cf}}(U)$, the collision-free $L^1$ error of $\widehat q_{d,U}$ relative to $p_U$, defined in Eq.~\eqref{eq:randomized-Delta-cf-definition}.

\begin{lemma}[Averaged coloring error]
\label{lem:colored-average-error}
Under the assumptions of Proposition~\ref{prop:randomized-marginal-complexity}, let $1\leq d\leq n$ and $L:=\lceil8C_Hde^{2d}/(\varepsilon^2\delta)\rceil$.
Let $\widehat q_{d,U}$ be the average of $L$ copies of $\widehat q_{d,U}^{(1)}$ generated by independent colorings.
Then $\E_{\mathrm{col}}[\widehat q_{d,U}(\bm r)]=q_{d,U}(\bm r)$ for every collision-free string $\bm r$, and, if $M\geq C_0n^2N$,
\begin{align}
\E_{U,\mathrm{col}}\!\left[\widehat\Delta_{\mathrm{cf}}(U)^2\right]
&\leq C_HR_d(n,N)+\frac{\varepsilon^2\delta}{8}.
\end{align}
\end{lemma}

\begin{proof}
Fix a collision-free string $\bm r$.
We first bound the mean-square error of $\widehat q_{d,U}^{(1)}(\bm r)-q_{d,U}(\bm r)$ under the Gaussian measure.
For this calculation, replace the scaled Haar submatrix $A_{\bm r}=\sqrt M\,U_{\bm r,[N]}$ by $X\sim\mathcal N_{\mathbb C}(0,1)^{n\times N}$ with independent standard complex Gaussian entries, drawn independently of the colorings.
Until we return to Haar submatrices below, we retain the notation $p_U(\bm r)$, $q_{d,U}(\bm r)$, $\widehat q_{d,U}^{(1)}(\bm r)$, and $\widehat q_{d,U}(\bm r)$ for their evaluations at $X$, with matrix expectations taken over $X$.

Since $\bm r$ is collision-free, $\delta_{r_a,r_b}=\delta_{ab}$.
Relabeling its physical rows by their positions, the factors in Eq.~\eqref{eq:colored-q-one} become $B_{ab}^{(l)}=X_{a,l}\overline{X_{b,l}}-\delta_{ab}$ for $a,b\in[n]$, while row $a$ retains the color $g_j(r_a)$.
The colors $g_j(r_1),\ldots,g_j(r_n)$ remain independent and uniform on $[j]$, so the same retention probabilities apply to these Gaussian evaluations.

Multiplying Eqs.~\eqref{eq:uncolored-q-expansion} and~\eqref{eq:colored-q-one} by $M^n$ removes their common factor $M^{-n}$.
We write both sums over the same indexed products and, for each $j,J,R$, put $I:=\mathbf 1_{\{h_j(J)=[j]\}}\mathbf 1_{\{g_j(\{r_a:a\in R\})=[j]\}}$.
This indicator does not depend on the two bijections $\sigma,\tau$.
Each order-$j$ product has coefficient $1/(N)_j$ in $M^nq_{d,U}$ and $j^{2j}I/((j!)^2(N)_j)$ in $M^n\widehat q_{d,U}^{(1)}$.
The constant terms are both one and cancel, leaving the difference of these coefficients.
For fixed coloring, these coefficient differences are constant under $\E_X$.

We now use Gaussian orthogonality to evaluate the squared error from these coefficient differences.
Equation~\eqref{eq:bilinear-orthonormal} gives $\E_X[B_{ab}^{(l)}\overline{B_{a'b'}^{(l)}}]=\delta_{aa'}\delta_{bb'}$ for one column.
For products using the same column set $J$, independence across columns gives
\begin{align}
\E_X\!\left[
\prod_{l\in J}B_{\sigma(l),\tau(l)}^{(l)}
\overline{\prod_{l\in J}B_{\sigma'(l),\tau'(l)}^{(l)}}
\right]
&=
\prod_{l\in J}
\delta_{\sigma(l),\sigma'(l)}
\delta_{\tau(l),\tau'(l)}.
\end{align}
Products with different column sets have zero inner product, since a column appearing in only one contributes an independent centered factor.
For equal column sets, the identity above forces $R,\sigma,\tau$ to agree for a nonzero inner product.
Thus all distinct indexed products, including those of different orders, are orthogonal and each has squared norm one.

For each fixed coloring, substituting the two expansions and applying this orthogonality gives
\begin{align}
M^{2n}\E_X\!\left[
\left|\widehat q_{d,U}^{(1)}(\bm r)-q_{d,U}(\bm r)\right|^2
\right]
&=
\E_X\!\left[
\left|
\sum_{j=1}^d
\sum_{\substack{J\subseteq[N],\,R\subseteq[n]\\\abs{J}=\abs{R}=j}}
\sum_{\sigma,\tau\in\operatorname{Bij}(J,R)}
\left(\frac{j^{2j}I}{(j!)^2(N)_j}-\frac1{(N)_j}\right)
\prod_{l\in J}B_{\sigma(l),\tau(l)}^{(l)}
\right|^2
\right]
\nonumber\\
&=
\sum_{j=1}^d
\sum_{\substack{J\subseteq[N],\,R\subseteq[n]\\\abs{J}=\abs{R}=j}}
\sum_{\sigma,\tau\in\operatorname{Bij}(J,R)}
\left|\frac{j^{2j}I}{(j!)^2(N)_j}-\frac1{(N)_j}\right|^2
\prod_{l\in J}\E_X\!\left[
\left|B_{\sigma(l),\tau(l)}^{(l)}\right|^2
\right]
\nonumber\\
&=
\sum_{j=1}^d
\sum_{\substack{J\subseteq[N],\,R\subseteq[n]\\\abs{J}=\abs{R}=j}}
\sum_{\sigma,\tau\in\operatorname{Bij}(J,R)}
\left|\frac{j^{2j}I}{(j!)^2(N)_j}-\frac1{(N)_j}\right|^2.
\label{eq:colored-gaussian-coefficient-error}
\end{align}
The coloring events for different products can be correlated, as their cross terms have already vanished under the Gaussian expectation.

To average these squared coefficient differences over the colorings, we use $\E_{\mathrm{col}}[I]=(j!/j^j)^2$ from Eq.~\eqref{eq:colored-retention-probability} and $I^2=I$.
The colored coefficient has mean $1/(N)_j$, and its variance is
\begin{align}
\E_{\mathrm{col}}\!\left[
\left|\frac{j^{2j}I}{(j!)^2(N)_j}-\frac1{(N)_j}\right|^2
\right]
&=\frac{j^{2j}-(j!)^2}{(j!)^2(N)_j^2}.
\end{align}
At order $j$, there are $\binom Nj$ choices of $J$, $\binom nj$ choices of $R$, and $(j!)^2$ choices of the two bijections.
The variance is the same for each such product, so taking the coloring expectation of Eq.~\eqref{eq:colored-gaussian-coefficient-error} gives
\begin{align}
M^{2n}\E_{X,\mathrm{col}}\!\left[
\left|\widehat q_{d,U}^{(1)}(\bm r)-q_{d,U}(\bm r)\right|^2
\right]
=\sum_{j=1}^d\binom Nj\binom nj(j!)^2
\frac{j^{2j}-(j!)^2}{(j!)^2(N)_j^2}
=\sum_{j=1}^d
\left[\frac{j^{2j}}{(j!)^2}-1\right]\frac{\binom nj}{\binom Nj}
\leq de^{2d}.
\label{eq:coloring-variance}
\end{align}
The second equality uses $(N)_j=\binom Nj j!$.
For the final bound, $\binom nj/\binom Nj\leq1$ and $j!\geq(j/e)^j$ give $j^{2j}/(j!)^2\leq e^{2j}\leq e^{2d}$.

We now pass from one coloring to the average of $L$ independent colorings and bound $\widehat q_{d,U}(\bm r)-q_{d,U}(\bm r)$.
For every fixed matrix $X$, these colorings give independent estimates with mean $q_{d,U}(\bm r)$ by Eq.~\eqref{eq:unbiasedness}.
Their average therefore has the one-copy variance divided by $L$.
Taking the expectation over $X$ and using Eq.~\eqref{eq:coloring-variance} gives
\begin{align}
M^{2n}\E_{X,\mathrm{col}}\!\left[
\left|\widehat q_{d,U}(\bm r)-q_{d,U}(\bm r)\right|^2
\right]
&\leq\frac{de^{2d}}{L}
\leq\frac{\varepsilon^2\delta}{8C_H}.
\nonumber
\end{align}
In the choice of $L$, the factor $C_H$ compensates for the subsequent Haar-density comparison, while the factor $8$ gives the required coloring-error contribution $\varepsilon^2\delta/8$.

To bound the total error $p_U(\bm r)-\widehat q_{d,U}(\bm r)$, we combine this coloring error with the truncation error.
At fixed $X$, $p_U(\bm r)-q_{d,U}(\bm r)$ is constant under the coloring expectation, whereas $q_{d,U}(\bm r)-\widehat q_{d,U}(\bm r)$ has mean zero.
Their cross term therefore vanishes under the coloring expectation, so expanding the squared total error and averaging over $X$ gives
\begin{align}
M^{2n}\E_{X,\mathrm{col}}
\left|p_U(\bm r)-\widehat q_{d,U}(\bm r)\right|^2
&=M^{2n}\E_X
\left|p_U(\bm r)-q_{d,U}(\bm r)\right|^2
+M^{2n}\E_{X,\mathrm{col}}
\left|q_{d,U}(\bm r)-\widehat q_{d,U}(\bm r)\right|^2
\nonumber\\
&\leq R_d(n,N)+\frac{\varepsilon^2\delta}{8C_H}.
\label{eq:colored-gaussian-error-split}
\end{align}
Proposition~\ref{prop:gaussian-truncation} gives the truncation contribution, and the preceding averaging bound gives the coloring contribution.

To transfer this pointwise bound to the Haar measure, fix the collection of colorings used to form $\widehat q_{d,U}$.
Both $p_U(\bm r)$ and $\widehat q_{d,U}(\bm r)$ are functions of $A_{\bm r}$ for fixed colorings, and their squared difference is nonnegative.
Under $M\geq C_0n^2N$, Lemma~\ref{lem:rectangular-domination} bounds its Haar expectation by $C_H$ times its Gaussian expectation for each fixed collection.
Averaging over the colorings and applying Eq.~\eqref{eq:colored-gaussian-error-split} gives
\begin{align}
\E_{U,\mathrm{col}}
\left|p_U(\bm r)-\widehat q_{d,U}(\bm r)\right|^2
&\leq
\frac{C_HR_d(n,N)+\varepsilon^2\delta/8}{M^{2n}}.
\label{eq:colored-pointwise-error}
\end{align}

Finally, we pass from the pointwise bound to the collision-free $L^1$ error $\widehat\Delta_{\mathrm{cf}}(U)$, which sums the absolute errors over the $(M)_n$ collision-free strings.
Applying Cauchy--Schwarz, taking expectations, and using Eq.~\eqref{eq:colored-pointwise-error} with $(M)_n\leq M^n$ gives
\begin{align}
\E_{U,\mathrm{col}}\!\left[\widehat\Delta_{\mathrm{cf}}(U)^2\right]
&\leq (M)_n\sum_{\bm r\in[M]_{\neq}^{n}}
\E_{U,\mathrm{col}}
\left|p_U(\bm r)-\widehat q_{d,U}(\bm r)\right|^2
\leq \frac{(M)_n^2}{M^{2n}}
\left(C_HR_d(n,N)+\frac{\varepsilon^2\delta}{8}\right)
\leq C_HR_d(n,N)+\frac{\varepsilon^2\delta}{8}.
\end{align}
The first factor $(M)_n$ comes from Cauchy--Schwarz, and the second comes from the number of strings in the sum.
This is Eq.~\eqref{eq:randomized-collision-free-error} and proves Lemma~\ref{lem:colored-average-error}.
\end{proof}

The proof of Proposition~\ref{prop:randomized-marginal-complexity} is now straightforward.

\begin{proof}[Proof of Proposition~\ref{prop:randomized-marginal-complexity}]
Assume first that $1\leq d\leq n$.
Equation~\eqref{eq:colored-q-one} defines a real signed function for each fixed coloring, and Lemma~\ref{lem:colored-exact-prefix} gives its exact prefix marginals.
Lemma~\ref{lem:colored-average-error} shows that averaging $L$ independent copies gives the unbiased function $\widehat q_{d,U}$ with the error bound in Eq.~\eqref{eq:randomized-collision-free-error}.

For one coloring, Lemma~\ref{lem:colored-one-realization-cost} gives a preprocessing cost of $MN^2 2^{O(d)}$ arithmetic operations.
After this preprocessing, the exact marginal of any collision-free prefix can be evaluated in $N^3 2^{O(d)}$ arithmetic operations using the same precomputed data.
For a given prefix, we evaluate its marginal for each of the $L$ fixed colorings and average the results.
By linearity, this gives the exact prefix marginal of $\widehat q_{d,U}$.
The choice in Lemma~\ref{lem:colored-average-error} satisfies $L=\varepsilon^{-2}\delta^{-1}2^{O(d)}$, and multiplying the single-coloring costs by $L$ therefore gives $MN^2\varepsilon^{-2}\delta^{-1}2^{O(d)}$ preprocessing operations and $N^3\varepsilon^{-2}\delta^{-1}2^{O(d)}$ operations per prefix query.

For $d=0$, no coloring is required, and $\widehat q_{0,U}(\bm r)=q_{0,U}(\bm r)=1/M^n$.
A collision-free prefix of length $\ell$ has $(M-\ell)_{n-\ell}$ completions, so its marginal is $(M-\ell)_{n-\ell}/M^n$, which can be evaluated within the stated complexity bounds.
Only the truncation error remains, and repeating the argument from the proof of Proposition~\ref{prop:collision-free-L1} gives $\E_U[\widehat\Delta_{\mathrm{cf}}(U)^2]\leq C_HR_0(n,N)$.
This proves Proposition~\ref{prop:randomized-marginal-complexity}.
\end{proof}

\subsection{Proof of the colored coefficient identity}
\label{app:colored-coefficient-identity-proof}

In the remainder of this appendix, we prove Eq.~\eqref{eq:colored-coefficient-identity}, which we use in the proof of Lemma~\ref{lem:colored-exact-prefix}.
We first express the sums over $J,\sigma,\tau$ as a coefficient of a permanent for fixed $\Omega$.
We then sum these permanents using an inverse-determinant formula~\cite{goulden1992immanants, chabaud2022quantum, macmahon2004combinatory} and extract the coefficient that imposes the output-color and prefix conditions.

Fix an order $j$ and a collision-free prefix $\mathbf a=(a_1,\ldots,a_\ell)$.
Write $h=h_j:[N]\to[j]$ and $g=g_j:[M]\to[j]$ for the input and output colorings, respectively, and recall that $\Omega_{\mathbf a}=[M]\setminus\{a_1,\ldots,a_\ell\}$.
For fixed $\Omega\subseteq[M]$ with $\abs{\Omega}=j$, consider the $j\times j$ matrix with entries $\sum_{l=1}^Ny_{h(l)}B_{uv}^{(l)}$ for $u,v\in\Omega$.
Its permanent sums over bijections $\pi:\Omega\to\Omega$ between the row and column mode labels.
Expanding each selected entry then chooses an input-column index $\lambda(u)\in[N]$ for each row $u$.
Summing over the maps $\lambda:\Omega\to[N]$ lists all these choices exactly once, giving
\begin{align}
\Per\!\left[
\sum_{l=1}^N y_{h(l)}B_{uv}^{(l)}
\right]_{u,v\in\Omega}
=
\sum_{\pi\in\operatorname{Bij}(\Omega,\Omega)}
\prod_{u\in\Omega}
\left(\sum_{l=1}^N y_{h(l)}B_{u,\pi(u)}^{(l)}\right)
=
\sum_{\pi\in\operatorname{Bij}(\Omega,\Omega)}
\sum_{\lambda:\Omega\to[N]}
\prod_{u\in\Omega}
y_{h(\lambda(u))}B_{u,\pi(u)}^{(\lambda(u))}.
\end{align}

The $y$-monomial in a term equals $y_1\cdots y_j$ precisely when the colors $h(\lambda(u))$, $u\in\Omega$, are all distinct.
If two values of $\lambda$ were equal, their colors would also be equal.
Every surviving $\lambda$ is therefore injective, so $J:=\lambda(\Omega)$ has size $j$, satisfies $h(J)=[j]$, and $\lambda:\Omega\to J$ is a bijection.

To recover the product in Eq.~\eqref{eq:colored-coefficient-identity}, we reindex $u \in \Omega$ by the selected input columns $l\in J$.
Specifically, for every surviving $\lambda$ and $\pi$, define $\sigma:=\lambda^{-1}:J\to\Omega$ and $\tau:=\pi\circ\lambda^{-1}:J\to\Omega$.
For each $l\in J$, the mode $\sigma(l)$ is the unique row that selected input column $l$, and $\tau(l)=\pi(\sigma(l))$ is the physical-mode label indexing the corresponding matrix column.
Reindexing the product by $l=\lambda(u)$ gives
\begin{align}
\prod_{u\in\Omega}B_{u,\pi(u)}^{(\lambda(u))}
&=\prod_{l\in J}B_{\sigma(l),\tau(l)}^{(l)}.
\end{align}
Every pair $\sigma,\tau\in\operatorname{Bij}(J,\Omega)$ recovers a unique original choice through $\lambda=\sigma^{-1}$ and $\pi=\tau\circ\sigma^{-1}$.
Hence, taking into account all possible choice of $J\subseteq[N]$ with $\abs{J}=j$ and $h(J)=[j]$, it follows that
\begin{align}\label{eq:y-coefficient-permanent}
[y_1\cdots y_j]
\Per\!\left[
\sum_{l=1}^N y_{h(l)}B_{uv}^{(l)}
\right]_{u,v\in\Omega}
&={}
\sum_{\substack{J\subseteq[N],\,\abs{J}=j\\h(J)=[j]}}
\sum_{\sigma,\tau\in\operatorname{Bij}(J,\Omega)}
\prod_{l\in J}B_{\sigma(l),\tau(l)}^{(l)}.
\end{align}

We next collect the permanents over all subsets $\Omega\subseteq[M]$.
To distinguish the sets before imposing the output-color condition, introduce temporary formal variables $t_1,\ldots,t_M$ that commute, where $t_u$ records the choice of mode $u$ and $t_u^2=0$ excludes repeated choices of that mode.
Each set $\Omega$ is then represented by $\prod_{u\in\Omega}t_u$.
These variables are used only in this proof, not in the evaluation procedure of Lemma~\ref{lem:colored-one-realization-cost}.

To rewrite this sum in inverse-determinant form, we invoke the square-free coefficient version of MacMahon's master theorem~\cite[Corollary~2.2]{goulden1992immanants} (see also Refs.~\cite{chabaud2022quantum, macmahon2004combinatory}).
Specifically, imposing the square-free relations $t_u^2=0$ in the theorem yields the identity
\begin{align}
\det\!\left[
\delta_{uv}-t_u\sum_{l=1}^Ny_{h(l)}B_{uv}^{(l)}
\right]_{u,v\in[M]}^{-1}
&={}
\sum_{\Omega\subseteq[M]}
\left(\prod_{u\in\Omega}t_u\right)
\Per\!\left[
\sum_{l=1}^Ny_{h(l)}B_{uv}^{(l)}
\right]_{u,v\in\Omega}.
\label{eq:colored-square-free-macmahon}
\end{align}
The term $\Omega=\varnothing$ is one by the convention $\Per(\varnothing)=1$.

We now substitute $t_u=x_{g(u)}z^{\mathbf 1_{\{u\notin\Omega_{\mathbf a}\}}}$ for each mode $u$.
The substituted expressions commute and square to zero because $x_{g(u)}^2=0$.
The substitution therefore preserves the defining relations $t_u^2=0$ and hence the identity above.
Then the left-hand side of Eq.~\eqref{eq:colored-square-free-macmahon} becomes $\mathcal P_j$ by Eq.~\eqref{eq:colored-prefix-generating-polynomial}, giving
\begin{align}
\mathcal P_j(\bm x,\bm y,z)
={}&
\sum_{\Omega\subseteq[M]}
\left(\prod_{u\in\Omega}x_{g(u)}\right)
z^{\abs{\Omega\setminus\Omega_{\mathbf a}}}
\Per\!\left[
\sum_{l=1}^Ny_{h(l)}B_{uv}^{(l)}
\right]_{u,v\in\Omega}.
\end{align}
Here the exponent of $z$ is $\sum_{u\in\Omega}\mathbf 1_{\{u\notin\Omega_{\mathbf a}\}}=\abs{\Omega\setminus\Omega_{\mathbf a}}$, the number of selected prefix modes.
If $\abs{\Omega}<j$, its product cannot contain all $j$ variables $x_1,\ldots,x_j$.
If $\abs{\Omega}>j$, two modes share a color and the product vanishes.
For $\abs{\Omega}=j$, the product equals $x_1\cdots x_j$ exactly when $g(\Omega)=[j]$, and otherwise vanishes.
Thus, for $0\leq p\leq j$, extracting $x_1\cdots x_jz^p$ retains precisely the sets with $\abs{\Omega}=j$, $g(\Omega)=[j]$, and $\abs{\Omega\setminus\Omega_{\mathbf a}}=p$.

Finally, for each such set $\Omega$, extracting $y_1\cdots y_j$ from the permanent gives the sums over $J$ and the two bijections, as established in Eq.~\eqref{eq:y-coefficient-permanent}.
Combining these coefficient extractions yields
\begin{align}
[x_1\cdots x_jy_1\cdots y_jz^p]\mathcal P_j
&={}
\sum_{\substack{\Omega\subseteq[M],\,\abs{\Omega}=j\\
\abs{\Omega\setminus\Omega_{\mathbf a}}=p,\ g(\Omega)=[j]}}
\sum_{\substack{J\subseteq[N],\,\abs{J}=j\\h(J)=[j]}}
\sum_{\sigma,\tau\in\operatorname{Bij}(J,\Omega)}
\prod_{l\in J}B_{\sigma(l),\tau(l)}^{(l)}   ,
\end{align}
thus proving Eq.~\eqref{eq:colored-coefficient-identity}.

\end{document}